\pdfoutput=1
\documentclass[a4paper,UKenglish,thm-restate,pdfa]{lipics-v2021-modified}

\usepackage{mathtools}					
\usepackage{xargs}                      
\usepackage[pdftex,dvipsnames]{xcolor}  
\usepackage{algorithm}
\usepackage{algpseudocode}

\def\RE{\ensuremath{\mathrm{RE}}}
\def\P{\ensuremath{\mathrm{P}}}
\def\QP{\ensuremath{\mathrm{QP}}}
\def\NP{\ensuremath{\mathrm{NP}}}
\def\NQP{\ensuremath{\mathrm{NQP}}}

\def\E{\ensuremath{\mathrm{E}}}

\def\EE{\ensuremath{\mathrm{EE}}}
\def\NE{\ensuremath{\mathrm{NE}}}
\def\NEE{\ensuremath{\mathrm{NEE}}}
\def\FP{\ensuremath{\mathrm{FP}}}
\def\UP{\ensuremath{\mathrm{UP}}}
\def\DisjNP{\ensuremath{\mathrm{DisjNP}}}
\def\DisjCoNP{\ensuremath{\mathrm{DisjCoNP}}}

\def\coNP{\ensuremath{\mathrm{coNP}}}
\def\coNQP{\ensuremath{\mathrm{coNQP}}}
\def\coNE{\ensuremath{\mathrm{coNE}}}
\def\coNEE{\ensuremath{\mathrm{coNEE}}}

\def\TFNP{\ensuremath{\mathrm{TFNP}}}

\def\SPARSE{\ensuremath{\mathrm{SPARSE}}}

\def\TAUT{\ensuremath{\mathrm{TAUT}}}
\def\SAT{\ensuremath{\mathrm{SAT}}}
\def\QBF{\ensuremath{\mathrm{QBF}}}

\def\TFNP{\ensuremath{\mathrm{TFNP}}}
\def\BPP{\ensuremath{\mathrm{BPP}}}

\def\ZPP{\ensuremath{\mathrm{ZPP}}}
\def\AM{\ensuremath{\mathrm{AM}}}

\def\PSPACE{\ensuremath{\mathrm{PSPACE}}}
\def\PH{\ensuremath{\mathrm{PH}}}

\def\NTIME{\ensuremath{\mathrm{NTIME}}}

\def\UTIME{\ensuremath{\mathrm{UTIME}}}

\def\N{\ensuremath{\mathrm{\mathbb{N}}}}

\def\R{\ensuremath{\mathrm{\mathbb{R}}}}

\def\impl{\ensuremath{\Rightarrow}}

\def\leqmp{\ensuremath{\leq_\mathrm{m}^\mathrm{p}}}
\def\psim{\ensuremath{\leq^\mathrm{p}}}

\DeclareMathOperator{\dom}{dom}
\DeclareMathOperator{\ran}{ran}

\DeclareMathOperator{\supp}{supp}

\DeclareMathOperator{\enc}{\mathrm{enc}}
\DeclarePairedDelimiter\ceil{\lceil}{\rceil}
\DeclarePairedDelimiter\floor{\lfloor}{\rfloor}
\def\sqsubsetneq{\mathrel{\sqsubseteq\kern-0.92em\raise-0.15em\hbox{\rotatebox{313}{\scalebox{1.1}[0.75]{\(\shortmid\)}}}\scalebox{0.3}[1]{\ }}}
\def\sqsupsetneq{\mathrel{\sqsupseteq\kern-0.92em\raise-0.15em\hbox{\rotatebox{313}{\scalebox{1.1}[0.75]{\(\shortmid\)}}}\scalebox{0.3}[1]{\ }}}
\newcommand{\card}[1]{\##1}
\def\runtime{\ensuremath{\mathrm{time}}}
\newcommand{\ls}[1]{\mathrm{ls}(#1)}

\newcommand{\xTrueThenElse}[2]{\if\x1{#1}\else{#2}\fi}

\newcommand{\lqq}{\lq\lq}   
\newcommand{\rqq}{\rq\rq}   
\newcommand{\elqq}{\lqq}    
\newcommand{\erqq}{\rqq}    

\newcommand{\eqq}[1]{\elqq #1\erqq}

\newcommand{\tn}[1]{\textnormal{#1}}    

\newcommand{\hhat}[1]{#1}

\algdef{SE}[SUBALG]{Indent}{EndIndent}{}{\algorithmicend\ }%
\algtext*{Indent}
\algtext*{EndIndent}
\hideLIPIcs
\title{Optimal Proof Systems for Complex Sets are Hard to Find}
\titlerunning{Optimal Proof Systems for Complex Sets are Hard to Find}
\author{Fabian Egidy}{University of Würzburg, Germany}{fabian.egidy@uni-wuerzburg.de}{https://orcid.org/0000-0001-8370-9717}{supported by the German Academic Scholarship Foundation.}
\author{Christian Glaßer}{University of Würzburg, Germany}{christian.glasser@uni-wuerzburg.de}{https://orcid.org/0009-0006-0572-8748}{}
\authorrunning{Fabian Egidy and Christian Glaßer}
\Copyright{Fabian Egidy and Christian Glaßer}
\acknowledgements{We thank anonymous reviewers for their helpful feedback and suggestions.}
\nolinenumbers

\begin{document}
\maketitle
\begin{abstract}
We provide the first evidence for the inherent difficulty of finding complex sets with optimal proof systems.
For this, we construct oracles $O_1$ and $O_2$ with the following properties,
where $\RE$ denotes the class of recursively enumerable sets and
$\NQP$ the class of sets accepted in non-deterministic quasi-polynomial time.
\medskip
\begin{itemize}
    \item[] $O_1$: No set in $\PSPACE \setminus \NP$ has optimal proof systems and $\PH$ is infinite
    \item[] $O_2$: No set in $\RE \setminus \NQP$ has optimal proof systems and $\NP \neq \coNP$
\end{itemize}
\medskip
\noindent
Oracle $O_2$ is the first relative to which
complex sets with optimal proof systems do not exist.
By oracle $O_1$, no relativizable proof can show that
there exist sets in $\PSPACE \setminus \NP$ with optimal proof systems, even when assuming an infinite $\PH$.
By oracle $O_2$, no relativizable proof can show that
there exist sets outside $\NQP$ with optimal proof systems, even when assuming $\NP \neq \coNP$.
This explains the difficulty of the following longstanding open questions raised by
Krajíček and Pudlák in 1989, Sadowski in 1997, Köbler and Messner in 1998, and Messner in 2000.
\medskip
\begin{itemize}
    \item[] Q1: Are there sets outside $\NP$ with optimal proof systems?
    \item[] Q2: Are there arbitrarily complex sets outside $\NP$ with optimal proof systems?
\end{itemize}
\medskip
\hspace{1.5em}
Moreover, relative to $O_2$,
there exist arbitrarily complex sets $L \notin \NQP$ having almost optimal algorithms,
but none of them has optimal proof systems.
This explains the difficulty of Messner's approach to translate almost optimal algorithms into optimal proof systems.

\end{abstract}

\section{Introduction}

The study of proof systems was initially motivated by the $\NP$ versus $\coNP$ question.
Cook and Reckhow \cite{cr79} define a proof system for a set $L$
as a polynomial-time computable function $f$ whose range is $L$.
They show that $\NP = \coNP$ if and only if there exists a proof system
with polynomial-size proofs for the set of propositional tautologies ($\TAUT$).
Hence, in order to show $\NP \neq \coNP$, it suffices to rule out polynomial-size proofs for increasingly powerful proof systems for $\TAUT$. This approach to the $\NP$ versus $\coNP$ question is known as the Cook-Reckhow program.

Krajíček and Pudlák \cite{kp89} capture the idea of the {\em most powerful proof system}
with their notions of optimal and p-optimal proof systems.
A proof system for a set $L$ is optimal if its proofs are at most polynomially longer than proofs in any other proof system for $L$. A proof system $f$ for a set $L$ is p-optimal if proofs from any other proof system for $L$ can be translated in polynomial-time into proofs in $f$. Krajíček and Pudlák \cite{kp89} raise the question of whether $\TAUT$ and other sets have (p-)optimal proof systems.
It is still one of the major open questions
in proof complexity 
with deep relations
to other fields like structural complexity theory and finite model theory (cf.~the optimality problem in Krajíček's book \cite{kra19a}).
For the Cook-Reckhow program,
the existence of an optimal proof system $f$ for $\TAUT$ means that
in order to show $\NP \neq \coNP$
it suffices to prove a superpolynomial lower bound for proofs in $f$.

It is known that
all nonempty sets in $\NP$ have optimal proof systems.
Outside $\NP$ this is different:
By Messner \cite{mes99,mes00},
all $\coNE$-hard sets and even all $\coNQP$-hard sets
have no optimal proof systems.
Hence there exist arbitrarily complex sets without optimal proof systems.
In contrast, we do not know a single set outside $\NP$ {\em with} optimal proof systems. The following are longstanding open questions.
\begin{spreadlines}{2ex}
\begin{align*}
&\text{\textbf{Q1}: {\em Are there sets outside $\NP$ with optimal proof systems?}}
\\
&\text{\textbf{Q2}: {\em Are there arbitrarily complex sets outside $\NP$ with optimal proof systems?}}
\end{align*}
\end{spreadlines}
Krajíček and Pudlák \cite{kp89} and Sadowski \cite{sad97} implicitely raise Q1
by asking whether $\TAUT$ and $\QBF$ (the set of true quantified Boolean formulas) have optimal proof systems.
Köbler and Messner \cite{km98} and Messner \cite{mes00} explicitly raise Q1 and Q2.
Over the past decades, these questions have been investigated across multiple contexts. 
We contribute to these questions by providing evidence for the inherent difficulty of finding optimal proof systems for complex sets. Before making this precise, we give an overview on optimal proof systems in complexity theory and summarize the current state of Q1 and Q2.

\subparagraph*{Sufficient conditions.}
Krajíček and Pudlák \cite{kp89} were the first to study sufficient conditions for the existence of optimal proof systems outside $\NP$. They prove that $\NE = \coNE$ implies the existence of optimal proof systems for $\TAUT$, and $\E = \NE$ implies the existence of p-optimal proof systems for $\TAUT$. Köbler, Messner, and Torán \cite{kmt03} improve this result to $\NEE = \coNEE$ for optimal and $\EE = \NEE$ for p-optimal proof systems for $\TAUT$. Köbler and Messner \cite{km98} and Messner \cite{mes00} show that the existence of (p-)optimal proof systems also follows from other unlikely assumptions. Chen, Flum, and Müller \cite{cfm14} show that if $\NP$ does not have measure $0$ in $\E$ (the Measure Hypothesis \cite{lut97}), then there exist sets in $\NE \setminus \NP$ with optimal proof systems.

Krajíček and Pudlák \cite{kp89} show that $\TAUT$ has p-optimal proof systems if and only if there exists an almost optimal algorithm (an algorithm that is optimal on positive instances, also called optimal acceptor) for $\TAUT$. Sadowski \cite{sad99} proves the same relation for $\SAT$ (the set of satisfiable Boolean formulas) and Messner \cite{mes99,mes00} generalizes both results to all paddable sets. Messner uses this and a related result for optimal proof systems to obtain that $\NP = \coNP$ (resp., $\P = \NP$) implies the existence of arbitrarily complex sets with optimal (resp., p-optimal) proof systems.

\subparagraph*{Necessary conditions.}
Optimal proof systems have shown to be closely related to promise classes. Initiated by Razborov \cite{raz94}, who shows that the existence of optimal proof systems for $\TAUT$ implies the existence of complete sets in the class $\DisjNP$ (the class of disjoint $\NP$-pairs),
further work reveals similar implications for other promise classes. Sadowski \cite{sad97} shows that the existence of p-optimal proof systems for $\QBF$ implies complete sets for $\NP \cap \coNP$. Köbler, Messner, and Torán \cite{kmt03} improve and generalize this result by showing that the existence of optimal and p-optimal proof systems for various sets of the polynomial-time hierarchy implies complete sets for promise classes such as $\UP, \NP \cap \coNP$, $\NP \cap \SPARSE$ and probabilistic classes such as $\ZPP$, $\BPP$, and $\AM$ \cite{kmt03}. Beyersdorff, Köbler, and Messner \cite{bkm09} and Pudlák \cite{pud17} reveal connections between proof systems and function classes. They show that p-optimal proof systems for $\SAT$ imply complete sets for $\TFNP$.

\subparagraph*{Other variants of proof systems.}
Modified definitions of {\em proof system} and {\em optimality}
yield positive answers to the corresponding questions Q1 and Q2.
Cook and Krajíček \cite{ck07} show that (p-)optimal proof systems for $\TAUT$ exist if we allow one bit of non-uniform advice. Beyersdorff, Köbler, and Müller \cite{bkm11} generalize this to arbitrary languages. In a heuristic setting, Hirsch, Itsykson, Monakhov, and Smal \cite{hims12} show that almost optimal algorithms exist, but it is open whether this translates into optimal heuristic proof systems. Pitassi and Santhanam \cite{ps10} show that there exists an optimal proof system for $\QBF$ under a weak notion of simulation.

\subparagraph*{Current state of Q1 and Q2.} In addition to the reasons mentioned so far, both questions are especially interesting, because in contrast to other open questions in complexity theory,
it is unclear whether a negative or positive answer is to be expected \cite{bs11,kra19a}. In particular, no widely believed structural assumption (like $\NP \neq \coNP$) is known to imply the (non)-existence of optimal proof systems \cite{hir10,bs11}.

So far, no evidence for the difficulty of proving a positive answer to Q1 or Q2 was known. But for a negative answer such evidence was known:
A negative answer to Q1 or Q2 implies $\NP \neq \coNP$, since Messner \cite{mes00} shows that $\NP = \coNP$ implies that outside $\NP$ there exist arbitrarily complex sets with optimal proof systems. There exist also several oracles showing that a negative answer to Q1 can not have a relativizable proof, even under rather likely assumptions (e.g., assuming 
that $\NP \cap \coNP$ has no complete sets \cite{dos20a} or that $\DisjCoNP$ has no complete pairs \cite{kha22}).

\subsection{Our Contribution}
We construct the following oracles $O_1$ and $O_2$, where we consider $O_2$ as our main result,
because of its significance and difficulty.
\begin{spreadlines}{2ex}
\begin{align*}
    &O_1\colon \text{ No set in $\PSPACE \setminus \NP$ has optimal proof systems and $\PH$ is infinite}
    \\
    &O_2\colon \text{ No set in $\RE \setminus \NQP$ has optimal proof systems, $\NP \neq \coNP$, and the}
    \\[-2 	ex]
    \vspace{-100em}&\phantom{O_2\colon} \text{ Measure Hypothesis fails}
\end{align*}
\end{spreadlines}
Recall that $\NQP$ stands for non-deterministic quasi-polynomial time\footnote{Note that the term \eqq{NQP} is not used consistently in the literature. We use it for non-deterministic quasi-polynomial time, other papers sometimes for quantum classes.}.
\subparagraph{Evidence for the difficulty of Q1 and Q2.}
Our oracles provide the first evidence for the difficulty of proving {\em positive} answers for these questions.
For Q2, a positive answer is difficult to prove,
since it requires non-relativizable methods (even if we assume $\NP \neq \coNP$).
Regarding Q1,
our oracles provide two pieces of information:
First, proving the existence of sets outside $\NQP$ with optimal proof systems
requires non-relativizable methods (even if we assume $\NP \neq \coNP$).
Second, proving the existence of sets in $\PSPACE \setminus \NP$ with optimal proof systems
requires non-relativizable methods (even if we assume an infinite $\PH$). In particular, relativizable methods can not prove that an infinite $\PH$ implies the existence of optimal proof systems for $\TAUT$ or $\QBF$, which was not known before.
Hence, statements slightly stronger than a positive answer to Q1
require non-relativizable methods (even under widely believed assumptions).
This hints at the difficulty of proving a positive answer for this question.

Overall, our oracles explain why the search for sets outside $\NP$
with optimal proof systems has so far been unsuccessful.
Moreover, they show the inherent difficulty of finding optimal proof systems for complex sets outside $\NP$, 
even when assuming $\NP \neq \coNP$.

    \subparagraph*{Oracles that widely exclude optimal proof systems.}
    Previous oracles achieve that specific sets $A$ have no optimal proof systems
    (e.g., $\coNP$-complete $A$).
    Although this implies no optimal proof systems for all $B$ with $A \leqmp B$,
    it does not imply that {\em all} sets in a class like $\PSPACE \setminus \NP$
    have no optimal proof systems.
    Hence, these oracles only exclude optimal proof systems for {\em some sets} of a class.
    In contrast, our oracles are the first that exclude optimal proof systems
    for {\em whole classes}, namely $\PSPACE \setminus \NP$ and $\RE \setminus \NQP$.
    Therefore, these oracles provide insights
    on the concept of optimal proof systems as a whole,
    whereas previous oracles tended to provide insights
    on optimal proof systems for specific languages.

As explained before, it is unclear
which answer to Q2 is to be expected.
Nevertheless, it seems that researchers
tend to favor a negative answer \cite{bs11,kra19a}.
$O_2$ is the first oracle showing
that such a situation is possible in principle.

\subparagraph*{Oracle where almost optimal algorithms do not imply optimal proof systems.}
Messner \cite{mes00} aims to construct 
complex sets with optimal proof systems.
His approach consists of proving the following statements:
\begin{enumerate}
    \item\label{enum:opt-algo1} There exist arbitrarily complex sets $L$ having almost optimal algorithms.
    \item\label{enum:opt-algo2} $\tn{$L$ has an almost optimal algorithm} \;\impl\; \tn{$L$ has an optimal proof system}$.
\end{enumerate}
While Statement \ref{enum:opt-algo1} holds in general \cite{mes00},
Statement \ref{enum:opt-algo2} is still an open question \cite{bhks15, cf14}.
Our oracle $O_2$ provides the first evidence for the difficulty of proving statement 2, i.e., it cannot be proved by relativizable methods.
More precisely, as Statement \ref{enum:opt-algo1} relativizes, the following situation holds relative to our oracle:
There exist arbitrarily complex sets $L \notin \NQP$ having almost optimal algorithms,
but none of them has optimal proof systems.
So in this relativized world, Statement \ref{enum:opt-algo2} holds {\em nowhere} outside $\NQP$.

\subsection{Key Challenges and Conceptual Ideas}

We give a first impression of our oracle constructions.
Proof sketches and detailed proofs can be found in
sections~\ref{sec:thin-oracle} and \ref{sec:thick-oracle}.

\subparagraph{Towards $\text{O}_\text{1}$.}
The construction consists of two parts.

Part 1: We construct a sparse oracle without optimal proof systems for sets in $\PSPACE \setminus \NP$.
All $f \in \FP$ and $L \in \PSPACE$ are treated as follows.
First try to add sparse information to the oracle such that $f$ is not a proof system for $L$.
If this is possible, we are done with $f$ and $L$.
Otherwise, we can rely on $f$ being always a proof system for $L$.

Now try the following.
For infinitely many $y \in L$ with superpoly-size $f$-proofs
add poly-size proofs (not noticed by $f$) to the oracle.
If this is possible, then $f$ is not an optimal proof system for $L$.
Otherwise, no matter which short proofs we add to the oracle,
$f$ notices this and develops a short proof.
Hence $f$ can {\em improve itself}\/ by guessing short proofs and
assuming that they are contained in the oracle.
Note that $f$ cannot make mistakes (i.e., produce elements outside $L$),
since by assumption, $f$ is always a proof system for $L$. This yields a polynomially bounded proof system for $L$ showing $L \in \NP$.

Part 2: Do the same, but now with an infinite $\PH$.
Start with an oracle where $\PH$ is infinite and repeat the construction from part 1.
This adds sparse information to the oracle and
prevents optimal proof systems in $\PSPACE \setminus \NP$.
By a relativized version of the generalized Karp-Lipton theorem \cite{kl80,ls86,bbs86,bcs95},
the sparse information we added cannot collapse the $\PH$.

\subparagraph*{Towards $\text{O}_\text{2}$.}
 The main difficulties in preventing optimal proof systems for sets in $\RE \setminus \NQP$ arise from dealing with arbitrarily complex computations. These computations may depend on large sections of the oracle and may change their behavior at any moment of the construction, making sound, complete, and robust encodings of information a particular challenge. Additionally, in this general setting we cannot exploit any special property of our machines (e.g., runtime bounds or promise conditions) that are typically used to rule out certain behaviour.
This generality makes our construction difficult.

One crucial point is the right amount of information that is encoded in the oracle:
sparse enough to keep the number of dependencies that encodings induce small,
but dense enough to allow $\NQP$-machines to use these encodings. Another challenge is the encoding of short facts via short code words whose associated proof is long.
We follow the idea that a computation of runtime $t$ can not query all words of length $n > \log t$,
which leads to code words that are independent of certain computations,
but significantly shorter than their runtimes.
At the beginning of section \ref{sec:thick-oracle} we describe in more detail how these code words
help to tackle the main challenges.

\section{Basic Definitions and Notations}\label{sec:prelim}
\subparagraph*{Words and sets.}
Let $\Sigma \coloneqq \{0,1\}$ be the default alphabet and $\Sigma^*$ be the set of finite words over $\Sigma$. For $a \in \Sigma$ and $L \subseteq \Sigma^*$ define $aL \coloneqq \{aw \mid w \in L\}$. We denote the length of a word $w \in \Sigma^*$ by $|w|$. The empty word has length $0$ and is denoted by $\varepsilon$. The $(i+1)$-th letter of a word $w$ for $0 \leq i < |w|$ is denoted by $w(i)$, i.e., $w=w(0)w(1)\cdots w(|w|-1)$. If $v$ is a (strict) prefix of $w$, we write $v \sqsubseteq w$ ($v \sqsubsetneq w$) or $w \sqsupseteq v$ ($w \sqsupsetneq v$).

The set of all natural numbers is denoted by $\N$, of all positive natural numbers by $\N^+$, and of all natural numbers starting at $i \in \N$ by $\N_i$. For $a,b \in \N$, we define $a\N+b \coloneqq \{a \cdot n + b \mid n \in \N\}$. We denote the set of real numbers by $\R$. We write the empty set as $\emptyset$.  The cardinality of a set $A$ is denoted by $\card{A}$. For a set $A \subseteq \Sigma^*$ and a number $n \in \N$, we define $A^{\leq n} \coloneqq \{w \in A \mid |w| \leq n\}$. We define this notation for $<$, $=$, $\geq$, and $>$ analogous. For a clearer notation we use $\Sigma ^{\leq n}$ for ${\Sigma^*}^{\leq n}$ (analogous for other inequality symbols) and $\Sigma^n$ for ${\Sigma^*}^{=n}$. The operators $\cup$, $\cap$, and $\setminus$ denote the union, intersection and set-difference. We denote the marked union of sets $A,B$ by $A \oplus B \coloneqq 2A \cup (2B+1)$. We say that two sets $A$ and $B$ agree on $C$, if for all $x \in C$ it holds that $x \in A$ if and only if $x \in B$. A set $A$ is $\mathcal{C}$-immune for a complexity class $\mathcal{C}$ if no infinite subset of $A$ is in $\mathcal{C}$.

\subparagraph*{Functions.}
The domain and range of a function $f$ are denoted by $\dom (f)$ and $\ran (f)$. The support of $f$ is defined as $\supp (f) \coloneqq \{x \in \dom (f) \mid f(x) \not = 0\}$. A function $f'$ is an extension of a function $f$, denoted by $f \sqsubseteq f'$ or $f' \sqsupseteq f$, if $\dom(f) \subseteq \dom(f')$ and $f(x) = f'(x)$ for all $x \in \dom (f)$. If $x \notin \dom (f)$, then $f \cup \{x \mapsto y\}$ denotes the extension $f'$ of $f$ such that $f'(z) = f(z)$ for $z \not = x$ and $f'(x)=y$. We denote the composition of two functions $f$ and $g$ by $f \circ g$.

We identify $\Sigma^*$ with $\N$ through the polynomial time computable and invertible bijection $\enc \colon \Sigma^* \to \N$ defined as $\enc(w) = \sum _{i<|w|}(1+w(i))2^i$. This is a variant of the dyadic representation. Thus, we can treat words from $\Sigma^*$ as numbers from $\N$ and vice versa, which allows us to use notations, relations and operations of words for numbers and vice versa (e.g., we can define the length of a number by this). Note that $x < y$ if and only if $x < _{\text{lex}} y$ for $x,y \in \Sigma ^*$, where $<_{\text{lex}}$ is the quasi-lexicographic ordering relation for words. Expressions like $0^i$ and $1^i$ are ambiguous, because they can be interpreted as words from $\Sigma ^*$ or numbers. We resolve this ambiguity by the context or explicitly naming the interpretation. 

The logarithm to the base $2$ is denoted by $\log \colon \R^+ \to \R$. The iterated logarithm is defined as 
\[\log^* n \coloneqq \begin{cases}
0 & \text{if } n \leq 1\\
1+\log^*(\log n) & \text{if } n > 1
\end{cases}\]
Furthermore, we define polynomial functions $p_k \colon \N \to \N$ for $k \in \N^+$ by $p_k(n) \coloneqq n^k$. The function computing the maximum element of a finite subset of $\N$ is denoted by $\max$. We use the operators $\ceil{\cdot}$ and $\floor{\cdot}$ to round values up and down to integers.

\subparagraph*{Machines.}
We use the default model of a Turing machine in both the deterministic and non-deterministic variant. The language decided by a Turing machine $M$ is denoted by $L(M)$. We use Turing transducer to compute functions. For a Turing transducer $F$ we
write $F(x)=y$ when on input $x$ the transducer outputs $y$. Hence, a Turing transducer $F$ computes a function and we may denote ''the function computed by $F$'' by $F$ itself. We denote the number of steps the longest path of a computation $M(x)$ takes by $\runtime (M(x))$. 

\subparagraph*{Complexity classes.}
$\P$, $\NP$, and $\coNP$ denote standard complexity classes \cite{pap81}. We denote the set of polynomial-time computable functions by $\FP$, the polynomial hierarchy by $\PH$, and the set of recursively enumerable sets by $\RE$. A set $S \subseteq \Sigma^*$ is sparse if the there is a polynomial $p$ such that $\card{S^{\leq n}} \leq p(n)$ for all $n \in \N$. We denote the class of all sparse sets as $\SPARSE$. For some total function $t$, we define the complexity class $\NTIME (t)$ as the class of languages that can be decided by a non-deterministic Turing machine $M$ such that $\runtime(M(x)) \leq t(|x|)$ for all $x \in \Sigma^*$. The class $\UTIME (t)$ is defined the same, but on every input, $M$ has at most one accepting path. We define the class of sets decidable in non-deterministic quasi-polynomial time as $\NQP \coloneqq \bigcup_{c \in \N} \NTIME(2^{(\log n)^c})$.

\subparagraph*{Proof systems.}
We use proof systems for sets defined by Cook and Reckhow \cite{cr79}. They define a function $f \in \FP$ to be a proof system for $\ran (f)$. Furthermore:
\begin{itemize}
\item A proof system $g$ is (p-)simulated by a proof system $f$, denoted by $g \leq f$ (resp., $g \psim f$), if there exists a total function $\pi$ (resp., $\pi \in \FP$) and a polynomial $p$ such that $|\pi(x)| \leq p(|x|)$ and $f(\pi(x)) = g(x)$ for all $x \in \Sigma^*$. In this context the function $\pi$ is called simulation function. Note that $g \psim f$ implies $g \leq f$.
\item A proof system $f$ is (p-)optimal for $\ran(f)$, if $g \leq f$ (resp., $g \psim f$) for all $g \in \FP$ with $\ran (g) = \ran (f)$.
\end{itemize}

If $y \in \ran (f)$ for some proof system $f$, then we denote the shortest $f$-proof for $y$ by $\hhat{y}_f$, i.e., $f(\hhat{y}_f)=y$ and for all $x < \hhat{y}_f$ it holds that $f(x) \not = y$. If $y \notin \ran (f)$, then $\hhat{y}_f$ is not defined. 

\subparagraph*{Oracle specific definitions and notations.} We relativize the concept of Turing machines and Turing transducers by giving them access to a write-only oracle tape. We relativize complexity classes, proof systems, machines, and (p-)simulation by defining them over machines with oracle access, i.e., whenever a Turing machine or Turing transducer is part of a definition, we replace them by an oracle Turing machine or an oracle Turing transducer. We indicate the access to some oracle $O$ in the superscript of the mentioned concepts, i.e., $\mathcal{C}^O$ for a complexity class $\mathcal{C}$ and $M^O$ for a Turing machine or Turing transducer $M$. We sometimes omit the oracles in the superscripts, e.g., when sketching ideas in order to convey intuition, but never in actual proofs. We also transfer all notations to the respective oracle concepts. 

For any oracle $O$ let $\{F_k^O\}_{k \in \N^+}$ be a standard enumeration of polynomial time oracle Turing transducers, also called $\FP^O$-machines, where $F_k^O$ has running time exactly $p_k$. Similarly, for any oracle $O$ let $\{M_k^O\}_{k \in \N^+}$ denote a standard enumeration of polynomial space oracle Turing machines, also called $\PSPACE^O$-machines, where $M_k^O$ has space bound exactly $p_k$. We utilize the oracle Turing machine model presented by Simon \cite{sim77}, in which the oracle tape also underlies the polynomial space bound restriction. 
For a deterministic machine $F$, an oracle $O$ and an input $x$, we define $Q(F^O(x))$ as the set of oracle queries of the computation $F^O(x)$. 

A word $w \in \Sigma ^*$ can be interpreted as a set $\{i \in \N \mid w(i) = 1\}$\footnote{Note that we interpret this set as a subset of $\Sigma^*$ via $\enc$.} and subsequently as a partial oracle, which is defined for all words up to $|w|-1$ (more precisely up to $\enc^{-1}(|w|-1)$). So $|w|$ is the first word that $w$ is not defined for, i.e., $|w| \in w1$ and $|w| \notin w0$. 

A computation $F^w(x)$ is called definite, if $w$ is defined for all words of length $\leq \runtime(F^w(x))$.
This means that $F^w(x)$ is definite, if $||w|| > \runtime(F^w(x))$, because $|w|$ denotes the first undefined word of $w$ and the length of the word $|w|$ (i.e., $||w||$) has to be greater than $\runtime(F^w(x))$. Additionally, a space-bounded computation $M^w(x)$ with bound $p(|x|)$ is definite, if $w$ is defined for all words of length $p(|x|)$. A computation definitely accepts (resp., rejects), if it is definite and accepts (resp., rejects). We may combine computations to more complex expressions and call those definite, if every individual computation is definite, e.g., ``$F^O(x) \in L(M^O)$ is definite'' means $F^O(x)$ outputs some $y$ and is definite, and $M^O(y)$ is definite.

During the oracle construction we often use words from $\Sigma ^*$ to denote partially defined oracles. In particular, oracle queries for undefined words of a partial oracle are answered negatively.
\begin{observation}\label{obs:oracle-questions}
Let $x \in \Sigma^*$ and $C \subseteq \Sigma^*$. For an oracle $w$ and a Turing transducer $F$, it holds that
\[C \cap Q(F^w(x)) = \emptyset \Rightarrow F^w(x) = F^{w\cup C}(x).\]
\end{observation}
\begin{observation}\label{obs:definite}
If $F^w(x)$ is definite, then $F^w(x) = F^v(x)$ for any extension $v \sqsupseteq w$.
\end{observation}

\section{Oracle with an Infinite PH and No Set in PSPACE\textbackslash NP Has Optimal Proof Systems}\label{sec:thin-oracle}
\subparagraph{Idea of the construction.} A relativized version of the generalized Karp-Lipton theorem \cite{kl80,ls86,bbs86,bcs95} shows that if $\PH$ is infinite relative to an oracle $A$, it is also infinite relative to $A \oplus B$ if $B \in \SPARSE$. Hence, it suffices to construct for any $A$ a sparse $B$ such that relative to $A \oplus B$ no set in $\PSPACE \setminus \NP$ has optimal proof systems. For simplicity, here we restrict to the task of preventing a proof system $f$ to be optimal for $L \in \PSPACE$. We try to achieve this as follows:
\begin{enumerate}[{Case} 1]
\item\label{case:idea-1} There is a sparse extension of the oracle such that $\ran(f) \not \subseteq L$ (hence, $f$ is not a proof system for $L$).
\item\label{case:idea-2} For all sparse extensions of the oracle, $\ran(f) \subseteq L$.
\begin{enumerate}[(a)]
\item\label{case:idea-2a} It is possible to sparsely encode infinitely many facts $y \in L$ by code words $c(y)$ of polynomial lengths while $f$ has only superpolynomially long proofs for the facts $y$. Here $f$ is not optimal (the gap between the length of the $c(y)$ and the shortest $f$-proofs for $y$ is superpolynomial).
\item\label{case:idea-2b} Such an encoding of infinitely many facts is not possible. If $\ran(f) \neq L$, then $f$ is not a proof system for $L$ and we are done. Otherwise, $\ran(f) = L$, which we exploit as follows to show $L \in \NP$.
\end{enumerate}
\end{enumerate}
Assume $\ran(f) = L$. Since we are in Case \ref{case:idea-2b}, whenever we encode some fact $y \in L$ into the oracle, $f$ develops a polynomially long proof for $y$.
This gives the impression that by encoding all these facts we can achieve polynomially bounded proofs for $f$ and hence $L \in \NP$. But this adds too many words to the oracle such that it is not sparse.

Surprisingly, we can show $L \in \NP$ \emph{without} encoding anything! For this, we define a proof system $H$, which on inputs $(x,c(y))$ simulates $f(x)$ while assuming that there is a code word $c(y)$ inside the oracle. If $f(x)=y$ relative to the modified oracle, $H$ outputs $y$. Otherwise $H$ outputs some fixed element from $L$.

$H$ has polynomially bounded proofs, since it simulates the computation of $f$ with an assumed code word in the oracle (and $f$ has polynomially bounded proofs in this situation, since we are in Case \ref{case:idea-2b}). 
Moreover, $H$ can not output elements outside $L$, because otherwise $f$ does the same relative to some sparse oracle, which shows that Case \ref{case:idea-1} holds, a contradiction. 
Hence, $H$ is a polynomially bounded proof system for $L$.

\subsection{Definitions for the Oracle Construction}

For the rest of section \ref{sec:thin-oracle}, let $A \subseteq 2\N$ be arbitrary. We construct a sparse oracle $B \subseteq 2\N + 1$, such that all sets in $\PSPACE^{B \cup A} \setminus \NP^{B \cup A}$ do not have optimal proof systems relative to $B \cup A$.

In the following, we discuss computations relative to $w \cup A$ for a partially defined oracle $w$. Note that since $A$ is already fully defined, a computation relative to $w \cup A$ is definite, if $w$ is defined for sufficiently long words.

\subparagraph*{Code words.} For machines $F_i$, we want to encode facts $y \in \ran(F_i)$ into the oracle. We also want to pad the encodings to length at least $p_k(|y|)$ such that $M_k(y)$ can not query such a padded code word. 
\begin{definition}[Code words]\label{def:codewords1}
Let $i,k \in \N^+$ and $y \in \Sigma^*$. We define a code word as $c(i,k,y) \coloneqq 00^i1^k0^{p_k(|y|)}1y$.
\end{definition}
\begin{observation}\label{obs:codeword-property}
For fixed $i,k$, the function $y \mapsto c(i,k,y)$ is polynomial time computable and polynomial time invertible with respect to $|y|$.
\end{observation}
Note that because of the leading $0$, the interpretation of code words as numbers via $\enc$ is always an odd number. So the set of all code words is a subset of $2\N + 1$.

\subparagraph*{Valid oracles.}
During the construction of the oracle, we successively add requirements that we maintain. These are specified by a function $r \colon \N^+ \times \N^+ \times \N \mapsto \N$, called \textit{requirement function}. An oracle $w$ is $r$-valid for a requirement function $r$, if it satisfies the following requirements (let $i,k \in \N^+$):
\begin{enumerate}[V1]
\item\label{thin-V1} $\varepsilon \notin w$ and for all $n \in \N^+$ and $a \in \{0,1\}$, $\card{(w \cap a\Sigma^{n-1})} \leq 1-a$.
\\
(Meaning: $w \subseteq 2\N + 1$ and $w$ is sparse.)
\item\label{thin-V2} If $r(i,k,0)=0$, then there is some $x \in \Sigma^*$ such that $F_i^{w \cup A}(x) \notin L(M_k^{w \cup A})$ is definite.
\\
(Meaning: $F_i$ is no proof system for $L(M_k)$ relative to $w \cup A$ and all its extensions.)
\item\label{thin-V3} If $r(i,k,0)>0$, then for $y \in \Sigma^{>r(i,k,0)}$ it holds that $c(i,k,y) \in w$ implies $y \in L(M_k^{w \cup A})$.
\\
(Meaning: Code words encode sufficiently long facts correctly.)
\item\label{thin-V4} If $r(i,k,j) = 0$ for $j \in \N^+$, then there is some $y \in \Sigma^{>r(i,k,0)}$ such that $c(i,k,y) \in w$ and for all $x \in \Sigma^{\leq p_j(|y|)}$ it holds that $F_i^{w \cup A}(x) \neq y$ is definite. 
\\
(Meaning: There is a gap of certain size between the length of a code word $c(i,k,y) \in w$ and the shortest $F_i^{w \cup A}$-proof for $y$.)
\end{enumerate}
We will prove that V\ref{thin-V1} to V\ref{thin-V4} are satisfied for various oracles and requirement functions. To prevent confusion, we use the notation V\ref{thin-V1}($w$), V\ref{thin-V2}($w,r$), V\ref{thin-V3}($w,r$), and V\ref{thin-V4}($w,r$) to clearly state the referred oracle $w$ and requirement function $r$. 

The following lemma shows how to extend $r$-valid oracles such that they remain $r$-valid.
\begin{lemma}\label{lemma:extension}
If $w$ is an $r$-valid oracle, then
\begin{itemize}
\item $w0$ is $r$-valid.
\item $w1$ is $r$-valid if V\ref{thin-V1}($w1$) and V\ref{thin-V3}($w1,r$) are satisfied.
\end{itemize}
\end{lemma}
\begin{proof}
Let $b \in \{0,1\}$ be arbitrary. We show that both $w0$ and $w1$ are $r$-valid, when assuming that V\ref{thin-V1}($w1$) and V\ref{thin-V3}($w1,r$) are satisfied.
\smallskip
\\
To V\ref{thin-V1}($w0$): Since no words are added to $w$ and V\ref{thin-V1}($w$) is satisfied, also V1($w0$) is satisfied. 
\smallskip
\\
To V\ref{thin-V2}($wb,r$): Let $i,k$ be arbitrary such that $r(i,k,0)=0$. By V\ref{thin-V2}($w,r$), there is some $x$ such that $F_i^{w \cup A}(x) \notin L(M_k^{w \cup A})$ is definite. Since the computations $F_i^{w \cup A}(x)=y$ and $M_k^{w \cup A}(y)$ are definite, they do not recognize any difference between $w$ and $wb$, because they can not query words of length $||w||$. Hence, $F_i^{wb \cup A}(x) \notin L(M_k^{wb \cup A})$. 
\smallskip
\\
To V\ref{thin-V3}($w0,r$): Let $i,k,y$ be arbitrary such that $r(i,k,0) > 0$ and $c(i,k,y) \in w0$ with $y \in \Sigma^{>r(i,k,0)}$. Then also $c(i,k,y) \in w$, because the oracle $w0$ contains the exact same elements as $w$. By V\ref{thin-V3}($w,r$), from $c(i,k,y) \in w$ follows $y \in L(M_k^{w \cup A})$. The computation $M_k^{w \cup A}(y)$ is definite, because $w$ is defined for words of length $\geq |c(i,k,y)| > p_k(|y|)$. Hence, $M_k^{wb \cup A}(y)$ also accepts.
\smallskip
\\
To V\ref{thin-V4}($wb,r$): Let $i,k,j$ be arbitrary such that $r(i,k,j) = 0$ with $j \in \N^+$. By V\ref{thin-V4}($w,r$), there is some $y \in \Sigma^{>r(i,k,0)}$ such that $c(i,k,y) \in w$ and for all $x \in \Sigma^{\leq p_j(|y|)}$ it holds that $F_i^{w \cup A}(x) \neq y$ is definite. From $wb \sqsupseteq w$ it holds that $c(i,k,y) \in wb$ and it holds that $F_i^{wb \cup A}(x) \neq y$ for all $x \in \Sigma^{\leq p_j(|y|)}$, because $F_i^{w \cup A}(x) \not = y$ is definite.
\end{proof}

\subsection{Oracle Construction}
The oracle construction will take care of the following set of tasks $\{\tau _{i,k,j} \mid i,k \in \N^+, j\in \N\}$. Let $\mathcal{T}$ be an enumeration of these tasks with the property that $\tau _{i,k,j}$ appears earlier than $\tau _{i,k,j+1}$. In each step we treat the smallest task in the order specified by $\mathcal{T}$, and after treating a task we remove it
from $\mathcal{T}$. 

In step $s=0$ we define $w_0 \coloneqq \varepsilon$ and $r_0$ as the nowhere defined requirement function. 
In step $s > 0$ we define $w_s \sqsupsetneq w_{s-1}$ and $r_s \sqsupsetneq r_{s-1}$ such that $w_s$ is $r_s$-valid by treating the earliest task $\tau $ in $\mathcal{T}$ and removing it from $\mathcal{T}$. We do this according to the following procedure:

\begin{itemize}
\item[] {\bf Task} $\tau_{i,k,j}$: Let $r^0 \coloneqq r_{s-1} \cup \{(i,k,j) \mapsto 0\}$. If there exists an $r^0$-valid partial oracle $v \sqsupsetneq w_{s-1}$, then assign $r_s \coloneqq r^0$ and $w_s \coloneqq v$. 

Otherwise, let $w_s \coloneqq w_{s-1}0$ and $r_s \coloneqq r_{s-1} \cup \{(i,k,j) \mapsto ||w_{s}||+1\}$. 
\medskip
\\
(Meaning: First, try to rule out $F_i$ as a proof system for $L(M_k)$ or try to create a gap of certain size between the length of an encoded code word for a fact and its shortest $F_i$-proof (i.e., $r_s(i,k,0) = 0$ or $r_s(i,k,j) = 0$ for all $j \in \N^+$). If this is not possible (i.e., $r_s(i,k,0) > 0$ and $r_s(i,k,j) > 0$ for some $j \in \N^+$), intuitively, $F_i$ is inherently a proof system for $L(M_k)$ and $F_i$ always has short proofs for a fact when encoding the fact by a code word. As described at the start of section \ref{sec:thin-oracle}, we will show $\ran(F_i) \in \NP$ without encoding anything.)
\end{itemize}

From this construction, we define the oracle with our desired properties as $B \coloneqq \bigcup _{s \in \N} w_s$. Also define $r \coloneqq \bigcup_{s\in \N} r_s$, which helps to show these properties for $B$. 

The following claims show that $B$ and $r$ are both well-defined (Claim \ref{claim:strict-extension-each-step}), every task is treated because $r$ is a total function (Claim \ref{claim:r-total}), $w_s$ is $r_s$-valid for all $s \in \N$ (Claim \ref{claim:ws-is-rs-valid}), $B$ is $r$-valid (Claim \ref{claim:O-is-rs-valid}), and every intermediate extension $w$ from $w_s$ to $B$ (i.e., $w_s \sqsubseteq w \sqsubseteq B$) is $r_s$-valid (Claim \ref{claim:valid-downward}).

\begin{claim}\label{claim:strict-extension-each-step}
For every $s \in \N$ it holds that $w_s \sqsubsetneq w_{s+1}$ and $r_s \sqsubsetneq r_{s+1}$.
\end{claim}
\begin{claimproof}
Let $s \in \N$ be arbitrary. When treating any task $\tau_{i,k,j}$ at step $s+1$, $w_s$ gets extended by at least one bit and $r_s$ gets defined for $(i,k,j)$. The enumeration $\mathcal{T}$ does not have duplicate tasks and whenever a task was treated, it got removed from $\mathcal{T}$. Hence, $r_s$ can not be defined for $(i,k,j)$ before treating the task $\tau_{i,k,j}$. So, $w_s$ and $r_s$ get both extended.
\end{claimproof}
\begin{claim}\label{claim:r-total}
$\dom(r) = \N^+ \times \N^+ \times \N$.
\end{claim}
\begin{claimproof}
Every task $\tau_{i,k,j} \in \mathcal{T}$ gets treated once and $r$ gets defined for $(i,k,j)$ when treating $\tau_{i,k,j}$.
\end{claimproof}
\begin{claim}\label{claim:ws-is-rs-valid}
For every $s \in \N$ it holds that $w_s$ is $r_s$-valid.
\end{claim}
\begin{claimproof}
Observe that V\ref{thin-V1}($w_0$) is satisfied and that V\ref{thin-V2}($w_0,r_0$) to V\ref{thin-V4}($w_0,r_0$) do not require anything. So $w_0$ is $r_0$-valid. 

Let $s \in \N^+$ be arbitrary. By induction hypothesis, $w_{s-1}$ is $r_{s-1}$-valid. Let $\tau_{i,k,j}$ be the task treated at step $s$.

Consider $j=0$. If $\tau_{i,k,0}$ is treated such that $r_s(i,k,0)=0$, then $w_s$ is chosen to be $r_s$-valid. Otherwise, by Lemma \ref{lemma:extension}, $w_s = w_{s-1}0$ is $r_{s-1}$-valid. Furthermore, $r_s(i,k,0) = ||w_{s}||+1$, which affects the requirement V\ref{thin-V3}($w_s,r_s$). But for $(i,k,0)$, this requirement can only be violated if a code word $c(i,k,y)$ of length at least $r(i,k,0) > ||w_s||$ is inside $w_s$, which is impossible. Hence $w_s$ is $r_s$-valid.

Consider $j > 0$. If $\tau_{i,k,j}$ is treated such that $r_s(i,k,j) = 0$, then $w_s$ is chosen to be $r_s$-valid. Otherwise, by Lemma \ref{lemma:extension}, $w_s = w_{s-1}0$ is $r_{s-1}$-valid. Since $r_s(i,k,j) > 0$ does not affect any requirement V\ref{thin-V1} to V\ref{thin-V4}, $w_s$ is also $r_s$-valid.
\end{claimproof}

\begin{claim}\label{claim:O-is-rs-valid}
$B$ is $r$-valid.
\end{claim}
\begin{claimproof}
Intuitively, any violation of V\ref{thin-V1} to V\ref{thin-V4} for $B$ and $r$ leads to a contradiction that $w_s$ is $r_s$-valid for all $s \in \N$.
\medskip
\\
Assume that V\ref{thin-V1}($B$) is violated. Either $\varepsilon \in B$, which implies $\varepsilon \in w_1$ and contradicts that V\ref{thin-V1}($w_1$) is satisfied. Or there are $n \in \N^+$ and $a \in \{0,1\}$ such that $\card{(B \cap a\Sigma^{n-1})} > 1-a$. Let $s$ be the first step such that $||w_s|| > n$. Since $B$ is an extension of $w_s$ and $w_s$ is defined for all words of length $n$, $w_s^{=n} = B^{=n}$. Hence, 
\[\card{w_s \cap a\Sigma^{n-1}} = \card{B \cap a\Sigma^{n-1}} > 1-a\] 
and V\ref{thin-V1}($w_s$) is not satisfied, a contradiction to Claim \ref{claim:ws-is-rs-valid} that $w_s$ is $r_s$-valid.
\medskip
\\
Let $i,k$ with $r(i,k,0)=0$ be arbitrary. Let $s$ be the step that defines $r_s(i,k,0)=0$. By Claim~\ref{claim:ws-is-rs-valid}, $w_s$ is $r_s$-valid. By V\ref{thin-V2}($w_s,r_s$), there is some $x$ such that $F_i^{w_s \cup A}(x) \notin L(M_k^{w_s \cup A})$ is definite. Since $B$ is an extension of $w_s$, also $F_i^{B \cup A}(x) \notin L(M_k^{B \cup A})$. Thus, V\ref{thin-V2}($B,r$) is satisfied for $r(i,k,0)$.
\medskip
\\
Assume that V\ref{thin-V3}($B,r$) is violated. Then there are $i,k$ such that $r(i,k,0)>0$ and there is some $y \in \Sigma^{>r(i,k,0)}$ such that $c(i,k,y) \in B$, but $y \notin L(M_k^{B \cup A})$. Let $s$ be the first step such that $c(i,k,y) \in w_s$ and $(i,k,0) \in \dom(r_s)$. Then $y \notin L(M_k^{w_s \cup A})$, because $w_s$ is defined for all words of length $\leq p_k(|y|)$ and hence $w_s^{\smash{\leq p_k(|y|)}} = B^{\leq p_k(|y|)}$. Thus, V\ref{thin-V3}($w_s,r_s$) is not satisfied, a contradiction to Claim \ref{claim:ws-is-rs-valid} that $w_s$ is $r_s$-valid.
\medskip
\\
Let $i,k,j$ with $r(i,k,j) = 0$ and $j \in \N^+$ be arbitrary. Let $s$ be the first step such that $r_s(i,k,j) = 0$. From Claim \ref{claim:ws-is-rs-valid}, $w_s$ is $r_s$-valid. By V\ref{thin-V4}($w_s,r_s$), there is some $y \in \Sigma^{>r(i,k,0)}$ such that $c(i,k,y) \in w_s $ and for all $x \in \Sigma^{\leq p_j(|y|)}$ it holds that $F_i^{w_s \cup A}(x) \not = y$ is definite. Since $B$ is an extension of $w_s$, we have that $c(i,k,y) \in B$ and the definite computations $F_i^{w_s \cup A}(x) \neq y$ stay the same relative to $B$. Thus, \ref{thin-V4}($B,r$) is satisfied for $r(i,k,j)$.
\end{claimproof}

\begin{claim}\label{claim:valid-downward}
Let $s \in \N$ and $w$ be an oracle. If $w_s \sqsubseteq w \sqsubseteq B$, then $w$ is $r_s$-valid.
\end{claim}
\begin{claimproof}
We want to invoke Lemma \ref{lemma:extension} repeatedly to extend $w_s$ to $w$ such that $w$ is $r_s$-valid. For this, we have to make sure that the necessary extensions by $1$ are possible. Let $v$ be $r_s$-valid such that $w_s \sqsubseteq v \sqsubseteq v1 \sqsubseteq w$, i.e., we have to extend $v$ by $1$ in order to get to $w$. We show that V\ref{thin-V1}($v1$) and V\ref{thin-V3}($v1,r_s$) are satisfied. Then Lemma \ref{lemma:extension} can be invoked when extending $w_s$ to $w$ whenever an extension by $1$ is necessary, showing that $w$ is $r_s$-valid. 

If V\ref{thin-V1}($v1$) is not satisfied, then also V\ref{thin-V1}($B$) would not be satisfied, because $B \sqsupseteq w \sqsupseteq v1$, a contradiction to Claim \ref{claim:O-is-rs-valid}.

If V\ref{thin-V3}($v1,r_s$) is not satisfied, then there are $i,k$ such that $r_s(i,k,0) > 0$ and there is some $y \in \Sigma^{>r_s(i,k,0)}$ such that $c(i,k,y) \in v1$ and $M_k^{v1 \cup A}(y)$ rejects. Since $|c(i,k,y)| > p_k(|y|)$, $M_k(y)$ also rejects relative to all extensions of $v1 \cup A$. Since $B \sqsupseteq v1$ and $r \sqsupseteq r_s$, we have that $r(i,k,0) > 0$ and $c(i,k,y) \in B$, but $y \notin L(M_k^{B \cup A})$, implying that \ref{thin-V3}($B,r$) is not satisfied, contradicting Claim~\ref{claim:O-is-rs-valid}. 
\end{claimproof}

\begin{theorem}\label{thm:sparse-oracle}
Let $A \subseteq 2\N$ be arbitrary. There is a sparse oracle $B \subseteq 2\N+1$ such that all $L \in \PSPACE^{B \cup A} \setminus \NP^{B \cup A}$ have no optimal proof system relative to $B \cup A$.
\end{theorem}
\begin{proof}
Choose $B$ and $r$ as constructed in the oracle construction. By V\ref{thin-V1}($B$) and the definition of $\enc$, we get that $B \subseteq 2\N+1$ is a sparse oracle. Let $L \in \PSPACE^{B \cup A} \setminus \NP^{B \cup A}$ be arbitrary and let $k \in \N^+$ such that $L(M_k^{B \cup A})=L$.

\textbf{Assume} that $L$ has an optimal proof system $f \in \FP^{B \cup A}$. Since $\{F_i^{B \cup A}\}_{i \in \N^+}$ is a standard enumeration of $\FP^{B \cup A}$-machines, there is some $i \in \N^+$ such that $F_i^{B \cup A}(x) = f(x)$ for all $x \in \Sigma^*$. We prove the theorem by making a case distinction based on the definition of $r$ and derive a contradiction in each case, i.e.,
\begin{enumerate}[C1]
\item\label{case:thm-sparse-oracle-1} if $r(i,k,j) > 0$ for some $j \in \N^+$, then $L \in \NP^{B \cup A}$, contradicting $L \notin \NP^{B \cup A}$.
\item\label{case:thm-sparse-oracle-2} if $r(i,k,j) = 0$ for all $j \in \N^+$, then $F_i^{B \cup A}$ is not optimal for $L$.
\end{enumerate}
This shows that the assumption that $L$ has an optimal proof system must be false.
\medskip

It must hold that $r(i,k,0) > 0$. Otherwise, V2($B,r$) gives that there is some $x$ such that $F_i^{B \cup A}(x) \notin L(M_k^{B \cup A})$, which is a contradiction to the assumption that $F_i^{B \cup A}$ is an optimal proof system for $L(M_k^{B \cup A})$. So $r(i,k,0) > 0$ holds for the rest of this proof.

We define a proof system $H^{B \cup A}$, which is used to derive a contradiction for the case ``$r(i,k,j) > 0$ for some $j \in \N^+$''. Note that $i$ and $k$ are fixed. After that, we explain the idea of the definition.
\begin{algorithm}
\caption{$H^{B \cup A}$}
\begin{algorithmic}[1]
\State \textbf{Input:} $x \in \Sigma^*$ \label{alg:line1}
\If{$x = 0x'$} \Return $F_i^{B \cup A}(x')$ \label{alg:line2}
\EndIf\label{alg:line3}
\If{$x = 1c(i,k,y)x'$ and $|c(i,k,y)| > r(i,k,0)$ for suitable $x',y$}\label{alg:line4}
\State let $Q \coloneqq B^{<|c(i,k,y)|} \cup \{c(i,k,y)\}$\label{alg:line5}
\State simulate $F_i^{Q \cup A}(x')$\label{alg:line6}
\If{the simulation returns $y$} \Return $y$\label{alg:line7}
\EndIf\label{alg:line8}
\EndIf\label{alg:line9}
\State \Return the smallest element of $L$\label{alg:line10}
\end{algorithmic}
\end{algorithm}

\textbf{Idea:} Intuitively, either $F_i$ produces short proofs for $y$ whenever a suitable code word for $y$ is in the oracle (i.e., $r(i,k,j) > 0$ for some $j \in \N^+$), or $F_i$ is not significantly affected by code words (i.e., $r(i,k,j) = 0$ for all $j \in \N^+$). The proof system $H$ is useful for arguing in the former case. On specific inputs, $H$ checks whether $F_i$ outputs $y$ when assuming that a suitable code word for $y$ is in the oracle. If $F_i$ does so, $H$ also outputs $y$. Since $F_i$ obtains short proofs for any word whenever a suitable code word for it exists in the oracle, $H$ has short proofs for any word even \emph{without} suitable code words appearing inside the oracle. This shows that $L \in \NP^{B \cup A}$ in this case. Importantly, in Claim \ref{claim:Hi-is-proofsystem} we will show that $H$ can not output a word outside of $L$, because otherwise we can exploit this to obtain an oracle where $F_i$ outputs a word outside $L$ and hence $r(i,k,0)=0$ would have been possible.
\medskip

The following claims prove the ideas from the paragraph above. We start by showing that the oracle $Q$ defined in line \ref{alg:line5} is in some sense valid. We use this to show that $H^{B \cup A}$ is a polynomially bounded proof system for $L$ (Claims \ref{claim:Hi-is-proofsystem} and \ref{claim:polybounded-proofs}).

\begin{claim}\label{claim:extension}
Let $i,k \in \N^+$, $s \in \N$ and $y \in \Sigma^*$ such that $|c(i,k,y)| > ||w_s||$. If
\begin{enumerate}[(i)]
\item\label{claim:extension-i} $r_s(i,k,0)$ is not defined, or
\item\label{claim:extension-ii} $y \in L(M_k^{B \cup A})$,
\end{enumerate} 
then $w \coloneqq B^{<|c(i,k,y)|} \cup \{c(i,k,y)\}$, interpreted as the partial oracle defined for all words of length $|c(i,k,y)|$, is an $r_s$-valid extension of $w_s$. 
\end{claim}
\begin{claimproof}
Since $M_k^{B \cup A}(y)$ can only query words of length $p_k(|y|) < |c(i,k,y)|$ and $B^{<|c(i,k,y)|} = w^{<|c(i,k,y)|}$, the computation does not recognize any difference between $B$ and $w$. Hence, 
\begin{align}\label{claim:extension-tag1}
y \in L(M_k^{w \cup A}) \Longleftrightarrow y \in L(M_k^{B \cup A}).
\end{align}

Consider the partial oracle $v$ which is defined like $B$ up to words of length $< |c(i,k,y)|$. Then $w_s \sqsubseteq v \sqsubseteq B$ and by Claim \ref{claim:valid-downward}, $v$ is $r_{s}$-valid. We extend $v$ bitwise to $w$ such that $w$ is defined for words up to length $|c(i,k,y)|$ and only add $c(i,k,y)$ as additional word. To get that $w$ is $r_{s}$-valid via Lemma \ref{lemma:extension}, we have to make sure that we are allowed to add $c(i,k,y)$ (the extensions with $0$ are possible without any problems). Let $v'$ denote the partially extended $r_s$-valid oracle at the point of the single extension by $1$ (i.e., when $c(i,k,y)$ is added next).

V\ref{thin-V1}($v'1$) is satisfied, because $c(i,k,y)$ starts with $0$ and is the only word of its length in $v'1$. For V\ref{thin-V3}($v'1,r_s$), observe that when $c(i',k',y') \in v'$, then $y' \in L(M_{k'}^{v' \cup A})$ by V\ref{thin-V3}($v',r_s$). Since 
\[p_{k'}(|y'|) < |c(i',k',y')| < |c(i,k,y)| = ||v'||,\] 
it also holds that $y' \in L(M_{k'}^{v'1 \cup A})$. So V\ref{thin-V3}($v'1,r_s$) can only be violated for $c(i,k,y)$. If (\ref{claim:extension-i}) holds, then $r_s(i,k,0)$ is not defined and V\ref{thin-V3}($v'1,r_s$) can not be violated by adding $c(i,k,y)$. Otherwise, (\ref{claim:extension-ii}) holds and V\ref{thin-V3}($v'1,r_s$) is satisfied, because
\[y \in L(M_k^{B \cup A}) \overset{\text{(\ref{claim:extension-tag1})}}{\Longrightarrow} y \in L(M_k^{w \cup A}) = L(M_k^{v'1 \cup A}).\] 
This shows that $v'1$ is $r_s$-valid, from which follows that $w$ is $r_s$-valid.
\end{claimproof}
\begin{claim}\label{claim:Hi-is-proofsystem}
$H^{B \cup A}$ is a proof system for $L$.
\end{claim}
\begin{claimproof}
It holds that $H^{B \cup A} \in \FP^{B \cup A}$, because testing whether $x$ is of the form $x = 1c(i,k,y)x'$ can be done in polynomial time by Observation \ref{obs:codeword-property}, we only simulate $F_i^{Q \cup A}$ on words smaller than the input, and oracle queries to $Q \cup A$ can be simulated using $B \cup A$ and $\{c(i,k,y)\}$. Also $\ran(H^{B \cup A}) \supseteq L$ by line \ref{alg:line2} and the assumption that $F_i^{B \cup A}$ is a proof system for $L$.

Suppose there is some $x$ such that $H^{B \cup A}(x) = y \notin L$. Then this can only happen when $x$ is of the form described in line \ref{alg:line4}, i.e., $x = 1c(i,k,y)x'$, $|c(i,k,y)| > r(i,k,0)$, and $F_i^{Q \cup A}(x')=y$ with $Q \coloneqq B^{<|c(i,k,y)|} \cup \{c(i,k,y)\}$. Then, also $F_i^{Q \cup A}(x') = y \notin L = L(M_k^{B \cup A})$, because $H^{B \cup A}(x)=F_i^{Q \cup A}(x')$. Since $M_k^{B \cup A}(y)$ can only query words of length $p_k(|y|) < |c(i,k,y)|$ and $B^{<|c(i,k,y)|} = Q^{<|c(i,k,y)|}$, the computation does not recognize any difference between $B$ and $Q$. Hence, $M_k^{Q \cup A}(y)$ also rejects, which gives 
\begin{align}\label{claim:Hi-is-proofsystem-tag}
F_i^{Q \cup A}(x') = y \notin L(M_k^{Q \cup A}).
\end{align}
Consider the step $s$ where the task $\tau_{i,k,0}$ was treated. Note that 
\[||w_{s-1}|| + 1 \leq ||w_{s}|| + 1 = r(i,k,0) < |c(i,k,y)|.\]
Since $r_{s-1}(i,k,0)$ is not defined, by Claim \ref{claim:extension}, there is an $r_{s-1}$-valid extension $v \coloneqq Q$ of $w_{s-1}$, interpreted as partial oracle defined for all words of length $|c(i,k,y)|$.
We can further extend $v$ to $v'$ such that $v'$ is defined for all words of length $\max\{|c(i,k,y)|,p_i(|x'|)\}$ by appending zeros. Lemma \ref{lemma:extension} gives that $v'$ is $r_{s-1}$-valid.

In fact, $v'$ is even $r_{s-1} \cup \{(i,k,0) \mapsto 0\}$-valid, because from (\ref{claim:Hi-is-proofsystem-tag}) and $v' = Q$ it follows that 
\begin{align*}
F_i^{v' \cup A}(x') = F_i^{Q \cup A}(x') \notin L(M_k^{Q \cup A}) = L(M_k^{v' \cup A}).
\end{align*}
These computations are also definite, because $v'$ is defined for all words of length $p_i(|x'|)$ and $p_k(|y|) < |c(i,k,y)|$. Hence, at step $s$, the oracle construction does not choose $w_s$ and $r_s$ with $r_s(i,k,0) > 0$, but instead $v'$ as $r_{s-1} \cup \{(i,k,0) \mapsto 0\}$-valid extension of $w_{s-1}$, a contradiction to the construction of $r$. 

So for all $x$ it holds that $H^{B \cup A}(x) \in L$. Hence, $\ran(H^{B \cup A}) \subseteq L$ and $H^{B \cup A}$ is a proof system for $L$.
\end{claimproof}
\emph{Assume} case C\ref{case:thm-sparse-oracle-1}, i.e., $r(i,k,j) > 0$ for some $j \in \N^+$. Under this assumption, we prove that $H^{B \cup A}$ has polynomially bounded proofs and hence $L \in \NP^{B \cup A}$, which is a contradiction.
\begin{claim}\label{claim:polybounded-proofs}
$H^{B \cup A}$ has polynomially bounded proofs.
\end{claim}
\begin{claimproof}
Let $s$ be the step that treats the task $\tau_{i,k,j}$. Consider any $y \in L = L(M_k^{B \cup A})$ with $|y| > r(i,k,j)$. Only finitely many words in $L$ do not satisfy this. Note that 
\[||w_{s-1}|| + 1 < ||w_s||+1 = r(i,k,j) < |c(i,k,y)| \text{ and } r(i,k,0) < r(i,k,j).\] By Claim \ref{claim:extension}, there is an $r_{s-1}$-valid extension $v=B^{<|c(i,k,y)|} \cup \{c(i,k,y)\}$ of $w_{s-1}$, interpreted as partial oracle defined for all words of length $|c(i,k,y)|$. We can further extend $v$ to $w$ such that $w$ is defined for all words of length $\max\{|c(i,k,y)|,p_i(p_j(|y|))\}$ by appending zeros. Lemma \ref{lemma:extension} gives that $w$ is $r_{s-1}$-valid.

Recall that $\hhat{y}_{\smash{F_i^{w \cup A}}}$ denotes the shortest $F_i^{w \cup A}$-proof for $y$. Suppose that $\hhat{y}_{\smash{F_i^{w \cup A}}}$ does not exist or $|\hhat{y}_{F_i^{w \cup A}}| > p_j(|y|)$. Then $F_i^{w \cup A}(x) \neq y$ is definite for all $x \in \Sigma^{\leq p_j(|y|)}$, because $w$ is defined for all words of length $\leq p_i(p_j(|y|))$. Together with $c(i,k,y) \in w$ and $|y| > r(i,k,j) > r(i,k,0)$, we get that $w$ is $r_{s-1} \cup \{(i,k,j) \mapsto 0\}$-valid. Hence, at step $s$, the oracle construction does not choose $w_s$ and $r_s$ with $r_s(i,k,j)>0$ as assumed, but instead $w$ as $r_{s-1} \cup \{(i,k,j) \mapsto 0\}$-valid extension of $w_{s-1}$, a contradiction to the supposition. 

Hence, $|\hhat{y}_{F_i^{w \cup A}}| \leq p_j(|y|)$. It follows that 
\[H^{B \cup A}(1c(i,k,y)\hhat{y}_{F_i^{w \cup A}}) = F_i^{Q \cup A}(\hhat{y}_{F_i^{w \cup A}}) = F_i^{w \cup A}(\hhat{y}_{F_i^{w \cup A}}) = y,\] 
showing that $H^{B \cup A}$ has a $1+ |c(i,k,y)| + p_j(|y|)$ bounded proof for all but finitely many $y \in L$. Since by Observation \ref{obs:codeword-property}, $|c(i,k,y)|$ is polynomially bounded with respect to $|y|$ for fixed $i$ and $k$, $H^{B \cup A}$ has polynomially bounded proofs.
\end{claimproof}
Cook and Reckhow \cite[Prop.~1.4]{cr79} show relativizably that if a set has polynomially bounded proof systems, it is contained in $\NP$. By Claims \ref{claim:Hi-is-proofsystem} and \ref{claim:polybounded-proofs}, $H^{B \cup A}$ is a polynomially bounded proof system for $L$. Hence, $L \in \NP^{B \cup A}$, which is a contradiction to the choice of $L$ as $L \notin \NP^{B \cup A}$. This rules out case C\ref{case:thm-sparse-oracle-1} for the definition of $r$.
\medskip
\\
\emph{Assume} case C\ref{case:thm-sparse-oracle-2}, i.e., $r(i,k,j) = 0$ for all $j \in \N^+$. Under this assumption, we define a proof system $G^{B \cup A}$ for $L$ (Claim \ref{claim:Gi-is-proofsystem}), which is not simulated by $F_i^{B \cup A}$ (Claim \ref{claim:Fi-simulationfail-Gi}). Recall that $i$ and $k$ are fixed.
\[G^{B \cup A}(x) = 
\begin{cases}
F_i^{B \cup A}(x') & \text{if } x=1x'\\
y & \text{if } x = 0y \land |y| > r(i,k,0) \land c(i,k,y) \in B \cup A\\
\text{smallest element of $L$} & \text{else}
\end{cases}
\]
\begin{claim}\label{claim:Gi-is-proofsystem}
$G^{B \cup A}$ is a proof system for $L$.
\end{claim}
\begin{claimproof}
It holds that $G^{B \cup A} \in \FP^{B \cup A}$, because we only simulate $F_i$ on words smaller than the input and computing $c(i,k,y)$ from $y$ can be done in polynomial time by Observation \ref{obs:codeword-property}. Also, $\ran(G^{B \cup A}) \supseteq L$ by line 1 of the definition and the assumption that $F_i^{B \cup A}$ is a proof system for $L$. Line 2 never outputs some $y \notin L$, because $A \subseteq 2\N$ contains no code word and by V\ref{thin-V3}($B,r$), whenever $c(i,k,y) \in B$ with $|y| > r(i,k,0)$, it holds that $y \in L(M_k^{B \cup A})=L$. Hence, $\ran(G^{B \cup A}) \subseteq L$.
\end{claimproof}
\begin{claim}\label{claim:Fi-simulationfail-Gi}
$F_i^{B \cup A}$ does not simulate $G^{B \cup A}$.
\end{claim}
\begin{claimproof}
Suppose that $F_i^{B \cup A}$ simulates $G^{B \cup A}$ via $\pi$ whose output-length is restricted by the polynomial $p_{j-1}$ for $j \in \N_2$. By V\ref{thin-V4}($B,r$) and $r(i,k,j)=0$ (case C\ref{case:thm-sparse-oracle-2}), there is some $y \in \Sigma^{>r(i,k,0)}$ such that $c(i,k,y) \in B$ and for all $x \in \Sigma^{\leq p_j(|y|)}$ it holds that $F_i^{B \cup A}(x) \neq y$, i.e., $|\hhat{y}_{F_i^{B \cup A}}| > p_j(|y|)$. We have that $G^{B \cup A}(0y) = y$, so $G^{B \cup A}$ has a proof for $y$ of length $|y|+1$. The simulation function $\pi$ translates the $G^{B \cup A}$-proof $0y$ to some $F_i^{B \cup A}$-proof $x$ of length $\geq |\hhat{y}_{F_i^{B \cup A}}| >p_j(|y|)$. This is not possible, because the $F_i^{B \cup A}$-proof is too long for $\pi$:
\[|\pi(0y)| \leq p_{j-1}(|y|+1) \leq p_j(|y|) < |\hhat{y}_{F_i^{B \cup A}}| \leq |x|.\]
This is a contradiction to the supposition that $F_i^{B \cup A}$ simulates $G^{B \cup A}$ via $\pi$.
\end{claimproof}
The Claims \ref{claim:Gi-is-proofsystem} and \ref{claim:Fi-simulationfail-Gi} show that $G^{B \cup A}$ is a proof system for $L$ which is not simulated by $F_i^{B \cup A}$. This contradicts the assumption that $F_i^{B \cup A}$ is optimal. Hence, we can rule out case C\ref{case:thm-sparse-oracle-2} for the definition of $r$.

In total, we get a contradiction to our main assumption that $F_i^{B \cup A}$ is an optimal proof system for $L$, because the requirement function $r$ can be neither of the type $r(i,k,j) = 0$ for all $j \in \N^+$, nor of the type $r(i,k,j) > 0$ for some $j \in \N^+$. But $r$ is a total function with domain $\N^+ \times \N^+ \times \N$ (Claim \ref{claim:r-total}) and as such one of these two cases holds. Hence, $L$ can not have an optimal proof system.
\end{proof}
\begin{corollary}\label{cor:sparse-oracle}
There is a sparse oracle $B$ such that all $L \in \PSPACE^B \setminus \NP^B$ have no optimal proof system relative to $B$.
\end{corollary}
\begin{proof}
Follows by Theorem \ref{thm:sparse-oracle} for $A=\emptyset$.
\end{proof}
\subsection{Infinite Polynomial Hierarchy and Sparse Oracles}

Long and Selman \cite{ls86} and independently Balcázar, Book, and Schöning \cite{bbs86} show the following generalization of the Karp-Lipton theorem \cite{kl80}.
\begin{theorem}[\cite{bbs86,ls86}]\label{thm:bbs86}
If there exists a sparse set $S$ such that the polynomial-time hierarchy relative
to $S$ collapses, then the (unrelativized) polynomial-time hierarchy collapses.
\end{theorem}
As Bovet, Crescenzi, and Silvestri \cite{bcs95} point out, an analysis of the corresponding proofs shows that Theorem \ref{thm:bbs86} relativizes, which leads to the following result.
\begin{theorem}\label{thm:ph-trick}
Let $A \subseteq \N$ and $B \in \SPARSE$. If\/ $\PH^{A \oplus B}$ collapses, then $\PH^A$ collapses.
\end{theorem}
\begin{corollary}\label{cor:ph-trick}
Let $A \subseteq 2\N$ and $B \subseteq 2\N+1$ such that $B \in \SPARSE$. If\/ $\PH^A$ is infinite, then $\PH^{A \cup B}$ is infinite.
\end{corollary}
Corollary \ref{cor:ph-trick} shows how we can combine an oracle where the $\PH$ is infinite with a sparse oracle, such that the $\PH$ is infinite relative to the combined oracle. 
We use Corollary \ref{cor:ph-trick} to get the main result of section \ref{sec:thin-oracle}.
\begin{theorem}\label{thm:sparse-and-ph-oracle}
There exists an oracle $O_1$ such that $\PH^{O_1}$ is infinite and all $L \in \PSPACE^{O_1} \setminus \NP^{O_1}$ have no optimal proof system relative to $O_1$.
\end{theorem}
\begin{proof}
Let $A \subseteq 2\N$ be an oracle such that $\PH^A$ is infinite (exists by Yao \cite{yao85}), let $B$ be according to Theorem \ref{thm:sparse-oracle} invoked with $A$, and let $O_1 \coloneqq A \cup B$. Then $B \subseteq 2\N+1$ is sparse and all $L \in \PSPACE^{O_1} \setminus \NP^{O_1}$ have no optimal proof system relative to $O_1$. By Corollary \ref{cor:ph-trick}, $\PH^{O_1}$ is infinite.
\end{proof}

\section{Oracle without Optimal Proof Systems Outside NQP}\label{sec:thick-oracle}
\subparagraph{Idea of the construction.} The main approach is similar to the previous oracle. For an arbitrary proof system $f$, we distinguish two cases:
\begin{enumerate}[{Case} 1]
\item\label{idea2:case-1} It is possible to encode infinitely many facts $y \in \ran(f)$ by code words $c(y)$ while $f$ has only superpolynomial long proofs in $|c(y)|$ for the facts $y$. Here $f$ is not optimal (the gap between the length of $c(y)$ and the shortest $f$-proof for $y$ is superpolynomial).
\item\label{idea2:case-2} Such an encoding of infinitely many facts is not possible. Then encode almost all facts $y \in \ran(f)$ by quasi-polynomial long code words into the oracle to achieve $\ran(f) \in \NQP$ (note that contrary to case \ref{case:idea-2b} in section \ref{sec:thin-oracle}, $\ran(f) \in \NQP$ does not follow without encoding).
\end{enumerate}
The second case involves several challenges, because $\ran(f)$ has arbitrary complexity, but we have to encode $\ran(f)$ via quasi-polynomial long code words. For simplicity, let us focus on two main challenges:
\begin{enumerate}[{Challenge} 1]
\item\label{challenge:1} We have to encode every fact $y \in \ran(f)$. How can we extend the oracle while maintaining existing encodings and their dependencies?
\item\label{challenge:2} How do we handle short facts with long $f$-proofs emerging late in the oracle construction?
\end{enumerate}
We solve Challenge \ref{challenge:1} by encoding facts with a suitable padding. The padding for a fact $y$ must be large enough such that the number of dependencies (i.e., reserved words) the oracle has when encoding $y$ is less than the number of available paddings. At the same time, the padding must be small enough such that the resulting code word is at most quasi-polynomially longer than $y$. To keep the number of dependencies of a code word small, we must bound the length of $f$-proofs for $y$ in the length of code words for $y$. This is difficult, since we are dealing with arbitrarily complex sets. Nevertheless, it is possible when applying the principle which solves Challenge \ref{challenge:2}.

We solve Challenge \ref{challenge:2} by avoiding this situation based on the following principle: let there be a long $f$-proof $x$ for a short fact $y$. Then we encode $y$ by a code word that is superpolynomially shorter than $x$ (but not much shorter). Since we are in Case \ref{idea2:case-2}, $f$ notices this and develops a shorter proof, because otherwise we make progress towards Case \ref{idea2:case-1}. By repeating this, we can force $f$ to have increasingly shorter proofs. Finally, $f$ ends up having a rather short proof, which can only depend on short code words. Hence, when defining the membership of short code words in the oracle construction, we check if we can force $f$-proofs for $y$ to be rather short. If we can not, then Challenge \ref{challenge:2} can not appear for $f$ and $y$. 

\bigskip

The following oracle construction is quite long and requires the maintenance of many properties during the construction. Therefore, we will additionally explain the definitions and proof ideas on an intuitive level.

\subsection{Preparation of the Oracle Construction}

\subparagraph{Basic estimates.}
This section collects basic arithmetic and algebraic results regarding logarithmic and exponential functions.
\begin{observation}\label{obs:algebra}
\begin{romanenumerate}
\item $\floor{\log(\log^* 3)} = 1$ and $\floor{\log(\log^* 17)}=2$.\label{obs:algebra-i}
\item $\log^*(i+1) \leq \log^*(i) + 1$ and $\log(i+1) \leq \log (i) + 1$ for $i \geq 1$, $i \in \R$.\label{obs:algebra-ii}
\item $\log i \leq i$ for $i \in \R^+$.\label{obs:algebra-iii}
\item $2i \leq 2^{i-2}$ for $i \geq 6$, $i \in \R$.\label{obs:algebra-iv}
\item $i/2 \geq \log i$ for $i \geq 4$, $i \in \R$.\label{obs:algebra-v}
\end{romanenumerate}
\end{observation}

\begin{claim}\label{claim:algebra}
\begin{romanenumerate}
\item $\sum_{n=3}^{19} \floor{\log (\log^* n )} = 20$ and $\sum_{n=3}^{i} \floor{\log (\log^* n )} \geq i+1$ for $i \in \N_{19}$.\label{claim:algebra-ii}
\item $2^i \geq \floor{\log(\log^* i)} \cdot i + 4i$ for $i \in \N_5$.\label{claim:algebra-iii}
\end{romanenumerate}
\end{claim}
\begin{claimproof}
To (\ref{claim:algebra-ii}): Base case for $i = 19$: 
\[\sum_{n=3}^{19} \floor{\log (\log^* n )} = \underbrace{\floor{\log(2)} \cdot 2}_{n=3,4} + \underbrace{\floor{\log(3)} \cdot 12}_{n=5,\dots,16} + \underbrace{\floor{\log(4)} \cdot 3}_{n=17,\dots,19} = 2 + 12 + 6 = 20\]
Induction step: Let $i \in \N_{19}$ be arbitrary.
\begin{align*}
\sum_{n=3}^{i+1} \floor{\log (\log^* n )} &= \floor{\log (\log^* (i+1) )} + \sum_{n=3}^{i} \floor{\log (\log^* n )}\\ 
&\overset{\text{IH}}{\geq} \floor{\log (\log^* (i+1) )} + i + 1 \overset{\text{\ref{obs:algebra}.\ref{obs:algebra-i}}}{\geq} i + 2
\end{align*}
\medskip
\\
To (\ref{claim:algebra-iii}): Base case for $i = 5$ and $i= 6$:
\[\floor{\log(\log^* 5)} \cdot 5 + 4 \cdot 5 = 25 < 32 = 2^5 \text{ and } \floor{\log(\log^* 6)} \cdot 6 + 4 \cdot 6 = 30 < 64 = 2^6 .\]

Induction step: Let $i \in \N_6$ be arbitrary.
\begin{align*}
\floor{\log(\log^* (i+1))} \cdot (i+1) + 4 (i+1) &\overset{\text{\ref{obs:algebra}.\ref{obs:algebra-ii}}}{\leq} (\floor{\log(\log^* i)}+1) \cdot (i+1) + 4 (i+1)\\
 &= \floor{\log(\log^* i)} \cdot i + \floor{\log(\log^* i)} + i + 4i + 5\\
 &= \floor{\log(\log^* i)} \cdot i + 4i + \floor{\log(\log^* i)} +5 +i\\
 &\leq 2(\floor{\log(\log^* i)} \cdot i + 4i)\\
 &\overset{\text{IH}}{\leq} 2 \cdot 2^i = 2^{i+1}\tag*{\claimqedhere}
\end{align*}
\end{claimproof}

\subparagraph{Stages.}
We use the following notations to describe the lengths of words in our oracle. 
\begin{definition}[Stages]
\ 
\begin{itemize}
\item $t \colon \R^+ \to \R^+,~n \mapsto n^{\log^* n}$. 
\item $h \colon \R^+ \to \R^+,~n \mapsto 2^{(\log n)^{40}}$. 
\item $t_0 \coloneqq 16$, $t_i \coloneqq t(t_{i-1})$ for $i \in \N^+$. 
\item $s_0 \coloneqq 2t_0$, $s_i \coloneqq t_i + t_{i + \floor{\log i}}$ for $i \in \N^+$. 
\item $T \coloneqq \{t_i \mid i \in \N\}$.
\item $S \coloneqq \{s_i \mid i \in \N\}$.
\end{itemize}
\end{definition}

The oracle contains only words of lengths from $S$. Hence, we consider it divided into \emph{stages} $s \in S$ consisting of the words of length $s$. We denote the largest stage a partial oracle $w$ is defined for as $\ls{w}$. The function $h$ will be used as a runtime bound, namely to express that optimal proof systems can only exist for sets in $\NTIME (h^c)$ for some constant $c$.

As described in the sketch of the oracle construction, we have to choose a padding length of the right size. The following claim will be the basis that $t_i$ and $s_i$ have suitable properties to define the correct padding length. We give a brief intuition behind the most important properties: 
\begin{enumerate}[left=0.]
\item[\ref{claim:stages2}.\ref{claim:stages2-vi}] We construct the oracle such that facts of length $\approx t_{i-2}$ from sets with optimal proof systems are encoded by code words of length $\approx t_{i+\floor{\log i}}$. Hence, we can decide such facts by searching for code words of length $\approx t_{i+\floor{\log i}}$. Since $h$ grows faster than the distance between these two values (Claim \ref{claim:stages2}.\ref{claim:stages2-vi}), we can decide these facts in approximately non-deterministic time $h$. 
\item[\ref{claim:stages2}.\ref{claim:stages2-vii}] This property shows that the distance between stages grows faster than any polynomial. Thus, if a proof system has a proof of length $\mathrm{poly}(s_i)$, then in general, it can not pose queries of length $s_{i+1}$. This makes the encodings at stage $s_i$ robust to changes at stage $s_{i+1}$.
\item[\ref{claim:stages2}.\ref{claim:stages2-viii}] The left-most term of the inequality will be the number of available paddings to encode a fact by a code word of length $s_{i+1}$. The right hand side will be an upper bound of the number of dependencies the code words of length $\leq s_{i+\floor{\log i}}$ in the oracle can have. In total, the inequality shows that we can choose suitable paddings to encode facts by code words of length $s_{i+1}$ without disturbing existing code words of length $\leq s_{i+\floor{\log i}}$.

\end{enumerate}
\begin{claim}\label{claim:stages2} 
\begin{romanenumerate}
\item\label{claim:stages2-i} $S,T \subseteq \N^+$ and $S,T \in \P$.
\item\label{claim:stages2-ii} $t_i \geq 2^{2^{i \cdot \floor{\log(\log^* i)} / 4}}$ for $i \in \N_{19}$.
\item\label{claim:stages2-iii} $t_i \leq 2^{2^{i \cdot \floor{\log(\log^* i)} + 4i}}$ for $i \in \N_{5}$.
\item\label{claim:stages2-iv} $s_i \leq 2 t_{i+\floor{\log i}}$ for $i \in \N^+$.
\item\label{claim:stages2-v} $h$ is strictly monotone increasing.
\item\label{claim:stages2-vi} $h(t_{i-2}) \geq t_{i+\floor{\log i}}$ for $i \in \N_{21}$.
\item\label{claim:stages2-vii} For every $k \in \N^+$ there is some $i_0 \in \N$ such that $p_k(s_i) < s_{i+1}$ for all $i \geq i_0$.
\item\label{claim:stages2-viii} There is some $\hat{m} \in \N$ such that for all $i \in \N_{\hat{m}}$ the following holds:
\[2^{t_{i+1+\floor{\log(i+1)}}} - s_{i+4+\floor{\log i}} > 2^{t_{i+\floor{\log i}}+1} \cdot s_{i+2+\floor{\log i}} \cdot (i + \floor{\log i})\]
\end{romanenumerate}
\end{claim}
\begin{claimproof}
To (\ref{claim:stages2-i}): $S,T \in \N^+$ follows from the fact that $t_0 = 16$ and $\ran (\log^*) \subseteq \N$. Since the iterated logarithm can be computed in polynomial time, $t \in \FP$. Hence, when deciding whether $n \in T$, we can compute $t_0,t_1,t_2, \dots$ until we arrive at $t_i \geq n$. Observe that $|t_j| < |t_{j+1}|$ for $j \in \N$. So we have to compute $t(x)$ at most $|n|$ times for inputs $x$ of length $\leq |n|$, which can be done in polynomial time in $|n|$. So $T \in \P$. From this follows $S \in \P$.
\medskip
\\
To (\ref{claim:stages2-ii}): Let $i \in \N_{19}$. First, we show that 
\begin{align}\label{claim:stages2-ii-eq1}
t_i \geq 2^{2^{\sum_{n=3}^i \floor{\log (\log^* n )}}}
\end{align} 
holds by induction. Second, we argue that $\sum_{n=3}^i \floor{\log (\log^* n )}$ has $i/2$ elements of size at least $\floor{\log (\log^* i)}/2$. Hence, $\sum_{n=3}^i \floor{\log (\log^* n )} \geq i \cdot \floor{\log (\log^* i)}/4$ and the claimed statement follows.
\medskip
\\
Base case $i=19$: By Claim \ref{claim:algebra}.\ref{claim:algebra-ii}, it holds that $\sum_{n=3}^{19} \floor{\log (\log^* n )} = 20$. So, we have to show that $t_{19} \geq 2^{2^{20}}$. Since $\log^*(4) = 2$ and $\log^*$ is monotonically increasing, it holds that
\begin{align*}\label{claim:stages2-ii-al1}
t(n) = n^{\log^* n} \ge n^2 \text{ for } n \ge 4. \hfill \tag{$\ast$}
\end{align*}
With this, we can show that $t_j \geq 2^{2^{j+2}}$ for $j \in \N$, which gives $t_{19} \geq 2^{2^{20}}$. For $j=0$, we have $t_0 = 16 \geq 2^{2^{0+2}}$. For an arbitrary $j \in \N$, we get
\[ t_{j+1}=t(t_j) \overset{\text{(\ref{claim:stages2-ii-al1})}}{\geq} t_j^2 \overset{\text{IH}}{\geq} (2^{2^{j+1}})^2 = 2^{2^{j+2}}.\]
\\
Induction step: Let $i \in \N_{19}$ be arbitrary.
\begin{align*}
t_{i+1} &= t_i^{\log^* t_i} \overset{\text{IH}}{\geq} t_i^{\log^* \left(2^{2^{\sum_{n=3}^i \floor{\log(\log^* n)}}}\right)} = t_i^{\log^* \left(\sum_{n=3}^i \floor{\log(\log^* n)} \right)+2} \overset{\text{\ref{claim:algebra}.\ref{claim:algebra-ii}}}{\geq} t_i^{\log^* (i+1)+2}\\
 &\geq t_i^{\log^*(i+1)} \overset{\text{IH}}{\geq} \left (2^{2^{\sum_{n=3}^i \floor{\log(\log^* n)}}}\right )^{\log^*(i+1)} = 2^{2^{\sum_{n=3}^i \floor{\log(\log^* n)}} \cdot \log^*(i+1)}\\
 &= 2^{2^{\sum_{n=3}^i \floor{\log(\log^* n)}} \cdot 2^{\log(\log^*(i+1))}} = 2^{2^{\sum_{n=3}^i \floor{\log(\log^* n)} + \log(\log^*(i+1))}}\\
 &\geq 2^{2^{\sum_{n=3}^i \floor{\log(\log^* n)} + \floor{\log(\log^*(i+1))}}} = 2^{2^{\sum_{n=3}^{i+1} \floor{\log(\log^* n)}}} 
\end{align*}
This proves (\ref{claim:stages2-ii-eq1}) stated above. The following equation holds for all $i \in \N_{19}$. Note that $\log^*(i) \geq 4$.
\begin{align}\label{claim:stages2-ii-eq2}
\begin{aligned}
\floor{\log(\log^* (i/2))} &\overset{\text{\ref{obs:algebra}.\ref{obs:algebra-v}}}{\geq} \floor{\log(\log^* (\log i))} = \floor{\log(\log^* (i)-1)} \overset{\text{\ref{obs:algebra}.\ref{obs:algebra-ii}}}{\geq} \floor{\log(\log^* i) - 1} \\
 &= \floor{\log(\log^* i)} -1 \overset{\text{\ref{obs:algebra}.\ref{obs:algebra-i}}}{\geq} \floor{\log(\log^* i)}/2
\end{aligned}
\end{align}
Finally, for $i \in \N_{19}$, we get the desired upper bound of $t_i$.
\begin{align*}
t_i &\overset{\text{(\ref{claim:stages2-ii-eq1})}}{\geq} 2^{2^{\sum_{n=3}^{i} \floor{\log(\log^* n)}}} \geq 2^{2^{\sum_{n=\ceil{i/2}}^{i} \floor{\log(\log^* (i/2))}}} = 2^{2^{(i+1 - \ceil{i/2}) \cdot \floor{\log(\log^*(i/2))}}}\\
&\overset{\text{(\ref{claim:stages2-ii-eq2})}}{\geq}  2^{2^{(i+1 - \ceil{i/2}) \cdot \floor{\log (\log^* i)}/2}} \geq 2^{2^{(i/2) \cdot \floor{\log(\log^* i)}/2}} = 2^{2^{i \cdot \floor{\log(\log^* i)}/4}}
\end{align*}
\\
To (\ref{claim:stages2-iii}): Base case $i = 5$: It holds that
\[\floor{\log(\log^* 5)} \cdot 5 + 4 \cdot 5 = \floor{\log 3} \cdot 5 + 20 = 5 + 20 = 25.\]
The following calculation shows $t_5 = 2^{6000} < 2^{2^{13}} < 2^{2^{25}}$. 
    \begin{eqnarray*}
        t_0 &=& 2^4\\
        t_1 &=& t(2^4) = (2^4)^{\log^* 2^4} = (2^4)^{1+\log^* 4} = (2^4)^3 = 2^{12}\\
        t_2 &=& t(2^{12}) = (2^{12})^{\log^* 2^{12}} = (2^{12})^{1+\log^* 12} = (2^{12})^4 = 2^{48}\\
        t_3 &=& t(2^{48}) = (2^{48})^{\log^* 2^{48}} = (2^{48})^{1+\log^* 48} = (2^{48})^5 = 2^{240}\\
        t_4 &=& t(2^{240}) = (2^{240})^{\log^* 2^{240}} = (2^{240})^{1+\log^* 240} = (2^{240})^5 = 2^{1200}\\
        t_5 &=& t(2^{1200}) = (2^{1200})^{\log^* 2^{1200}} = (2^{1200})^{1+\log^* 1200} = (2^{1200})^5 = 2^{6000}
    \end{eqnarray*}
Induction step: Let $i \in \N_5$ be arbitrary.
\begin{align*}
t_{i+1} &= t_i^{\log^* t_i} \overset{\text{IH}}{\leq} t_i^{\log^* \left(2^{2^{\floor{\log(\log^* i)} \cdot i + 4i}} \right)} = t_i^{\log^* \left(\floor{\log(\log^* i)} \cdot i + 4i \right) + 2} \overset{\text{\ref{claim:algebra}.\ref{claim:algebra-iii}}}{\leq} t_i^{\log^* (2^i) + 2}\\
 &= t_i^{\log^* (i) + 3}\overset{\text{IH}}{\leq} \left(2^{2^{\floor{\log(\log^* i)} \cdot i + 4i}}\right)^{\log^* (i) + 3} = 2^{2^{\floor{\log(\log^* i)} \cdot i + 4i} \cdot (\log^*(i) + 3)}\\
 &= 2^{2^{\floor{\log(\log^* i)} \cdot i + 4i} \cdot 2^{\log(\log^*(i)+3)}} = 2^{2^{\floor{\log(\log^* i)} \cdot i + 4i + \log(\log^* (i) + 3)}}\\
 &\leq 2^{2^{\floor{\log(\log^* i)} \cdot i + 4i + \floor{\log(\log^* (i) + 3)} + 1}} \overset{\text{\ref{obs:algebra}.\ref{obs:algebra-ii}}}{\leq} 2^{2^{\floor{\log(\log^* i)} \cdot i + 4i + \floor{\log(\log^* i)} + 4}}\\
 &= 2^{2^{\floor{\log(\log^* i)} \cdot (i+1) + 4(i+1)}} \leq 2^{2^{\floor{\log(\log^* (i+1))} \cdot (i+1) + 4(i+1)}}
\end{align*}
To (\ref{claim:stages2-iv}): For $i \in \N^+$, from $t_{i+\floor{\log i}} \geq t_i$ follows $s_i = t_i + t_{i+\floor{\log i}} \leq 2t_{i+\floor{\log i}}$. 
\medskip
\\To (\ref{claim:stages2-v}): For $n \in \R^+$, $\log(n)$ is defined and $\log$ is strictly monotone increasing on $\R^+$.
\medskip
\\
To (\ref{claim:stages2-vi}): First, we provide an upper bound for $t_{i+\floor{\log i}}$. Let $i \in \N_{21}$ be arbitrary.
\begin{align}\label{claim:stages2-v-tag}
\begin{aligned}
t_{i+\floor{\log i}} &\overset{\text{(\ref{claim:stages2-iii})}}{\leq} 2^{2^{\floor{\log(\log^*(i+\floor{\log i}))} \cdot (i+\floor{\log i}) + 4(i + \floor{\log i})}} \overset{\text{\ref{obs:algebra}.\ref{obs:algebra-iii}}}{\leq} 2^{2^{\floor{\log(\log^*(2i))} \cdot 2i + 8i}}\\
 &\overset{\text{\ref{obs:algebra}.\ref{obs:algebra-iv}}}{\leq} 2^{2^{\floor{ \log(\log^*(2^{i-2}))} \cdot 2i + 8i}} = 2^{2^{\floor{ \log(\log^*(i-2)+1)} \cdot 2i + 8i}} \overset{\text{\ref{obs:algebra}.\ref{obs:algebra-ii}}}{\leq} 2^{2^{(\floor{ \log(\log^*(i-2))}+1) \cdot 2i + 8i}}\\
 &= 2^{2^{(\floor{ \log(\log^*(i-2))}+5) \cdot 2i}}
\end{aligned}
\end{align}
With this we can bound $h(t_{i-2})$:
\begin{align*}
h(t_{i-2}) &= 2^{(\log t_{i-2})^{40}} \overset{\text{(\ref{claim:stages2-ii})}}{\geq} 2^{\left(2^{(i-2) \cdot \floor{\log(\log^* (i-2))}/4} \right)^{40}} = 2^{2^{10 (i-2) \cdot \floor{\log(\log^* (i-2))}}}\\
&= 2^{2^{(10i-20) \cdot \floor{\log(\log^*(i-2))}}} \overset{i \geq 21}{\geq} 2^{2^{8i \cdot \floor{\log(\log^*(i-2))}}} = 2^{2^{2i \cdot 4 \cdot \floor{\log(\log^*(i-2))}}}\\
 &\overset{\text{\ref{obs:algebra}.\ref{obs:algebra-i}}}{\geq} 2^{2^{2i \cdot (\floor{\log(\log^*(i-2))}+5)}} \overset{\text{(\ref{claim:stages2-v-tag})}}{\geq} t_{i+\floor{\log i}}
\end{align*}
To (\ref{claim:stages2-vii}): Let $k \in \N^+$ be arbitrary. Choose $m \in \N$ such that for all $n\geq m$ it holds that $\log^*(n) > 2k$. Choose $i_0$ such that $t_{i_0 + \floor{\log i_0}} > m$. Then for all $i \geq i_0$ holds $\log^* t_{i+\floor{\log i}} > 2k$. Furthermore, we obtain
\begin{align*}
p_k(s_i) &= s_i^k \overset{\text{(\ref{claim:stages2-iv})}}{\leq} (2t_{i + \floor{\log i}})^k \leq (t_{i+\floor{\log i}})^{2k} \leq (t_{i+\floor{\log i}})^{\log^* t_{i+\floor{\log i}}} = t(t_{i+\floor{\log i}})\\
 &= t_{i+1+\floor{\log i}} < t_{i+1} + t_{i+1+\floor{\log(i+1)}} = s_{i+1}.
\end{align*}
To (\ref{claim:stages2-viii}): Observe that one can choose some $\hat{m}$ such that for all $i \geq \hat{m}$, it holds
\begin{alphaenumerate}
\item\label{claimproof2:mhat-i} $t(t_i) \geq 4t_i^2$,
\item\label{claimproof2:mhat-ii} $i > 19+\floor{\log i} + \floor{\log(i+4 + \floor{\log i})}$,
\item\label{claimproof2:mhat-iii} $2^i > \floor{\log(\log^*(3i))} \cdot 3i + 12i$, and
\item\label{claimproof2:mhat-iv} $\floor{\log(\log^* i)} \geq 4$.
\end{alphaenumerate}
Then for $i \in \N_{\hat{m}}$ it holds that:
\begin{align*}
&\ 2^{t_{i+\floor{\log i}}+1} \cdot s_{i+2+\floor{\log i}} \cdot (i+\floor{\log i})+ s_{i+4+\floor{\log i}}\\
& \leq 2^{t_{i+\floor{\log i}}+1} \cdot s_{i+4+\floor{\log i}} \cdot (i+1+\floor{\log i})\\
& \overset{(\ref{claim:stages2-iv})}{\leq} 2^{t_{i+\floor{\log i}}} \cdot 2 \cdot 2 \cdot t_{i+4+\floor{\log i} + \floor{\log(i+4 + \floor{\log i})}} \cdot (i+1+\floor{\log i})\\
& \leq 2^{t_{i+\floor{\log i}}} \cdot 4 (t_{i+4+\floor{\log i} + \floor{\log(i+4 + \floor{\log i})}})^2
\\
& \overset{(\ref{claimproof2:mhat-i})}{\leq} 2^{t_{i+\floor{\log i}}} \cdot t(t_{i+4+\floor{\log i} + \floor{\log(i+4 + \floor{\log i})}})\\
& = 2^{t_{i+\floor{\log i}}} \cdot t_{i+5+\floor{\log i} + \floor{\log(i+4 + \floor{\log i})}} \overset{(\ref{claimproof2:mhat-ii})}{\leq} 2^{t_{i+\floor{\log i}}} \cdot t_{2i} < 2^{t_{i + \floor{\log i}}} \cdot t_{3i}\\
&\overset{\text{(\ref{claim:stages2-iii}),(\ref{claimproof2:mhat-ii})}}{\leq}2^{t_{i + \floor{\log i}}} \cdot  2^{2^{\floor{\log(\log^*(3i))} \cdot 3i + 12i}} \overset{\text{(\ref{claimproof2:mhat-iii})}}{\leq} 2^{t_{i + \floor{\log i}}} \cdot 2^{2^{2^{i}}}\\
&\overset{\text{(\ref{claimproof2:mhat-iv})}}{\leq} 2^{t_{i + \floor{\log i}}} \cdot 2^{2^{2^{i \cdot \floor{\log(\log^* i)}/4}}} \overset{\text{(\ref{claim:stages2-ii}),(\ref{claimproof2:mhat-ii})}}{\leq} 2^{t_{i + \floor{\log i}}} \cdot 2^{t_i} \leq 2^{t_{i + \floor{\log i}}} \cdot 2^{t_{i+\floor{\log i}}}\\
&= 2^{2t_{i + \floor{\log i}}} \overset{t_i\geq 4}{\leq} 2^{(t_{i + \floor{\log i}})^{\log^* t_{i + \floor{\log i}}}} = 2^{t_{i + \floor{\log i} + 1}} \leq 2^{t_{i + 1 + \floor{\log (i+1)}}}\tag*{\claimqedhere}
\end{align*}
\end{claimproof}

For the rest of this paper, we fix a constant $\hat{m}$ satisfying the requirements of Claim \ref{claim:stages2}.\ref{claim:stages2-viii}. We also limit the concept of partial oracles. From now on we require that a partial oracle $w$ is either completely defined or completely undefined for a stage $s_i \in S$, i.e., for all words of length $s_i$. We denote the largest stage an oracle $w$ is defined for as $\ls{w}$. Furthermore, whenever we write $w' \sqsupsetneq w$ or $w \sqsubsetneq w'$, then $\ls{w} < \ls{w'}$, i.e., $w'$ must be fully defined for at least the first undefined stage of $w$.

\subparagraph*{Code words.}

We want to encode a fact $y \in \ran(F_k)$ into the oracle using a \emph{code word} $z$ of length $s_i$ for some $i \in \N^+$. For this, we encode $k$ and $y$ in the front part of $z$ and use the back part to pad $z$ to the desired length. For combinatorial reasons, the front part should have length $t_i$, whereas the back part has length $t_{i+\floor{\log i}}$. The following definition makes this precise. 
\begin{definition}[Code words]\label{def:codewords2}
Let $b \in \Sigma^*$ such that $|b| = t_{i + \floor{\log i}}$ for some $i \in \N^+$. Let $k \in \N^+$ and $y \in \Sigma^*$ such that $k+2+|y|  \leq t_i$. Define 
\[c(k,y,b) \coloneqq \underbrace{0^k10^j1y}_{t_i}\underbrace{\phantom{y^k}b\phantom{y^k}}_{t_{i + \floor{\log i}}},\]
where $j \coloneqq t_i-(k+2+|y|)$. Such $c(k,y,b)$ is called code word.

For $i,k \in \N^+$ and $y \in \Sigma^*$, we define 
\[C_{k,y}^i \coloneqq \{c(k,y,b) \mid b \in \Sigma^{t_{i+\floor{\log i}}} \land k+2+|y| \leq t_i\}\]
 as the set of code words for $k$ and $y$ of length $s_i$.  We define $C^i \coloneqq \bigcup _{k\in \N^+,y \in \Sigma^*} C_{k,y}^i$ as the set of all code words of length $s_i$ and $C \coloneqq \bigcup_{i \in \N^+} C^i$ as the set of all code words.
\end{definition}
Observe that $c(k,y,b)$ is well-defined as $i$ is uniquely determined by $|b|$. In the oracle construction, we want to encode facts as compactly as possible, which motivates the following definition.
\begin{definition}[Shortest code words]
Let $i \colon \N^+ \times \Sigma^* \to \N^+$ be the function that maps $(k,y)$ to the smallest number $i$ with $C_{k,y}^i \neq \emptyset$. We say that $c(k,y,b) \in C$ is a shortest code word for $k$ and $y$, if $|c(k,y,b)| = s_{i(k,y)}$.
\end{definition}
For the rest of this paper, if not stated otherwise, we use the notation $c(k,y,b)$ only for appropriate choices of $k,y,b$, i.e., whenever $c(k,y,b) \in C$.

\begin{observation}\label{observation:codeword}
For every code word $c(k,y,b)$ it holds that $|c(k,y,b)| \in S$.
\end{observation}
The following claim shows that $h$ grows as fast as the distance between a word and its shortest code word. Hence, a machine with runtime $O(h)$ is able to query shortest code words for $y$ on input $y$.
\begin{claim}\label{claim:smallest-cw-length-y}
If $|y| \geq k + 2$ for $k \in \N^+$ and $i(k,y) \geq 21$, then $h(|y|) \geq s_{i(k,y)}/2$.
\end{claim}
\begin{claimproof}
Since $s_{i(k,y)}$ is the length of the shortest code word for $k$ and $y$, it holds that $t_{i(k,y)-1} < k + |y| + 2$, because otherwise code words of length $s_{i(k,y)-1}$ would be possible. Thus, 
\begin{align*}
|y| \geq (|y|+k+2)/2 > t_{i(k,y)-1}/2 \geq t_{i(k,y)-2},
\end{align*}
where the last inequality follows from the definition of $t_{i(k,y)}$. Consequently, 
\[h(|y|) \overset{\ref{claim:stages2}.\ref{claim:stages2-v}}{\geq} h(t_{i(k,y)-2}) \overset{\ref{claim:stages2}.\ref{claim:stages2-vi}}{\geq} t_{i(k,y)+\floor{\log i(k,y)}} \overset{\ref{claim:stages2}.\ref{claim:stages2-iv}}{\geq} s_{i(k,y)}/2.\tag*{\claimqedhere}\]
\end{claimproof}
Concluding this section, we bound the number of encodable facts per stage and the number of available code words per fact and stage.
\begin{claim}\label{claim:codewordDependency}
Let $i \in \N^+$.
\begin{romanenumerate}
\item\label{claim:codewordDependency-i} $\card{\{(k,y) \in \N^+ \times \Sigma^* \mid C^i_{k,y} \not = \emptyset\}} \leq 2^{t_i+1}$.

(Meaning: the number of encodable facts on stage $s_i$ is $\leq 2^{t_i + 1}$.)
\item\label{claim:codewordDependency-ii} For all $k \in \N^+$, $y \in \Sigma^*$, either $\card{C_{k,y}^i} = 2^{t_{i + \floor{\log i}}}$ or $C_{k,y}^i = \emptyset$.

(Meaning: the number of code words of length $s_i$ for a fact $(k,y)$ is $2^{t_{i+\floor{\log i}}}$ or $0$.)

\end{romanenumerate}
\end{claim}
\begin{claimproof}
To (\ref{claim:codewordDependency-i}): For a code word $c(k,y,b)$ of length $s_i$ it must hold that $t_i \geq |y|+k+2$. Thus, there are at most $\sum _{j=0}^{t_i} 2^j \leq 2^{t_i + 1}$ possible choices for the strings $0^k$ and $y$.
\medskip
\\
To (\ref{claim:codewordDependency-ii}): Let $k \in \N^+$ and $y \in \Sigma^*$ be fixed. If $C_{k,y}^i \neq \emptyset$, then there exists some $c(k,y,b') \in C_{k,y}^i$. By Definition \ref{def:codewords2}, for a code word $c(k,y,b)$ of length $s_i$ it must hold that $t_i \geq k + |y| + 2$ and $|b| = t_{i+\floor{\log i}}$. Since $c(k,y,b') \in C_{k,y}^i$, the first condition is satisfied. There are no further restrictions on $b$, so there are exactly $2^{t_{i + \floor{\log i}}}$ possible values for $b$, i.e., paddings for encoding fixed $k$ and $y$. Hence $\card{C_{k,y}^i} = 2^{t_{i+\floor{\log i}}}$.
\end{claimproof}

\subparagraph*{Proof selectors.}
In general, we want to use code words to encode facts like $y \in \ran(F_k)$. A code word's membership to the oracle is justified by some computation $F_k(x)=y$. To associate this computation to the code word, we introduce \emph{proof selectors}, which basically map code words to some (not necessarily definite) computation ensuring that the code word encodes ``correct'' information.
\begin{definition}[Proof selector]
A function $\mathcal{\sigma} \colon \N^+ \times \Sigma^* \times \Sigma ^* \to \Sigma^*$ is called a \emph{proof selector} relative to an oracle $w$, if $\dom(\sigma) = \{(k,y,b) \mid c(k,y,b) \in w\}$ and $F_k^w(\sigma (k,y,b))=y$ for $(k,y,b) \in \dom(\sigma)$. 

Let $\mathcal{S} \coloneqq \{s \colon \N^+ \times \Sigma^* \times \Sigma^* \to \Sigma^*\}$ denote the set of all potential proof selectors. For an oracle $w$, let $\mathcal{S}^w \coloneqq \{s \mid s \mbox{ is a proof selector relative to } w\}$ denote the set of all proof selectors relative to $w$.
\end{definition}
\begin{observation}\label{obs:proof-selection}
For every oracle $w$, $\sigma \in \mathcal{S}^w$, and $(k,y,b) \in \dom(\sigma)$ it holds that $\hhat{y}_{F_k^w}$ is defined and $|\hhat{y}_{F_k^w}| \leq |\sigma (k,y,b)|$.
\\
(Meaning: Recall that $\hhat{y}_{F_k^w}$ denotes the shortest $F_k^w$-proof for $y$, see the section ``Proof systems''. Hence, selected $F_k^w$-proofs via $\sigma$ have at least the length of shortest $F_k^w$-proofs.)
\end{observation} 
\begin{observation}\label{obs:well-defined}
For every oracle $w$, $X \subseteq \Sigma^*$, and $\sigma \in \mathcal{S}^w$, if $X \cap w = \emptyset$, then $\{(k,y,b) \mid c(k,y,b) \in X\} \cap \dom(\sigma) = \emptyset$.
\\
(Meaning: When extending an oracle $w$ by code words, we can extend $\sigma \in \mathcal{S}^w$ by corresponding triples without violating $\dom(\sigma)$.)
\end{observation}

\subparagraph*{Dependency graphs.}
Whenever we encode a fact into the oracle via some code word, the correctness of the encoding depends on other words inside or outside the oracle. This is because the correctness depends on some computation and the computation can change its behavior when given different answers to its oracle queries. We capture these dependencies in a graph data structure. The full idea behind the following definition can only be understood after the next two sections, hence we discuss it there. However, since the next two sections require the definition of dependency graphs, we are already stating it. Recall that $Q(\cdot)$ denotes the set of oracle queries of a computation.
\begin{definition}[Dependency graphs]\label{def:dependency-graph}
Define the dependency graph $\mathcal{G}^w_{\sigma } \coloneqq (V,E)$ for an oracle $w$ and $\sigma \in \mathcal{S}$ as
\begin{align*}
V &\coloneqq \Sigma^*\\
E &\coloneqq \{(u,v) \mid u = c(k,y,b) \in w \cap C,\ \sigma(k,y,b) \mbox{ defined, } v \in Q(F_k^w(\sigma(k,y,b)))\}
\end{align*}
For a graph $\mathcal{G}^w_{\sigma } = (V,E)$, $q \in V$ and $Q \subseteq V$ we define the set of neighbors of $q$ (resp., of $Q$) as follows:
\begin{align*}
N_{\mathcal{G}^w_{\sigma }}(q) &\coloneqq \{v \mid (q,v) \in E\}\\
N_{\mathcal{G}^w_{\sigma }}(Q) &\coloneqq \bigcup _{q \in Q}N_{\mathcal{G}^w_\sigma}(q)
\end{align*}
\end{definition}
So the graph considers all words in $\Sigma ^*$, edges only start at code words inside the oracle and edges represent oracle queries from an associated $\FP$-computation. As we will see in the section ``Managing dependencies'', this graph contains the information on which words encoded code words depend. 

Most of the time we will consider the set $N_{\mathcal{G}^w_{\sigma }}(w)$ for an oracle $w$, which consists of all words in the graph with an indegree $\geq 1$. Usually, this can be interpreted as the set of all words all encoded code words depend on.

\subparagraph*{Requirement functions.}
During the construction of the oracle, we successively add requirements that we maintain. These are specified by a (partial) function $r \colon \N^+ \cup (\N^+ \times \N^+) \to \N$, called \emph{requirement function}. The set of all requirement functions is defined as
\[\mathcal{R} = \{r \colon \N^+ \cup (\N^+ \times \N^+) \to \N\}.\]
By Claim \ref{claim:stages2}.\ref{claim:stages2-vii}, for every $k \in \N^+$ there is some $i_0 \in \N$ such that $p_k(s_i) < s_{i+1}$ for all $i \geq i_0$. For $k \in \N^+$, let $i_0(k)$ denote the smallest such $i_0$. We call a requirement function $r \in \mathcal{R}$ \emph{valid}, if it satisfies the following requirements (let $k,j \in \N^+$):
\begin{enumerate}[R1]
\item \label{R1} If $r(k)=m$, then $m \geq s_{i+\floor{\log i}}+k+2$ for $i \in \N$ with $\floor{\log i} > 2k+2$ and $i \geq i_0(k) + \hat{m}$.

(Meaning: Words $y$ of length $\geq m$ are sufficiently long such that the prerequisites of Claim \ref{claim:smallest-cw-length-y} are satisfied, the distance between two consecutive stages $m\leq s_i < s_{i+1}$ is bigger than $p_k(s_i)$, there are more than $2k$ stages between $m \leq s_i$ and $s_{i+\floor{\log i}}$, and the Property \ref{claim:stages2}.\ref{claim:stages2-viii} holds.)

\item \label{R2} If $r(k,j) = m > 0$, then $m > r(k) + s_{i_0(k \cdot j)}$ with $k \in \dom(r)$.
\\
(Meaning: Besides R\ref{R1}, we additionally require that the distance between two consecutive stages $m \leq s_i < s_{i+1}$ is bigger than $p_k(p_j(s_i))$.)
\item\label{R3} If $r(k,j)=r(k,j')>0$, then $j=j'$.
\\
(Meaning: For each $k$, there is at most one $j$ with $r(k,j)>0$.)
\end{enumerate}
We will prove that R\ref{R1} to R\ref{R3} are satisfied for various requirement functions. To prevent confusion, we use the notation R\ref{R1}($r$), R\ref{R2}($r$), and R\ref{R3}($r$) to clearly state the referred requirement function.
\subparagraph*{Valid pairs.}
A pair ($w,\sigma $) consisting of a partial oracle $w \in \Sigma^*$ and a function $\sigma \in \mathcal{S}$ is called \emph{$r$-valid} for $r \in \mathcal{R}$, if $r$ is valid and the following requirements are satisfied (let $k,j \in \N^+$):
\begin{enumerate}[V1]
\item \label{V1} $\sigma \in \mathcal{S}^w$.

{(Meaning: $\sigma$ is a proof selector relative to $w$. In particular, code words inside $w$ encode correct information.)}

\item \label{V2} $w \subseteq C$.

{(Meaning: The oracle contains only code words.)}

\item \label{V3} $\card{(w \cap C^i_{k',y})} \leq 1$ for all $k',i \in \N^+$ and $y \in \Sigma^*$. 

{(Meaning: No fact $y \in \ran(F_{k'}^w)$ is encoded by more than one code word per stage.)}

\item \label{V4} If $c(k,y,b) \in w$ and $s_i \coloneqq |c(k,y,b)|$, then $p_k(|\sigma(k,y,b)|) < s_{i+2}$.

{(Meaning: For code words inside the oracle, the proof selector provides proofs whose lengths are bounded.)}

\item \label{V5} If $r(k,j) = 0$, then there is some $c(k,y,b)$ such that $\hhat{y}_{F_k^w}$ is defined, $F_k^w(\hhat{y}_{F_k^w})$ is definite, $c(k,y,b) \in w$, and $p_j(|c(k,y,b)|) < |\hhat{y}_{F_k^w}|$. 

{(Meaning: For every extension of $w$, $F_k$ has a shortest proof for some $y$ that is at least $p_j$ longer than an encoded code word. Since $F_k^w(\hhat{y}_{F_k^w})$ is definite, no shorter proof can emerge when $w$ is extended. This helps to diagonalize against simulation functions whose output-lengths are bounded by $p_j$.)}

\item \label{V6} If $r(k,j) = m > 0$, then for all $y \in \Sigma^{>m}$, if $\hhat{y}_{F_k^w}$ exists, $\ls{w} \geq i(k,y)$,  and $|\hhat{y}_{F_k^w}| \leq p_j(\ls{w})$, then $C_{\smash{k,y}}^{\smash{i(k,y)}} \cap w \neq \emptyset$.

{(Meaning: We require that some shortest code word for $y \in \ran(F_k^w)$ belongs to the oracle when $y$ is ``long enough'', $\hhat{y}_{F_k^w}$ is not too long, and the oracle is defined for code words for $k,y$. This helps to compute $\ran(F_k)$ efficiently relative to the final oracle, because almost all its facts will be encoded by shortest code words. Observe that $F_k^w(\hhat{y}_{F_k^w})$ is not required to be definite, but $|\hhat{y}_{F_k^w}|$ is bounded.)}

\item \label{V7} For all $i \in \N^+$ such that $s_i \leq \ls{w}$, there is no nonempty set $X$ with
\begin{enumerate}
\item\label{V7a} $X \subseteq \{c(k,y,b) \in C_{k,y}^i \mid \exists j \mbox{ with } r(k,j)=m>0,\ k \in \N^+,\ y \in \Sigma^{>m}\}$,
\item\label{V7b} $X \cap N_{\mathcal{G}_{\sigma }^{w}}(w) = \emptyset$,
\item\label{V7c} $X \cap w = \emptyset$,
\item\label{V7d} $\card{((X \cup w) \cap C_{k,y}^i) } \leq 1$ for all $k \in \N^+$ and $y \in \Sigma^*$,
\item\label{V7e} if $c(k,y,b) \in X$, then $\hhat{y}_{F_k^{w \cup X}}$ exists with $|\hhat{y}_{F_k^{w \cup X}}| \leq p_j(s_i)$ for $j$ with $r(k,j)>0$.
\end{enumerate}

{(Meaning: All stages of the oracle are maximally filled with code words having certain properties. Property V\ref{V7a} restricts $X$ to code words of the same length, V\ref{V7b} and V\ref{V7c} ensure that those code words do not interfere with encodings already inside the oracle, and V\ref{V7d} ensures V\ref{V3}. The key property is V\ref{V7e}, because it ensures that the code words from $X$ have a ``short'' proof.)}

\item \label{V8} $\ls{w} \geq \max\{r(k) \mid k \in \dom(r)\}$. 

{(Meaning: The combinatorial properties of R\ref{R1}(r) hold for $\ls{w}$.)}
\end{enumerate}
As in the previous section, we use the following notation for these requirements to unambiguously refer to the respective oracle $w$, proof selector $\sigma$, and requirement function $r$: V1($w,\sigma$), V2($w$), V3($w$), V4($w,\sigma$), V5($w,r$), V6($w,r$), V7($w,\sigma,r$), V8($w,r$). Furthermore, we use the notation V\ref{V7}($w,\sigma,r,i$) for $i \in \N^+$ to refer to the requirement V\ref{V7}($w,\sigma,r$) for a specific stage $s_i$. Note that V\ref{V7}($w,\sigma,r$) is satisfied if and only if V\ref{V7}($w,\sigma,r,i$) is satisfied for all $i \in \N^+$ such that $s_i \leq \ls{w}$.

These requirements can be divided into two groups. The requirements V\ref{V1}, V\ref{V2}, V\ref{V3}, and V\ref{V4} are restrictive requirements, because they restrict the membership of words to the oracle. The requirement V\ref{V1} additionally restricts the function $\sigma$. The requirements V\ref{V5}, V\ref{V6}, V\ref{V7} and V\ref{V8} are demanding requirements, because they demand words to belong to the oracle. In particular, V\ref{V5} is about diagonalizing against proof systems, V\ref{V6} is about encoding the range of proof systems into $\NQP$, and V\ref{V7} basically describes how we perform the oracle construction.

Furthermore, we can now grasp the intuition behind the definition of dependency graphs. In order for property V\ref{V1} to remain satisfied for an ever-growing oracle, the already selected proofs of $\sigma$ have to stay correct. Hence, we define the dependency graph such that exactly these dependencies, i.e., the oracle queries of $\sigma$-selected computations, are captured by the graph. Then we can use this graph to avoid changing the membership of these words and keep V\ref{V1} satisfied.

We call $(w,\sigma ,r) \in \Sigma^* \times \mathcal{S} \times \mathcal{R}$ a \emph{valid triple}, when $r$ is valid and ($w,\sigma $) is $r$-valid. We call a pair $(w,\sigma) \in \Sigma^* \times \mathcal{S}$ \emph{valid}, if V\ref{V1}($w,\sigma$), V\ref{V2}($w$), V\ref{V3}($w$), and V\ref{V4}($w,\sigma$) are satisfied. Intuitively, a valid pair ($w,\sigma $) may not be $r$-valid for some $r \in \mathcal{R}$, but only because it is lacking some mandatory code words. So, in a sense, the pair is ``correct'', because it satisfies all \emph{restrictive} requirements, but it may be ``incomplete'', because it does not satisfy all demanding properties. We can prove useful properties for the graphs of such pairs which help to identify candidates for mandatory code words we can ``safely'' add to $w$ in order to complete $(w,\sigma )$ step by step.

Finally, we introduce further abbreviations. For an oracle $w$, a proof selector $\sigma$ and $\circ \in \{\sqsupseteq, \sqsupsetneq, \supseteq, \supsetneq\}$, we write $(w',\sigma') \circ (w,\sigma)$ for $w' \circ w$ and $\sigma' \sqsupseteq \sigma$. 

\subparagraph{Extended sketch of the construction.} 
Having all main definitions of this section, we give an extended overview of the oracle construction. We construct an oracle $O$ such that relative to $O$, all sets in $\RE^O \setminus \NQP^O$ do not have an optimal proof system. For this, we construct an $r$-valid pair $(O,\sigma)$ such that for all $k \in \N^+$, either $r(k,j) = 0$ for all $j \in \N^+$ or $r(k,j) > 0$ for some $j \in \N^+$. Then $O$ has the desired properties. Consider any $k \in \N^+$. 
\begin{itemize}
    \item If $r(k,j) = 0$ for all $j \in \N^+$, then by requirement V\ref{V5}: For each $j \in \N^+$, there is a shortest proof $\hhat{y}_{\smash{F_k^O}}$ and a code word $c(k,y,b) \in O$ such that the proof is more than $p_j$ longer than the code word. By V\ref{V1}, we can prove that whenever $c(k,y,b) \in O$, also $y \in \ran(F_{\smash{k}}^{\smash{O}})$, i.e., the code words encode correct information. Hence, we can define a witness proof system for $\ran(F_{\smash{k}}^{\smash{O}})$ that queries the oracle for code words $c(k, \cdot, \cdot)$ and, if successful, returns the encoded proof. By this, the witness proof system has infinitely often superpolynomial shorter proofs than $F_{\smash{k}}^{\smash{O}}$. This shows that $F_{\smash{k}}^{\smash{O}}$ is not optimal.
    \item If $m \coloneqq r(k,j) > 0$ for some $j \in \N^+$, then by requirement V\ref{V6}: Whenever $F_{\smash{k}}^{\smash{O}}$ has a proof for $y \in \Sigma^{>m}$, then there is some shortest code word $c(k,y,b) \in O$. Having this, we can define a non-deterministic algorithm that on input $y$ searches for shortest code words $c(k, y, \cdot)$ inside the oracle and accepts, if the search is successful. As argued above, by V\ref{V1}, the algorithm accepts $\ran(F_{\smash{k}}^{\smash{O}})$. The runtime of this algorithm can be bounded by the length of shortest code words, which can be shown to be $h^c$ for some constant $c$. This implies $\ran(F_{\smash{k}}^{\smash{O}}) \in \NQP^O$.
\end{itemize}
It remains to construct the respective $r$-valid pair $(O,\sigma)$. Similar to the oracle in section \ref{sec:thin-oracle}, we construct $(O,\sigma)$ in steps. Each step treats a task $\tau_{k,j}$, where we either define $r(k,j) \coloneqq 0$ or $r(k,j) \coloneqq m > 0$ and extend the pair such that it remains $r$-valid. We desribe step $s$ of the oracle construction. Let $w_{s-1}$, $\sigma_{s-1}$, and $r_{s-1}$ denote the oracle, proof selector, and requirement function from step $s-1$ and let $\ls{w_{s-1}} = s_{i-1}$. We define $w_s$, $\sigma_s$ and $r_s$.
\begin{itemize}
    \item Let $r^0 \coloneqq r_{s-1} \cup \{(k,j) \mapsto 0\}$. If there is an $r^0$-valid extension $(w,\sigma)$ of $(w_{s-1},\sigma_{s-1})$, then define $r_s \coloneqq r^0$, $w_s \coloneqq w$, and $\sigma_s \coloneqq \sigma$.
        
    \item Otherwise define $r_s \coloneqq r_{s-1} \cup \{(k,j) \mapsto m\}$ for some $m \in \N^+$ with specific properties. Extend $(w_{s-1},\sigma_{s-1})$ to $(w_s,\sigma_s)$ such that the latter is $r_s$-valid. Extend $w_s$ by at least one stage.
        
\end{itemize}
In the first case, we choose an $r_s$-valid pair $(w_s,\sigma_s)$. In the second case however, we have to extend $(w_{s-1},\sigma_{s-1})$ ``manually''.  

The central idea of extending $w_{s-1}$ to $w_s$ is to satisfy V\ref{V7} for stage $s_i$, i.e., we make the stage $s_i$ as ``full as possible''. In fact, all requirements can be satisfied in this way, which we briefly explain as follows: 
    
The requirement V\ref{V1} is satisfied, because the dependencies of existing encodings are respected by V\ref{V7b}. Furthermore, by V\ref{V7e}, for all newly added words we can choose an associated proof for the proof selector. V\ref{V2} is satisfied, because of V\ref{V7a}. V\ref{V3} is satisfied, because of V\ref{V7d}. V\ref{V4} is satisfied, because by V\ref{V7e}, the newly added code words have bounded proofs. V\ref{V5} is satisfied, because it is about definite computations that are not influenced by the modifications at stage $s_i$. V\ref{V7} is satisfied for the stage $s_i$, because we define it accordingly. V\ref{V7} for the stages $< s_i$ is not influenced by the modifications at stage $s_i$, because the stages are too far apart.

In contrast, V\ref{V6} is much more difficult to satisfy. Intuitively, whenever $r(k,j) > 0$ for some $j$ and an $F_k$-proof for some word $y$ emerges, V\ref{V6} requires a shortest code word for $k$ and $y$ in the oracle. The problem is that an $F_k$-proof for $y$ can emerge when the oracle is already defined far beyond the stage containing shortest code words for $k$ and $y$. Adding such a shortest code word afterwards makes the previous construction invalid, as the added word can affect longer code words. We can show that our construction prevents this ``delayed'' emergence of $F_k$-proofs. Greatly simplified, if the shortest code words for $k$ and $y$ have length less than $s_i$ and an $F_k$-proof for $y$ emerges when we define the stage $s_i$ for $w_{s-1}$, then we obtain a contradiction as follows: We add a code word $c(k,y,b)$ at stage $s_{i-1}$ respecting the existing encodings. Either $F_k$ develops a significantly shorter proof, but then V\ref{V7} was not satisfied for $(w_{s-1},\sigma_{s-1}, r_{s-1})$, because the stage $s_{i-1}$ was not ``as full as possible'', a contradiction. Or the shortest $F_k$-proof remains ``long'', but then we could have achieved V\ref{V5} with $r(k,j)=0$ (i.e., long shortest proof with rather short encoding) instead of V\ref{V6} with $r(k,j)>0$, which the oracle construction would have preferred, a contradiction. The latter case requires that we provide a new triple that is valid and that the oracle construction prefers over some chosen triple. Hence, for this new triple, we also have to show that it satisfies V\ref{V6}. So we can show the satisfaction of V\ref{V6} for one triple by showing the satisfaction of V\ref{V6} for another triple. We formalize this relationship as a recursion and show that V\ref{V6} is satisfied in the base case, from which we obtain that $(w_s,\sigma_s,r_s)$ satisfies V\ref{V6}. This recursive relationship makes the proof significantly more technical.

\subparagraph*{Managing dependencies.}
The following claim shows that if we extend a pair $(w,\sigma)$ in a certain way, then the new set of dependencies is a superset of the previous one. This helps to argue that the dependencies for $(w,\sigma)$ are still maintained even when considering extensions of $(w,\sigma)$.
\begin{claim}\label{claim:graph-monotonicity}
Let $v,w \in \Sigma^*$ and $\sigma, \rho \in \mathcal{S}$. If $(v,\rho) \supseteq (w,\sigma)$, and $v$ and $w$ agree on $N_{\mathcal{G}^w_{\sigma }}(w)$, then $N_{\mathcal{G}^v_{\rho }}(v)  \supseteq N_{\mathcal{G}^w_{\sigma }}(w)$.
\end{claim}
\begin{claimproof}
Let $u \in N_{\mathcal{G}^w_\sigma}(w)$ be arbitrary. Then there is $c(k,y,b) \in w$ with $(k,y,b) \in \dom(\sigma)$, and $u \in Q(F_k^w(\sigma(k,y,b))) \subseteq N_{\mathcal{G}^w_\sigma}(w)$. Since $v$ and $w$ agree on $N_{\mathcal{G}^w_\sigma}(w)$, it follows that $Q(F_k^w(\sigma(k,y,b))) = Q(F_k^v(\sigma(k,y,b)))$, because all oracle queries are answered in the same way relative to $w$ and $v$. Observe that $c(k,y,b) \in v$ as $v \supseteq w$ and $(k,y,b) \in \dom (\rho)$ with $\rho(k,y,b)=\sigma(k,y,b)$ as $\rho \sqsupseteq \sigma$. This implies 
\[u \in Q(F_k^v(\sigma(k,y,b))) = Q(F_k^v(\rho(k,y,b))) \subseteq N_{\mathcal{G}^v_\rho}(v).\] 
Hence, $N_{\mathcal{G}^v_{\rho }}(v)  \supseteq N_{\mathcal{G}^w_{\sigma }}(w)$.
\end{claimproof}
The next lemma shows that valid pairs do not contain too many dependencies. It is possible to find many ``independent'' spots for code words of encodable facts, even approximately $\log$ stages below the largest defined stage. Recall that the constant $\hat{m}$ is defined after Claim \ref{claim:stages2}.\ref{claim:stages2-viii}.
\begin{lemma}\label{lemma:graph-structure}
Let $k \in \N^+$, $y \in \Sigma^*$, $r \in \mathcal{R}$, $(w,\sigma)$ be a valid pair with V\ref{V8}($w,r$) being satisfied, $\kappa \in \N^+$ with $\kappa \coloneqq \max\{k \mid k \in \dom(r)\}$ and $\ls{w} \leq s_\ell$ for some $\ell \in \N^+$.
\[\text{For }j > \ell - 2\kappa\text{, if }C_{k,y}^j \not = \emptyset\text{, then }\card{(C_{k,y}^j \setminus N_{\mathcal{G}^w_{\sigma }}(w))} \geq s_{\ell+2}.\]
\end{lemma}
\begin{proof}
We first show a slightly different version of the lemma, which is more compatible with the estimates from Claim \ref{claim:stages2}. 
\begin{claim}\label{claim:graph-structure}
Let $k \in \N^+$, $y \in \Sigma^*$, $r \in \mathcal{R}$, $(w,\sigma)$ be a valid pair, and $i \in \N_{\hat{m}}$ with $\ls{w} \leq s_{i+\floor{\log i}}$.
\begin{romanenumerate}
\item\label{claim:graph-structure-i} $\card{N_{\mathcal{G}^w_{\sigma }}(w)} < 2^{t_{i+1+\floor{\log (i+1)}}} - s_{i+4+\floor{\log i}}$.

(Meaning: Bounds the number of all dependencies in $w$.)
\item\label{claim:graph-structure-ii} For $j > i$, if $C_{k,y}^j \not = \emptyset$, then $\card{(C_{k,y}^j \setminus N_{\mathcal{G}^w_{\sigma }}(w))} \geq s_{i+4+\floor{\log i}}$.

(Meaning: The number of code words of length $s_j$ (i.e., up to approximately $\log i$ stages below $s_{i+\floor{\log i}}$) for an encodable fact is far greater than the number of all dependencies in $w$.)
\end{romanenumerate}
\end{claim}
\begin{claimproof}
To (\ref{claim:graph-structure-i}): We show an upper bound for $\card{N_{\mathcal{G}^w_{\sigma }}(w)}$ by multiplying the maximum number of code words that $w$ contains with the maximum number of neighbors each code word can have in $\mathcal{G}^w_{\sigma }$. By Claim \ref{claim:codewordDependency}.\ref{claim:codewordDependency-i}, we can encode for at most $2^{t_{i+\floor{\log i}}+1}$ different pairs $(k,y)$ at stage $s_{i+\floor{\log i}}$. By V\ref{V3}($w$), $w$ contains at most one code word of length $s_{i+\floor{\log i}}$ for each such pair. These bounds also hold for all stages $< s_{i+\floor{\log i}}$. Hence, the number of code words inside $w$ is at most 
\[(i+\floor{\log i}) \cdot 2^{t_{i+\floor{\log i}}+1}.\]
The number of outgoing edges of a code word $c(k,y,b) \in w$ is bounded by the computation length of $F_k^w(\sigma(k,y,b))$, which is $p_k(|\sigma (k,y,b)|)$. Note that by V\ref{V1}($w,\sigma$), $\sigma(k,y,b)$ is defined. Since $w$ contains only code words of length $\leq s_{i+\floor{\log i}}$ and V\ref{V4}($w,\sigma$) holds, the number of outgoing edges of a code word $c(k,y,b)$ is bounded by $p_k(|\sigma (k,y,b)|) < s_{i+\floor{\log i} + 2}$. Thus, the number of nodes with indegree $1$ in $\mathcal{G}^w_{\sigma}$ is
\[\card{N_{\mathcal{G}^w_{\sigma }}(w)} \leq 2^{t_{i + \floor{\log i}}+1} \cdot s_{i+\floor{\log i} + 2} \cdot (i + \floor{\log i}).\]
By $i \in \N_{\hat{m}}$ and Claim \ref{claim:stages2}.\ref{claim:stages2-viii}, we obtain $\card{N_{\mathcal{G}^w_{\sigma }}(w)} < 2^{t_{i+1+\floor{\log(i+1)}}} - s_{i+4+\floor{\log i}}$.
\medskip
\\
To (\ref{claim:graph-structure-ii}): For $j > i$ and $C_{k,y}^j \not = \emptyset$, we obtain $\card{C_{k,y}^j} = 2^{t_{j+\floor{\log j}}} \geq 2^{t_{i+1+\floor{\log(i+1)}}}$ by Claim \ref{claim:codewordDependency}.\ref{claim:codewordDependency-ii}. Using (\ref{claim:graph-structure-i}), we obtain
\[\card{(C_{k,y}^j \setminus N_{\mathcal{G}^w_{\sigma }}(w))} \geq 2^{t_{i+1+\floor{\log(i+1)}}} - (2^{t_{i+1+\floor{\log(i+1)}}} - s_{i+4+\floor{\log i}}) = s_{i+4+\floor{\log i}}.\tag*{\claimqedhere}\]
\end{claimproof}
Now, we begin to prove the lemma. Choose $i \in \N_2$ such that 
\begin{equation}\label{eqn:graph-structure-1}
i-1 + \floor{\log (i-1)} \leq \ell < i + \floor{\log i}
\end{equation}
and hence
\begin{equation}\label{eqn:graph-structure-2}
i-2+\floor{\log i} \leq i - 1 + \floor{\log  (i-1)} \leq \ell.
\end{equation}
By V\ref{V8}($w,r$) it holds that $\ls{w} \geq \max\{r(k) \mid k \in \dom(r)\}$, implying $\ls{w} \geq r(\kappa)$. Since $\kappa \in \dom(r)$ exists, we have 
\[s_{i + \floor{\log i}} \overset{\text{(\ref{eqn:graph-structure-1})}}{>} s_\ell \geq \ls{w} \geq r(\kappa) \overset{\text{R\ref{R1}($r$)}}{\geq} s_{i' + \floor{\log i'}} + \kappa + 2 \geq s_{i' + \floor{\log i'}}\]
for some $i' \geq \hat{m}$ with $\floor{\log i'} > 2\kappa + 2$. Hence, $i > i' \geq \hat{m}$ and thus also $\floor{\log i} > 2\kappa + 2$. By Claim \ref{claim:graph-structure}.\ref{claim:graph-structure-ii} invoked with $r,w,\sigma,i$, it holds that for
\[j > \ell - 2\kappa \overset{\text{(\ref{eqn:graph-structure-2})}}{\geq} i - 2 + \floor{\log i} - 2\kappa \geq i - 2 + \floor{\log i} - (\floor{\log i} - 2) = i,\]
if $C_{k,y}^j \not = \emptyset$, then $\card{(C_{k,y}^j \setminus N_{\mathcal{G}^w_{\sigma }}(w))} \geq s_{i+4+\floor{\log i}} \overset{\text{(\ref{eqn:graph-structure-1})}}{\geq} s_{\ell+2}$.
\end{proof}
\subparagraph*{Valid extensions.}
We collect several arguments concerning the task to maintain the validity of a given pair. The first claim shows that $r$-validity is preserved for $r' \sqsubseteq r$.
\begin{claim}\label{claim:hereditary-r}
Let $X \in \{1,\dots,8\}$. If VX is satisfied relative to $(w,\sigma,r)$, then also relative to $(w,\sigma,r')$ for $r' \sqsubseteq r$.
\end{claim}
\begin{claimproof}
The requirements V\ref{V1}, V\ref{V2}, V\ref{V3}, and V\ref{V4} do not depend on a requirement function. For V\ref{V5} and V\ref{V6}, for each $(k,j) \in \dom(r)$ the required properties are independent from other values from $r$. For V\ref{V8} it holds that $\ls{w} \geq \max\{r(k) \mid k \in \dom(r)\} \geq \max\{r'(k) \mid k \in \dom(r')\}$. 

For V\ref{V7}, the requirement function appears in V\ref{V7a} and V\ref{V7e}. In V\ref{V7a}, every permitted choice of $X$ relative to $r'$ is also permitted relative to $r$. Similarly, any code word from $X$ satisfying V\ref{V7e} relative to $r'$ also satisfies V\ref{V7e} relative to $r$. Hence, if there is an $X$ satisfying V\ref{V7a}($w,\sigma,r',i)$ to V\ref{V7e}($w,\sigma,r',i$) for some $i \in \N^+$, $X$ also satisfies V\ref{V7a}($w,\sigma,r,i)$ to V\ref{V7e}($w,\sigma,r,i$). Thus, when V\ref{V7}($w,\sigma,r$) is satisfied, also V\ref{V7}($w,\sigma,r'$) is satisfied.
\end{claimproof}
The next claim shows that most requirements are ``local''. If they are satisfied for a pair $(w,\sigma)$ or some $r \in \mathcal{R}$, then they are also satisfied for some kind of extension of them.
\begin{claim}\label{claim:requirement-monotonicity}
Let $r \in \mathcal{R}$ be valid, $w \in \Sigma^*$ be a partial oracle, $\sigma \in \mathcal{S}^w$.
\begin{romanenumerate}
\item\label{claim:requirement-monotonicity-i} For $r' \sqsupseteq r$, R\ref{R1}($r'$) and R\ref{R2}($r'$) are not violated for any element from $\dom(r)$.
\item\label{claim:requirement-monotonicity-ii} If V\ref{V2}($w$) is satisfied, then for $(w',\sigma') \supseteq (w,\sigma)$, V\ref{V2}($w'$) is not violated for any $x \in w$.
\item\label{claim:requirement-monotonicity-iii} If V\ref{V3}($w$) is satisfied, then for $(w',\sigma') \sqsupseteq (w,\sigma)$, no word from $\Sigma^{\leq \ls{w}}$ is part of a violation of V\ref{V3}($w'$). 
\item\label{claim:requirement-monotonicity-iv} If V\ref{V4}($w,\sigma$) is satisfied, then for $(w',\sigma') \supseteq (w,\sigma)$, V\ref{V4}($w',\sigma'$) is not violated for any $x \in w$.
\item\label{claim:requirement-monotonicity-v} If V\ref{V5}($w,r$) is satisfied, then for $w' \sqsupseteq w$, V\ref{V5}($w',r$) is satisfied.
\item\label{claim:requirement-monotonicity-vi}If V\ref{V7}($w,\sigma,r,i$) is satisfied, then for $(w',\sigma') \supseteq (w,\sigma)$ with $w'^{<s_{i+1}} = w^{<s_{i+1}}$ and $N_{\mathcal{G}^{\smash{w'}}_{\smash{\sigma'}}}(w') \supseteq N_{\mathcal{G}^w_\sigma}(w)$, V\ref{V7}($w',\sigma',r,i$) is satisfied.
\item\label{claim:requirement-monotonicity-vii} If V\ref{V8}($w,r$) is satisfied, then for $w' \sqsupseteq w$, V\ref{V8}($w',r$) is satisfied.
\end{romanenumerate} 
\end{claim}
\begin{claimproof}
To (\ref{claim:requirement-monotonicity-i}): Observe that once the requirements R\ref{R1} and R\ref{R2} are satisfied for $r$ and an element from $\dom(r)$, they remain satisfied relative to any extension of $r$, because the respective inequalities still hold.
\medskip
\\
To (\ref{claim:requirement-monotonicity-ii}): If V\ref{V2}($w'$) is violated for some $x \in w$, then $x \notin C$ and V\ref{V2}($w$) would also be violated for this $x$.
\medskip
\\
To (\ref{claim:requirement-monotonicity-iii}): V\ref{V3}($w$) requires the absence of specific words in $w$. Since $w'$ extends $w$ only by words $> \ls{w}$, words from $\Sigma^{\leq \ls{w}}$ can not be part of a violation of V\ref{V3}($w'$).
\medskip
\\
To (\ref{claim:requirement-monotonicity-iv}):  If V\ref{V4}($w',\sigma'$) is violated for some $c(k,y,b) \in w$ with $s_i \coloneqq |c(k,y,b)|$, then by $\sigma \in \mathcal{S}^w$ and $\sigma \sqsubseteq \sigma'$, $p_k(|\sigma(k,y,b)|) = p_k(|\sigma'(k,y,b)|) \geq s_{i+2}$. Hence, V\ref{V4}($w,\sigma$) would also be violated for this $c(k,y,b)$. 
\medskip
\\
To (\ref{claim:requirement-monotonicity-v}): In general, V\ref{V5} is about definite computations and the existence of some word in the oracle. We show that both properties are not changed when extending $w$ to $w'$. Observe that for $c(k,y,b)$ coming from V\ref{V5}($w,r$), $F_k^w(\hhat{y}_{F_k^w})$ is definite. Hence, for all $x < \hhat{y}_{F_k^w}$, $F_k^w(x)$ is also definite. Thus, these computations do not change relative to $w' \sqsupseteq w$. From this follows $\hhat{y}_{F_k^{w'}} = \hhat{y}_{F_k^w}$ and $F_k^{w'}(\hhat{y}_{F_k^{w'}})$ is also definite. By V\ref{V5}($w,r$), $c(k,y,b) \in w \subseteq w'$ and $p_j(|c(k,y,b)|) < |\hhat{y}_{F_k^{w}}| = |\hhat{y}_{F_k^{w'}}|$ hold. Hence, V\ref{V5}($w',r$) is satisfied.
\medskip
\\
To (\ref{claim:requirement-monotonicity-vi}): Assume that V\ref{V7}($w',\sigma',r,i$) is violated. Then there is some nonempty set $X$ containing words of length $s_i$ and satisfying the properties V\ref{V7a}($w',\sigma',r,i$) to V\ref{V7e}($w',\sigma',r,i$). We show that in this case V\ref{V7}($w,\sigma,r,i$) is also violated witnessed by the set $X$: 
\begin{itemize}
\item Observe that $X$ satisfies V\ref{V7a}($w,\sigma,r,i$), because it is independent of $(w,\sigma)$.
\item $X$ satisfies V\ref{V7c}($w,\sigma,r,i$) and V\ref{V7d}($w,\sigma,r,i$), because $w^{\leq s_i} = w'^{\leq s_i}$.
\item By $N_{\mathcal{G}^w_\sigma} \subseteq N_{\smash{\mathcal{G}^{w'}_{\smash{\sigma'}}}}$ and $X \cap N_{\smash{\mathcal{G}^{w'}_{\smash{\sigma'}}}}(w') = \emptyset$, V\ref{V7b}($w,\sigma,r,i$) is satisfied.
\item Let $c(k,y,b) \in X$. Then by V\ref{V7a}($w,\sigma,r,i$), there is some $j$ with $r(k,j) > 0$ and $|c(k,y,b)| > r(k,j)$. By R\ref{R2}($r$) and subsequently Claim \ref{claim:stages2}.\ref{claim:stages2-vii}, if $c(k,y,b) \in X$, then $p_k(p_j(s_i)) < s_{i+1}$. By assumption, $X$ satisfies V\ref{V7e}($w',\sigma',r,i$). So for $c(k,y,b) \in X$, $\hhat{y}_{\smash{F_k^{w' \cup X}}}$ exists with $|\hhat{y}_{\smash{F_k^{w' \cup X}}}| \leq p_j(s_i) < p_k(p_j(s_i) < s_{i+1}$ for above $j$. Since $w^{\smash{<s_{i+1}}} = w'^{\smash{<s_{i+1}}}$, we obtain $\hhat{y}_{\smash{F_k^{w \cup X}}} = \hhat{y}_{\smash{F_k^{w' \cup X}}}$, because $F_k^{\smash{w' \cup X}}(\hhat{y}_{\smash{F_k^{w' \cup X}}})$ can only ask queries of length $p_k(p_j(s_i)) < s_{i+1}$, which are answered the same relative to $w \cup X$. Observe that $\hhat{y}_{\smash{F_k^{w \cup X}}}$ is indeed equal to $\hhat{y}_{\smash{F_k^{w' \cup X}}}$, because for inputs $x < \hhat{y}_{\smash{F_k^{w' \cup X}}}$, $F_k^{\smash{w \cup X}}(x) = F_k^{\smash{w' \cup X}}(x)$ by the same argument that the maximum query length is less than $s_{i+1}$. Hence, $\hhat{y}_{F_k^{w \cup X}}$ exists with $|\hhat{y}_{F_k^{w \cup X}}| \leq p_j(s_i)$. So V\ref{V7e}($w,\sigma,r,i$) is satisfied.
\end{itemize}
Thus, all properties V\ref{V7a}($w,\sigma,r,i$) to V\ref{V7e}($w,\sigma,r,i$) are satisfied for $X$, showing that V\ref{V7}($w,\sigma,r,i$) is violated, a contradiction. 
\medskip
\\
To (\ref{claim:requirement-monotonicity-vii}): Since by V\ref{V8}($w$), $\ls{w'} \geq \ls{w} \geq \max\{r(k) \mid k \in \dom(r)\}$, V\ref{V8}($w'$) is also satisfied.
\end{claimproof}
The requirements R\ref{R3}, V\ref{V1} and V\ref{V6} are not local. For R\ref{R3}($r$), we must make sure that we never extend $r$   via $(k,j') \mapsto a \in \N^+$ when $r(k,j) > 0$ with $j \neq j'$ already holds. For V\ref{V1}($w,\sigma$), whenever we add words to $w$, we must make sure that $\sigma$ remains a proof selector relative to $w$, which depends on many different computations. The dependency graph $\mathcal{G}_{\sigma}^w$ captures these dependencies, thus helping to extend $w$ only with suitable words. 

In contrast, V\ref{V6}($w,r$) is much more difficult to keep satisfied. Intuitively, whenever an $F_k$-proof for some word $y$ emerges, V\ref{V6} requires a shortest code word for $k$ and $y$ in the oracle. The problem is that an $F_k$-proof for $y$ can emerge when the oracle is already defined far beyond the stage containing shortest code words for $k$ and $y$. Adding such a shortest code word afterwards would ruin an arbitrary large part of the already defined oracle, because many longer code words may depend on such a short word. This prevents us from making progress in the oracle construction.\footnote{Usually, these setbacks are only problematic if they set us to the same stage infinitely often. Here, one could argue that on every stage we can add only finitely many words and every setback is repaired by adding an additional word. Hence, only finitely many setbacks per stage should be possible. However, V\ref{V4} allows the distance between a code word and its associated proof to be more than one stage (allowing this distance is necessary for V\ref{V5}). This means that a setback to stage $s_i$ can invalidate code words of length $s_i$, namely those whose associated proofs have length $\approx s_{i+1}$ and thus can recognize the setback. Hence, stage $s_i$ is also ruined, so in fact a setback to stage $s_{i-1}$ is necessary. The same argument leads to a setback to stage $s_{i-2}$ and so on. In this way, setbacks can infinitely often ruin the whole oracle.}
As we will show later in a long proof, we can fully prevent this ``delayed'' emergence of $F_k$-proofs by achieving V\ref{V5} when possible and keeping V\ref{V7} satisfied. Greatly simplified, if an $F_k$-proof for $y$ emerges ``late'', we can add a code word $c(k,y,b)$ one stage below this proof. Either the $F_k$-proof remains ``long'' and we achieve V\ref{V5}. Or $F_k$ develops a significantly shorter proof, but then V\ref{V7} could not have been satisfied before, witnessed by $c(k,y,b)$.

The following lemma shows that there is a generic way to extend some valid pair ($w,\sigma $) such that it remains valid.
\begin{lemma}\label{lemma:remain-valid}
Let $(w,\sigma)$ be a valid pair and $X \subseteq C$ be a set of code words. Define $w' \coloneqq w \cup X$ and $\sigma ' \coloneqq \sigma \cup \{(k,y,b) \mapsto \hhat{y}_{F_k^{w'}} \mid c(k,y,b) \in X\}$. If the following statements hold for all $c(k,y,b) \in X$, then $(w',\sigma')$ is valid. 
\begin{romanenumerate}
\item\label{lemma:remain-valid-i} $c(k,y,b) \notin w$.
\item\label{lemma:remain-valid-ii} $c(k,y,b') \notin w'$ for all $b' \in \Sigma^{|b|} \setminus \{b\}$.
\item\label{lemma:remain-valid-iii} $c(k,y,b) \notin N_{\mathcal{G}^w_{\sigma }}(w)$.
\item\label{lemma:remain-valid-iv} $\hhat{y}_{F_k^{w'}}$ exists and $p_k(|\hhat{y}_{F_k^{w'}}|) < s_{i+2}$ where $s_i \coloneqq |c(k,y,b)|$.
\end{romanenumerate}
\end{lemma}
\begin{proof}
To V\ref{V1}($w',\sigma'$): We have to show that (1) $\sigma'$ is well-defined, (2) $\dom(\sigma') = \{(k,y,b) \mid c(k,y,b) \in w'\}$, and (3) $F_k^{w'}(\sigma'(k,y,b)) = y$ for $(k,y,b) \in \dom(\sigma')$.

By V\ref{V1}($w,\sigma$), $\sigma \in \mathcal{S}^w$ and by (\ref{lemma:remain-valid-i}), $X \cap w = \emptyset$. Then Observation \ref{obs:well-defined} gives that $\dom(\sigma) \cap \{(k,y,b) \mid c(k,y,b) \in X\} = \emptyset$. By (\ref{lemma:remain-valid-iv}), for $c(k,y,b) \in X$, $\hhat{y}_{F_k^{w'}}$ exists. Hence, $\sigma'$ is well-defined. For the domain of $\sigma'$ it holds that
\begin{align*}
\dom(\sigma') &= \dom(\sigma) \cup \{(k,y,b) \mid c(k,y,b) \in X\}\\
 &= \{(k,y,b) \mid c(k,y,b) \in w\} \cup \{(k,y,b) \mid c(k,y,b) \in X\}\\
 &= \{(k,y,b) \mid c(k,y,b) \in w \cup X\} = \{(k,y,b) \mid c(k,y,b) \in w'\}.
\end{align*}
If $c(k,y,b) \in X$, then $F_k^{w'}(\sigma '(k,y,b)) = F_k^{w'}(\hhat{y}_{F_k^{w'}}) = y$ by the definition of $\sigma'$. If $c(k,y,b) \in w$, then $\sigma'(k,y,b) = \sigma (k,y,b)$. Together with $\sigma \in \mathcal{S}^w$, we have 
\[F_k^w(\sigma'(k,y,b)) = F_k^w(\sigma (k,y,b))=y.\]
Furthermore, $Q(F_k^w(\sigma(k,y,b))) \subseteq N_{\mathcal{G}^w_\sigma}(w)$ by the definition of the edge set of $\mathcal{G}^w_\sigma$ (Definition \ref{def:dependency-graph}). Since $X \cap N_{\mathcal{G}^w_\sigma}(w) = \emptyset$ (requirement (\ref{lemma:remain-valid-iii})), also 
\[X \cap Q(F_k^w(\sigma(k,y,b))) = X \cap Q(F_k^w(\sigma' (k,y,b))) = \emptyset.\]
By Observation \ref{obs:oracle-questions} it follows that $F_k^{w'}(\sigma'(k,y,b)) = F_k^w(\sigma '(k,y,b)) = y$.
\medskip
\\
to V\ref{V2}($w'$): V\ref{V2}($w'$) is not violated, because all words in $X$ are code words and for all words in $w$ Claim \ref{claim:requirement-monotonicity}.\ref{claim:requirement-monotonicity-ii} holds.
\medskip
\\
To V\ref{V3}($w'$): V\ref{V3}($w'$) is not violated, because the words in $X$ can not be part of a violation by requirement (\ref{lemma:remain-valid-ii}) and any violation of V\ref{V3}($w'$) must involve a word from $X$, because V\ref{V3}($w$) holds. 
\medskip
\\
To V\ref{V4}($w',\sigma'$): V\ref{V4}($w'$) is not violated, because Claim \ref{claim:requirement-monotonicity}.\ref{claim:requirement-monotonicity-iv} rules out a violation involving any $x \in w$ and for $c(k,y,b) \in X$, it holds $p_k(|\sigma' (k,y,b)|) = p_k(|\hhat{y}_{F_k^{w'}}|) < s_{i+2}$ by requirement (\ref{lemma:remain-valid-iv}) and the definition of $\sigma'$.
\end{proof}
We conclude this section by showing that if V\ref{V7} is violated, we can satisfy V\ref{V7} by extending the oracle via the largest set $X$ witnessing the violation.
\begin{lemma}\label{lemma:V8-extension}
Let $r$ and $(w,\sigma)$ be valid and $i \in \N^+$. Let $Y$ be a set of greatest cardinality satisfying all properties of V\ref{V7}($w,\sigma,r,i$). Then V\ref{V7}($w',\sigma',r,i$) is satisfied where $w' \coloneqq w \cup Y$ and $\sigma' \coloneqq \sigma \cup \{(k,y,b) \mapsto \hhat{y}_{F_k^{w'}} \mid c(k,y,b) \in Y\}$.
\end{lemma}
\begin{proof}
Assume V\ref{V7}($w',\sigma',r,i$) is violated, then there is a nonempty set $X$ satisfying all properties of V\ref{V7}($w',\sigma',r,i$). In the following, we show that (1) $\sigma'$ is well-defined, (2) $Y \cup X$ has bigger cardinality than $Y$, and (3) $Y \cup X$ satisfies all properties of V\ref{V7}($w,\sigma,r,i$).

By V\ref{V1}($w,\sigma$), $\sigma \in \mathcal{S}^w$ and by V\ref{V7c}($w,\sigma,r,i$), $Y \cap w = \emptyset$. Then Observation \ref{obs:well-defined} gives that $\dom(\sigma) \cap \{(k,y,b) \mid c(k,y,b) \in Y\} = \emptyset$. By V\ref{V7e}($w,\sigma,r,i$), for $c(k,y,b) \in Y$, $\hhat{y}_{\smash{F_k^{w \cup Y}}} = \hhat{y}_{\smash{F_k^{w'}}}$ exists. Hence, $\sigma'$ is well-defined.

By V\ref{V7c}($w',\sigma',r,i$), $X$ is chosen such that $\emptyset = X \cap w' = X \cap (w \cup Y)$. By $\card{X} > 0$, we have $\card{(Y \cup X)} > \card{Y}$. 

To get a contradictiom to the maximality of $Y$, it suffices to show that $Y \cup X$ satisfies the properties V\ref{V7a}($w,\sigma,r,i$) to V\ref{V7e}($w,\sigma,r,i$).
\medskip
\\
To V\ref{V7a}($w,\sigma,r,i$): This property is independent of $(w,\sigma)$. Since it holds for $Y$ and $X$ separately, it also holds for $Y \cup X$. 
\medskip
\\
To V\ref{V7b}($w,\sigma,r,i$): Since $Y \cap N_{\mathcal{G}^w_\sigma}(w) = \emptyset$, it follows that $w$ and $w'$ agree on $N_{\mathcal{G}_\sigma^w}(w)$. From this, $w' \supseteq w$, and $\sigma' \sqsupseteq \sigma$, we can invoke Claim \ref{claim:graph-monotonicity} to get $N_{\mathcal{G}^w_\sigma}(w) \subseteq N_{\smash{\mathcal{G}^{w'}_{\smash{\sigma'}}}}(w')$. Since $X \cap N_{\smash{\mathcal{G}^{w'}_{\smash{\sigma'}}}}(w') = \emptyset$, also $(Y \cup X) \cap N_{\mathcal{G}^w_\sigma}(w) = \emptyset$.
\medskip
\\
To V\ref{V7c}($w,\sigma,r,i$): $Y \cap w = \emptyset$, $X \cap w' = X \cap (w \cup Y) \supseteq X \cap w = \emptyset$ and hence, $(Y \cup X) \cap w = \emptyset$.
\medskip
\\
To V\ref{V7d}($w,\sigma,r,i$): By V\ref{V7d}($w',\sigma',r,i$) for all $k \in \N^+$ and $y \in \Sigma^*$ it holds that $ \card{(X \cup Y \cup w) \cap C_{k,y}^i} = \card{(X \cup w') \cap C_{k,y}^i} \leq 1$.
\medskip
\\
To V\ref{V7e}($w,\sigma,r,i$): By assumption, V\ref{V7e}($w',\sigma',r,i$) holds for $X$. Note that for $c(k,y,b) \in X$, V\ref{V7e}($w',\sigma',r,i$) for $X$ requires the same properties as V\ref{V7e}($w,\sigma,r,i$) for $Y \cup X$, because $w' \cup X = w \cup Y \cup X$. Hence, for $c(k,y,b) \in X$, V\ref{V7e}($w,\sigma,r,i$) holds for $Y \cup X$.

Consider $c(k,y,b) \in Y$. Since $Y$ satisfies V\ref{V7e}($w,\sigma,r,i$), it holds that $\hhat{y}_{F_k^{w'}}$ exists with $|\hhat{y}_{F_k^{w'}}| \leq p_j(s_i)$ for $j$ with $r(k,j) > 0$. The definition of $\sigma'$ gives $\sigma'(k,y,b) = \hhat{y}_{F_k^{w'}}$. By the definition of dependency graphs (Definition \ref{def:dependency-graph}) it holds that $Q(F_k^{w'}(\sigma'(k,y,b))) \subseteq N_{\smash{\mathcal{G}^{w'}_{\smash{\sigma'}}}}(w')$. Finally, by V\ref{V7b}($w',\sigma',r,i$) for $X$, we have that $X \cap Q(F_k^{w'}(\sigma'(k,y,b))) \subseteq X \cap N_{\smash{\mathcal{G}^{w'}_{\smash{\sigma'}}}}(w') = \emptyset$. Hence, we get 
\[F_k^{w\cup Y \cup X}(\hhat{y}_{F_k^{w'}}) = F_k^{w \cup Y \cup X}(\sigma'(k,y,b)) = F_k^{w' \cup X}(\sigma'(k,y,b)) \overset{\text{\ref{obs:oracle-questions}}}{=} F_k^{w'}(\sigma'(k,y,b)) = y.\]
and therefore $\hhat{y}_{F_k^{w\cup Y \cup X}}$ exists with
\[|\hhat{y}_{F_k^{w\cup Y \cup X}}| \leq |\hhat{y}_{F_k^{w'}}| \leq p_j(s_i) \text{ for } j \text{ with } r(k,j)>0.\]
Thus, for $c(k,y,b) \in Y$, V\ref{V7e}($w,\sigma,r,i$) holds for $Y \cup X$.
\end{proof}

\subsection{Oracle Construction}\label{sec:oracle}
We want to construct an oracle which satisfies the following theorem:

\begin{theorem}\label{thm:thick-oracle-1}
There exists an oracle $O$ such that for all $k \in \N^+$ one of the following holds:
\begin{romanenumerate}
\item $F_k^O$ is not optimal for $\ran(F_k^O)$.
\item $\ran(F_k^O) \in \UTIME^O (h^c)$ for a constant $c \in \N$.
\end{romanenumerate}
\end{theorem}

\noindent We will take care of the following set of tasks:
\[\{\tau _k \mid k \in \N^+\} \cup \{\tau _{k,j} \mid k,j \in \N^+\}\]
Let $\mathcal{T}$ be an enumeration of these tasks with the property that $\tau _k$ appears earlier than $\tau _{k,j}$.

Recall that $(w,\sigma ,r) \in \Sigma^* \times \mathcal{S} \times \mathcal{R}$ is a valid triple, when $r$ is valid and ($w,\sigma $) is $r$-valid. The oracle construction inductively defines a sequence $\{(w_s,\sigma_s,r_s)\}_{s\in \N}$ of valid triples. The $s$-th triple is defined in step $s$ of the oracle construction.

In every step we treat the smallest task in the order specified by $\mathcal{T}$, and after treating a task we remove it and possibly other higher tasks from $\mathcal{T}$. (In every step, there exists a task that we will treat, as we never remove \emph{all} remaining tasks from $\mathcal{T}$.)

In step $s=0$ we define $(w_0,\sigma_0,r_0)$ with the nowhere defined function $r_0 \in \mathcal{R}$, the nowhere defined function $s_0 \in \mathcal{S}^\emptyset$, and the oracle $w_0 = \varepsilon$. We will show that $(w_0,\sigma_0,r_0)$ is a valid triple.

In step $s>0$ we define $(w_s,\sigma_s,r_s)$ such that $w_{s} \sqsupsetneq w_{s-1}$, $\sigma _{s} \sqsupseteq \sigma _{s-1}$, $r_{s} \sqsupsetneq r_{s-1}$, $(w_s,\sigma_s,r_s)$ is valid, and the earliest task $\tau $ in $\mathcal{T}$ is treated and removed. The oracle construction  depends on whether $\tau _k$ or $\tau _{k,j}$ is the next task in $\mathcal{T}$.

\begin{itemize}
\item {\bf Task} $\tau _k$: Define $r_s(k) \coloneqq m$ with $m \geq s_{i+\floor{\log i}} + k + 2$ for $i \in \N$ with $\floor{\log i} > 2k+2$ and $i \geq i_0(k) + \hat{m}$. Extend $(w_{s-1},\sigma_{s-1})$ to $(v,\rho)$ via Lemma \ref{lemma:valid-extension} such that $(v,\rho) \sqsupsetneq (w_{s-1}, \sigma_{s-1})$ and $\ls{v} \geq m$. Define $w_s \coloneqq v$, $\sigma_s \coloneqq \rho$.

{(Meaning: We make sure that before coding for or diagonalizing against $F_k$, the combinatorial properties of R\ref{R1} hold for $\ls{w_s}$.)}

\item {\bf Task} $\tau _{k,j}$: Let $r^0 \coloneqq r_{s-1} \cup \{(k,j) \mapsto 0\}$. If there exists an $r^0$-valid pair ($v,\rho$) such that $v \sqsupsetneq w_{s-1}$ and $\rho \sqsupseteq \sigma_{s-1}$, then define $r_s \coloneqq r^0$, $w_s \coloneqq v$ and $\sigma _s \coloneqq \rho$.
 
Otherwise, extend ($w_{s-1},\sigma_{s-1}$) to ($v, \rho$) via Lemma \ref{lemma:valid-extension} such that $(v,\rho) \sqsupsetneq (w_{s-1}, \sigma_{s-1})$. Define $w_s \coloneqq v$ and $\sigma _s \coloneqq \rho$. Define $r_s \coloneqq r_{s-1} \cup \{(k,j) \mapsto m\}$ with $m \in \N^+$ sufficiently large such that $m > \ls{v} + \max \ran(r_{s-1}) +  s_{i_0(k \cdot j)}$. Remove all tasks $\tau _{k,j+1}, \tau_{k,j+2},\dots$ from $\mathcal{T}$.

{(Meaning: Try to ensure that $F_k$-proofs for a fact $y$ are more than $p_j$ longer compared to proofs based on code words in the oracle (i.e., $r_s(k,j)=0$ and satisfying V\ref{V5}). If this is not possible (i.e., $r_s(k,j) > 0$), intuitively, $F_k$ always has ``short'' proofs for a fact when encoding the fact by a code word. Thus, we promise to encode $F_k$-facts of length $>r_s(k,j)$ into the oracle via shortest code words (satisfying V\ref{V6}). We do this by encoding as much as possible on each stage (satisfying V\ref{V7}).)}
\end{itemize}

The remaining part of this section proves that $(w_s,\sigma_s,r_s)_{s\in\N}$ is a sequence of valid triples. We start by showing that $r_s$ is a valid requirement function for all $s \in \N$.
\begin{claim}\label{claim:r-always-valid}
$r_0 \sqsubseteq r_1 \sqsubseteq \dots$ and for every $s \in \N$, $r_s \in \mathcal{R}$ is valid.
\end{claim}
\begin{claimproof}
The oracle construction considers every task at most once, because every task is removed from $\mathcal{T}$ after it was treated. Recall that $r_0 \in \mathcal{R}$ is the nowhere defined function. In step $s+1$, the domain of the requirement function $r_s$ for $s \in \N$ is extended only by the index of the treated task. This index is an element of $\N^+ \cup (\N^+ \times \N^+)$, is mapped to some element from $\N$ and can not already be contained in the domain of $r_{s}$, because otherwise the task would have been treated twice. Hence, $r_{s} \sqsubseteq r_{s+1} \in \mathcal{R}$.

The requirements R\ref{R1}($r_0$), R\ref{R2}($r_0$), and R\ref{R3}($r_0$) are trivially satisfied. Hence, $r_0$ is valid. Let $s \in \N$ and $r_{s}$ be valid.

By Claim \ref{claim:requirement-monotonicity}.\ref{claim:requirement-monotonicity-i}, $r_{s+1}$ does not violate R\ref{R1}($r_{s+1}$) and R\ref{R2}($r_{s+1}$) via some $k \in \dom(r_{s})$ and $(k,j) \in \dom(r_{s})$. Suppose the task $\tau _k$ is treated in step $s+1$. Then $r_{s+1}(k) = m$ satisfies all properties of R\ref{R1}($r_{s+1}$), because $m$ is chosen exactly as required by R\ref{R1}($r_{s+1}$). Suppose the task $\tau_{k,j}$ is treated in step $s+1$. If $r_{s+1}(k,j)=0$, then R\ref{R2}($r_{s+1}$) holds. Otherwise, the task $\tau_k$ appears before $\tau_{k,j}$ in $\mathcal{T}$, so $r_s(k)$ is defined and $r_s(k) \leq \max \ran(r_{s})$. Then $r_{s+1}(k,j)=m$ with $m > \max \ran (r_{s}) + s_{i_0(k \cdot j)} \geq r(k) + s_{i_0(k \cdot j)}$ as required by R\ref{R2}($r_{s+1}$).

Since R\ref{R3}($r_{s}$) is satisfied, only the treatment of the task $\tau_{k,j}$ with $r(k,j)>0$ in step $s+1$ can lead to a violation of R\ref{R3}($r_{s+1}$). But if there is some prior task $\tau_{k,j'}$ where $r(k,j') > 0$ is chosen, then the task $\tau_{k,j}$ would have been removed from $\mathcal{T}$ and $\tau_{k,j}$ would have never been treated. Hence, R\ref{R3}($r_{s+1}$) is satisfied.
\end{claimproof}
Next we will prove an argument that will help us to show that $(w_s,\sigma_s)$ is $r_s$-valid. The oracle construction defines $(w_s,\sigma_s)$ as an extension of $(w_{s-1},\sigma_{s-1})$ via Lemma \ref{lemma:valid-extension}. Lemma \ref{lemma:valid-extension} will give that $(w_s,\sigma_s)$ is $r_{s-1}$-valid. The following claim shows that this implies that $(w_s,\sigma_s)$ is also $r_s$-valid.
\begin{claim}\label{claim:valid-stepup}
If $(w_s,\sigma_s)$ is $r_{s-1}$-valid, then it is also $r_s$-valid.
\end{claim}
\begin{claimproof}
Since $(w_s,\sigma_s)$ is $r_{s-1}$-valid, if $(w_s,\sigma_s$) is not $r_s$-valid, there must be a violated requirement that involves $r_s(k)$ or $r_s(k,j)$ defined in step $s$. 

If the task $\tau_k$ is treated in step $s$, then $r_s(k)$ was defined and hence, only V\ref{V8}($w_s,r_s$) can be violated (V\ref{V1}--V\ref{V7} do not depend on $r$ or only on values $r(k',j')$). Since $(w_s,\sigma_s)$ is $r_{s-1}$-valid, it holds that $\ls{w_s} \geq \max\{r_{s-1}(k') \mid k' \in \dom(r_{s-1})\}$. Furthermore, when treating the task $\tau_k$, we ensure that $\ls{w_s} \geq r_s(k)$. Hence, $(w_s,\sigma_s)$ is $r_s$-valid.

If the task $\tau_{k,j}$ is treated in step $s$, we distinguish two cases. If $r_s(k,j) = 0$, then $(w_s,\sigma_s)$ was chosen to be $r_s$-valid. Otherwise, $m \coloneqq r_s(k,j) > \ls{w_s}$. Since V\ref{V6}, V\ref{V7a} and V\ref{V7e} refer to $r_s(k,j) > 0$, we analyze them:
\begin{itemize}
\item The requirement V\ref{V6}($w_s,r_s$) can only be violated for $r_s(k,j)$ when $\ls{w_s} \geq m$, because V\ref{V6} requires $s_{i(k,y)} \leq \ls{w_s}$ with $y \in \Sigma^{>m}$. Since $m > \ls{w_s}$, V\ref{V6}($w_s,r_s$) can not be violated. 
\item Similarly, for V\ref{V7a}($w_s,\sigma_s,r_s$), the requirement-function $r_s$ can only make a difference to $r_{s-1}$ when $\ls{w_s} \geq m$, which is not the case. 
\item Since V\ref{V7a}($w_s,\sigma_s,r_{s-1}$) is equal to V\ref{V7a}($w_s,\sigma_s,r_s$), any set $X$ that satisfies the latter also satisfies the former. Furthermore, any set $X$ that satisfies V\ref{V7b}, V\ref{V7c}, and V\ref{V7d} for ($w_s,\sigma_s,r_s$), also satisfies these properties for ($w_s,\sigma_s,r_{s-1}$), because $r_s$ does not appear in these requirements. 

Let $X$ be an arbitrary set satisfying V\ref{V7}($w_s,\sigma_s,r_s$). Then $X$ also satisfies V\ref{V7a}--V\ref{V7d}($w_s,\sigma_s,r_{s-1}$). By V\ref{V7a}($w_s,\sigma_s,r_{s-1}$), for $c(k,y,b) \in X$ there is some $j$ with $r_{s-1}(k,j) > 0$. Since R\ref{R3}($r_s$) holds, there is no $j' \neq j$ with $r_s(k,j') > 0$. Thus, for V\ref{V7e}($w_s,\sigma_s,r_{s-1}$), a transition to $r_s$ does not change the requirement. Hence, V\ref{V7e}($w_s,\sigma_s,r_s$) is also equal to V\ref{V7e}($w_s,\sigma_s,r_{s-1}$). Therefore, V\ref{V7}($w_s,\sigma_s,r_s$) is satisfied, because it is equivalent to V\ref{V7}($w_s,\sigma_s,r_{s-1}$), which is satisfied.
\end{itemize}
This shows that V\ref{V6}($w_s,r_s$) and V\ref{V7}($w_s,\sigma_s,r_s$) are satisfied and $(w_s,\sigma_s)$ is $r_s$-valid. 
\end{claimproof}
Finally, we show that Lemma \ref{lemma:valid-extension} extends an $r_s$-valid pair by one stage such that the pair remains $r_s$-valid. This, together with the previous two claims, shows that the oracle construction gives a sequence of valid triples. The following proof is quite long and technical, because, as already mentioned, V\ref{V6} is a property that is difficult to preserve.
\begin{lemma}\label{lemma:valid-extension}
Let $(w_0,\sigma_0,r_0),(w_1,\sigma_1,r_1),\dots ,(w_{s},\sigma_s,r_{s})$ be valid triples from the oracle construction. Let $(w,\sigma) \sqsupseteq (w_s,\sigma_s)$ be $r_s$-valid and $\ls{w} = s_{i-1}$. Then there is some $r_s$-valid pair $(v,\rho)$ with $v \sqsupsetneq w$ and $\rho \sqsupseteq \sigma$. More precisely, $(v,\rho)$ is defined according to Algorithm \ref{alg:2} which basically adds words according to V\ref{V7}:

\begin{algorithm}
\caption{prod}\label{alg:2}
\begin{algorithmic}[1]
\State \textbf{Input:} $(w,\sigma,r,d) \in \Sigma^* \times \mathcal{S} \times \mathcal{R} \times \N$\label{alg2:1}
\State let $v_{d-1} \coloneqq w$ and $\rho_{d-1} \coloneqq \sigma$\label{alg2:2}
\For{$\ell$ from $d-1$ to $i-1$ inclusive}\label{alg2:3}
\State define $\mathcal{X}_\ell$ such that $X \in \mathcal{X}_\ell$ if \label{alg2:4}
\Indent
\State $X \subseteq \{c(k,y,b) \in C_{\smash{k,y}}^{\smash{\ell+1}} \mid \exists j \mbox{ with } r(k,j)=m>0,\ k \in \N^+,\ y \in \Sigma^{>m}\}$,\label{alg2:5} 
\State $X \cap N_{\smash{\mathcal{G}_{\smash{\rho_\ell}}^{\smash{v_\ell}}}}(v_\ell) = \emptyset$,\label{alg2:6}
\State $X \cap v_\ell = \emptyset$,\label{alg2:7}
\State $\card{((X \cup v_\ell) \cap C_{k,y}^{\ell+1}) } \leq 1$ for $k \in \N^+$ and $y \in \Sigma^*$,\label{alg2:8}
\State if $c(k,y,b) \in X$, then $\hhat{y}_{\smash{F_{\smash{k}}^{v_\ell \cup X}}}$ exists with $|\hhat{y}_{\smash{F_{\smash{k}}^{v_\ell \cup X}}}| \leq p_j(s_{\ell+1})$ for $j$ with $r(k,j)>0$\label{alg2:9}
\EndIndent
\State let $Y_\ell \in \mathcal{X}_\ell$ be of greatest cardinality\label{alg2:10}
\State define $v_{\ell+ 1} \coloneqq v_\ell \cup Y_\ell$\label{alg2:11}
\State define $\rho_{\ell+1} \coloneqq \rho_\ell \cup \{(k,y,b) \mapsto \hhat{y}_{\smash{F_k^{\smash{v_{\ell+1}}}}} \mid c(k,y,b) \in Y_\ell\}$\label{alg2:12}
\EndFor\label{alg2:13}
\State \Return $(v_{i},\rho_{i})$\label{alg2:14}
\end{algorithmic}
\end{algorithm}
\end{lemma}
\begin{proof}
Let $(v,\rho) \coloneqq \text{prod}(w,\sigma,r_s,i)$, i.e., the loop of $\text{prod}(w,\sigma,r_s,i)$ is executed one time, and let $Y \coloneqq Y_{i-1}$ where $Y_{i-1}$ is the set defined in line \ref{alg2:10} in the only iteration of the loop. By Observation \ref{obs:well-defined}, $\sigma \in \mathcal{S}^w$ (follows from V\ref{V1}($w,\sigma$)), line \ref{alg2:7} and line \ref{alg2:9}, $\rho$ is well-defined. It holds that $v \sqsupsetneq w$, because $w$ gets extended by words of length $s_i$, $\rho \sqsupseteq \sigma$ and $Y \subseteq C$ is a set of code words. We show via several claims that all requirements V\ref{V1} to V\ref{V8} are satisfied for $(v,\rho,r_s)$. All claims have short proofs, except the claim for V\ref{V6}($v,r_s$).
\begin{claim}\label{claim:small-polys}
If $c(k,y,b) \in Y$, then $p_k(p_j(s_i)) < s_{i+1}$ for the unique $j$ (cf.~R\ref{R3}($r_s$)) such that $r_s(k,j) > 0$.
\end{claim}
\begin{claimproof}
Let $r_s(k,j) = m$. Observe that $s_i = |c(k,y,b)| > m$, because $y \in \Sigma^{>m}$ (line \ref{alg2:5}). Since $r_s$ is valid and $m > 0$, by R\ref{R2}($r_s$) and Claim \ref{claim:stages2}.\ref{claim:stages2-vii}, $p_{k\cdot j}(s_i) = p_k(p_j(s_i)) < s_{i+1}$.
\end{claimproof}
\begin{observation}\label{obs:w-v-agree2}
$v^{<s_i} = w^{<s_i}$. 
\end{observation}
\begin{observation}\label{obs:w-ws-agree}
The oracles $v$ and $w$ agree on $N_{\mathcal{G}^{w}_{\sigma}}(w)$ by line \ref{alg2:6}.
\end{observation}
\begin{observation}\label{obs:graph-mono-sinlge}
$N_{\mathcal{G}^{v}_{\rho}}(v) \supseteq N_{\mathcal{G}^{w}_{\sigma}}(w)$ by Claim \ref{claim:graph-monotonicity}.
\end{observation}
\begin{claim}\label{claim:V2-V3-V5-V6}
The requirements V\ref{V1}($v,\rho$), V\ref{V2}($v$), V\ref{V3}($v$), and V\ref{V4}($v,\rho$) are satisfied.
\end{claim}
\begin{claimproof}
We have that $(w,\sigma)$ is valid and by Lemma \ref{lemma:remain-valid} invoked with $(w,\sigma)$ and $Y$, also $(v,\rho)$ is valid. Observe that \ref{lemma:remain-valid}.\ref{lemma:remain-valid-i} holds by line \ref{alg2:7}, \ref{lemma:remain-valid}.\ref{lemma:remain-valid-ii} holds by line \ref{alg2:8}, \ref{lemma:remain-valid}.\ref{lemma:remain-valid-iii} holds by line \ref{alg2:6}, and \ref{lemma:remain-valid}.\ref{lemma:remain-valid-iv} holds by Claim \ref{claim:small-polys} and line \ref{alg2:9}.
\end{claimproof}
\begin{claim}\label{claim:V5-V8}
The requirements V\ref{V5}($v,r_s$) and V\ref{V8}($v,r_s$) are satisfied.
\end{claim}
\begin{claimproof}
This holds by Claim \ref{claim:requirement-monotonicity}.\ref{claim:requirement-monotonicity-v} and \ref{claim:requirement-monotonicity}.\ref{claim:requirement-monotonicity-vii}.
\end{claimproof}
\begin{claim}\label{claim:V7}
The requirement V\ref{V7}($v,\rho,r_s$) is satisfied.
\end{claim}
\begin{claimproof}
V\ref{V7}($w,\sigma,r_s$) is satisfied (i.e., V\ref{V7}($w,\sigma,r_s,1$), \dots , V\ref{V7}($w,\sigma,r_s,i-1$) are satisfied), $(v,\rho) \supseteq (w,\sigma)$, $N_{\mathcal{G}^v_\rho}(v) \supseteq N_{\mathcal{G}^w_\sigma}(w)$ (Observation \ref{obs:graph-mono-sinlge}), and $v^{<s_i} = w^{<s_i}$ (Observation \ref{obs:w-v-agree2}). Then by Claim \ref{claim:requirement-monotonicity}.\ref{claim:requirement-monotonicity-vi}, V\ref{V7}($v,\rho,r_s,1$), \dots, V\ref{V7}($v,\rho,r_s,i-1$) are satisfied. It remains to show that V\ref{V7}($v,\rho,r_s,i$) is satisfied.

Observe that $Y$ is of greatest cardinality satisfying all properties of V\ref{V7}($w,\sigma,r_s,i$), because the lines \ref{alg2:5} to \ref{alg2:9} on input prod($w,\sigma,r_s,i$) are identical to the requirements V\ref{V7a}($w,\sigma,r_s,i$) to V\ref{V7e}($w,\sigma,r_s,i$).
Furthermore, $(w,\sigma)$ is valid and extended to $(v,\rho)$ via $Y$ exactly as described in Lemma \ref{lemma:V8-extension}. Then Lemma \ref{lemma:V8-extension} gives that V\ref{V7}($v,\rho,r_s,i$) is satisfied.
\end{claimproof}

The remainder of this proof focuses on the requirement V\ref{V6}($v,r_s$). We will use a proof by contradiction and assume that V\ref{V6}($v,r_s$) is violated for some tuple $(k,j)$. Under this assumption, we can extend $(v,\rho)$ to ($v',\rho'$), such that the oracle construction should have chosen this extension through which $r_{s_{k,j}}(k,j) = 0$ would have been possible, with $s_{k,j}$ being the step that treated the task $\tau_{k,j}$. However, when proving the $r_{s_{k,j}-1}$-validity of $(v',\rho')$, we have to prove again that V\ref{V6}($v',r_{s_{k,j}-1}$) is satisfied. Hence, we have to prove recursively that V\ref{V6} is satisfied. To properly state the invariants for this recursion, we need further definitions.

For $(k,j) \in \supp(r_s)$, let $s_{k,j}$ be the step in the oracle construction that treated the task $\tau _{k,j}$. Let $\varphi \colon \N^2 \cap \supp(r_s) \to \{1, \dots , \card{\N^2 \cap \supp(r_s)}\}$ be defined as the function that maps $(k,j)$ to the position\footnote{The positions start at $1$.} of $s_{k,j}$ in the list $[s_{k,j} \mid (k,j) \in \supp(r_s)]$ that is sorted in descending order. 
\begin{claim}\label{claim:wskj-defined-for-stage}
Let $(k,j) \in \supp(r_s)$. Then $\ls{w_{s_{k,j}}} \leq s_{i-\varphi(k,j)}$.
\end{claim}
\begin{claimproof}
By assumption (cf.~Lemma \ref{lemma:valid-extension}), $w \sqsupseteq w_s$ is defined up to stage $s_{i-1}$. There are at least $\varphi(k,j)-1$ steps of the oracle construction from step $s_{k,j}$ to step $s$. Hence, $w_{s_{k,j}}$ is defined for at least $\varphi(k,j)-1$ stages less than $w_s$, because for each task, the oracle gets extended by at least one full stage. Therefore, $\ls{w_{s_{k,j}}} \leq s_{i-1-(\varphi(k,j)-1)} = s_{i-\varphi(k,j)}$.
\end{claimproof}
The proof of the following Claim \ref{claim:V9-recursive} is quite long. We prove that we can either resolve the violation of V\ref{V6} or we can propagate this defect onto a previously treated tuple $(k',j')$. Since there are only finitely many tuples treated before $(k,j)$, we can repeatedly invoke this claim and eventually resolve the violation of V\ref{V6}. If the reader wants to finish the proof of Lemma \ref{lemma:valid-extension} first, it continues in the paragraph before Claim \ref{claim:V6}.
\begin{claim}\label{claim:V9-recursive}
Let $(k,j) \in \dom (\varphi)$ and $c \coloneqq \varphi(k,j)$. Let $v \supseteq w$ with $\ls{v} = s_{i}$ and $v$ and $w$ agree on $\Sigma^{\leq s_{i-c}} \cup N_{\mathcal{G}^{w}_{\sigma}}(w)$. Let $\rho \sqsupseteq \sigma$. Let $(v,\rho)$ be $r_{s_{k,j}}$-valid with the exception of V\ref{V6}($v,r_{s_{k,j}}$) for $(k,j)$.

Then there is some $v' \supseteq v$ with $\ls{v'} = s_i$, $v'$ and $w$ agree on $\Sigma^{\leq s_{i-c-1}} \cup N_{\mathcal{G}^{w}_{\sigma}}(w)$, and some $\rho' \sqsupseteq \rho$ such that $(v',\rho')$ is $r_{s_{k,j}-1} \cup \{(k,j) \mapsto 0\}$-valid except V\ref{V6}($v',r_{s_{k,j}-1} \cup \{(k,j) \mapsto 0\}$), which may be violated.
\end{claim}
\emph{Proof sketch:} If V\ref{V6}($v,\smash{r_{\smash{s_{k,j}}}}$) is violated for $(k,j)$, then we have an $F_k$-proof $\tilde{y}$ for some $y$, but $C^{\smash{i(k,y)}}_{k,y} \cap v = \emptyset$. So we somehow missed to code for $y$. We consider the following two possibilities: either $\tilde{y}$ is short or long compared to $i(k,y)$.

If $\tilde{y}$ is short, we get a contradiction to the satisfaction of V\ref{V7}, because V\ref{V7} requires that we encode for all words that have a short $F_k$-proof. This even holds if $\tilde{y}$ is long and a short $F_k$-proof arises after encoding.

The remaining case deals with a long $\tilde{y}$ and the shortest $F_k$-proof remains long, even after adding encodings for $k$ and $y$. Here, we can use the gap between $i(k,y)$ and a shortest $F_k$-proof to construct an oracle that is $r_{s_{k,j}-1} \cup \{(k,j) \mapsto 0\}$-valid (except V\ref{V6}). In particular, this oracle satisfies V\ref{V5} for $(k,j)$. The proof is therefore as follows:
\begin{enumerate}
    \item Show that the shortest $F_k^v$-proof $\tilde{y}$ for $y$ is long (Claims \ref{claim:small-polynomial}, \ref{claim:smallest-code-word-length-small}, \ref{claim:l<i}, \ref{claim:bounded-haty}, \ref{claim:short-cw-long-proof}).
    \item Show that there is a significantly shorter code word $c(k,y,b)$ than $\tilde{y}$ that we can add to the oracle without conflicting the validity of the oracle (Claims \ref{claim:noduplicatecw}, \ref{claim:cw-choice}, \ref{claim:v-rho-valid}).
    \item Construct an $r_{s_{k,j}-1} \cup \{(k,j) \mapsto 0\}$-valid (except V\ref{V6}) oracle $v'$ containing $c(k,y,b)$ and having some additional properties stated in Claim \ref{claim:V9-recursive} (Observation \ref{obs:prod1-iterative}, Claims \ref{claim:wc-graphvalid}, \ref{claim:wc-V9-valid}, \ref{claim:wc-V6-valid}, \ref{claim:wc-V8-valid}).
\end{enumerate}

\begin{claimproof}
We introduce the abbreviations $r \coloneqq r_{s_{k,j}}$ and $r^0 \coloneqq r_{s_{k,j}-1} \cup \{(k,j) \mapsto 0\}$. 
Since $(v,\rho)$ is $r_{s_{k,j}}$-valid except V\ref{V6}($v,r_{s_{k,j}}$) for $(k,j)$, by Claim \ref{claim:hereditary-r}, it holds that $(v,\rho)$ is $r_{s_{k,j}-1}$-valid.
\begin{observation}\label{obs:vrho-rskj-valid}
$(v,\rho)$ is $r_{s_{k,j}-1}$-valid.
\end{observation}
From the violation of V\ref{V6}($v,r$) for $(k,j)$ it follows that $m \coloneqq r(k,j) > 0$ and there is some $y$ such that 
\begin{romanenumerate}
\item\label{claim:supposition-i} $y \in \Sigma^{>m}$,
\item\label{claim:supposition-ii} $s_{i(k,y)} \leq s_i$,
\item\label{claim:supposition-iii} $\hhat{y}_{F_k^v}$ exists and $|\hhat{y}_{F_k^v}| \leq p_j(s_i)$,
\item\label{claim:supposition-iv} and $C_{k,y}^{\smash{i(k,y)}} \cap v = \emptyset$.
\end{romanenumerate}
From now on, we fix $y$ as the lexicographic smallest word having these properties. We define $\tilde{y} \coloneqq \hhat{y}_{F_k^v}$, i.e., $\tilde{y}$ is the smallest $F_k^v$-proof for $y$.
\subparagraph{1. The shortest proof for y is long.}
The following claims show that $\tilde{y}$ is rather long and that code words $c(k,y,b)$ that are significantly shorter than $\tilde{y}$ exist. 
\begin{claim}\label{claim:small-polynomial}
For $\ell \geq i(k,y)$ it holds that $p_k(p_j(s_\ell)) < s_{\ell+1}$.
\end{claim}
\begin{claimproof}
We have $r(k,j) = m < |y| < s_{i(k,y)} \leq s_{\ell}$ and $r$ is valid. Together with R\ref{R2}($r$) and subsequently Claim \ref{claim:stages2}.\ref{claim:stages2-vii}, $p_k(p_j(s_{\ell})) < s_{\ell+1}$.
\end{claimproof}

\begin{claim}\label{claim:smallest-code-word-length-small}
If $i(k,y) \geq i-c$, then $p_j(s_{i(k,y)}) < |\tilde{y}|$.
\end{claim}
\begin{claimproof}
Suppose $p_j(s_{i(k,y)}) \geq |\tilde{y}|$. We show that in this case V\ref{V7}($v,\rho,r,i(k,y)$) is not satisfied. In particular, $X = \{c(k,y,b)\}$ for some specific $b \in \Sigma^{t_{i(k,y) + \floor{\log (i(k,y))}}}$ (i.e., $|c(k,y,b)| = i(k,y)$) satisfies all requirements of V\ref{V7}($v,\rho,r,i(k,y)$). This is a contradiction to the assumption of Claim \ref{claim:V9-recursive} that V\ref{V7}($v,\rho,r,i(k,y)$) is satisfied. 

First, we bound the number of dependencies $F_k^v(\tilde{y})$ has, i.e.,
\[\card{Q(F_k^v(\tilde{y}))} \overset{\text{(\ref{claim:supposition-iii})}}{\leq} p_k(p_j(s_i)) \overset{\text{(\ref{claim:supposition-ii}),\ref{claim:small-polynomial}}}{<} s_{i+1}.\]
Next, we bound the number of dependencies the oracle has. Observe that $C_{k,y}^{i(k,y)} \not = \emptyset$, $(v,\rho)$ is valid, V\ref{V8}($v,r_s$) is satisfied, and $\ls{v} = s_i$. Then we can apply Lemma \ref{lemma:graph-structure} with $i$ for $\ell$, $r_s$ for $r$, $v$ for $w$, $\rho$ for $\sigma$, and 
\[i(k,y) \geq i - c \geq i - \card{\N^2 \cap \supp(r_s)} > i - 2\max\{k' \mid k' \in \dom(r_s)\}\] 
for $j$ to obtain $\card{C_{k,y}^{\smash{i(k,y)}} \setminus N_{\mathcal{G}^v_{\rho}}(v)} \geq s_{i+2}$.  Thus, there exists some $c(k,y,b) \in (C_{k,y}^{\smash{i(k,y)}} \setminus N_{\mathcal{G}^v_{\rho}}(v)) \setminus Q(F_k^v(\tilde{y}))$. 

Let $X \coloneqq \{c(k,y,b)\}$. Then $F_k^v(\tilde{y}) = F_k^{v\cup X}(\tilde{y}) = y$ by Observation \ref{obs:oracle-questions}. So $\hhat{y}_{\smash{F_k^{v \cup X}}}$ exists with $|\hhat{y}_{\smash{F_k^{v \cup X}}}| \leq |\tilde{y}| \leq p_j(s_i)$ by (\ref{claim:supposition-iii}), which shows that V\ref{V7e}($v,\rho,r,i(k,y)$) is satisfied. Further V\ref{V7a}($v,\rho,r,i(k,y)$) is satisfied by (\ref{claim:supposition-i}) and the choice of $c(k,y,b)$, V\ref{V7b}($v,\rho,r,i(k,y)$) is satisfied by the choice of $c(k,y,b)$, V\ref{V7c}($v,\rho,r,i(k,y)$) is satisfied by (\ref{claim:supposition-iv}), and V\ref{V7d}($v,\rho,r,i(k,y)$) is satisfied by (\ref{claim:supposition-iv}) and that V\ref{V3}($v$) is satisfied.
Hence, $X$ satisfies all requirements of V\ref{V7}($v,\rho,r,i(k,y)$) and thus V\ref{V7}($v,\rho,r,i(k,y)$) is not satisfied, giving the desired contradiction.
\end{claimproof}
\begin{claim}\label{claim:l<i}
$i(k,y) < i$.
\end{claim}
\begin{claimproof}
By (\ref{claim:supposition-iii}), $|\tilde{y}| \leq p_j(s_i)$. If $i(k,y) < i-c$, then $i(k,y) < i$. Otherwise, by Claim \ref{claim:smallest-code-word-length-small}, $p_j(s_{i(k,y)}) < |\tilde{y}| \leq p_j(s_i)$. Hence, $i(k,y) < i$.
\end{claimproof}
\begin{claim}\label{claim:bounded-haty}
If $i(k,y) < i-c$, then $p_j(s_{i-c}) < |\tilde{y}|$.
\end{claim}
\begin{claimproof}
Assume $p_j(s_{i-c}) \geq |\tilde{y}|$, then $F_k^{v}(\tilde{y}) = F_k^{w}(\tilde{y})=y$, because $v$ and $w$ agree on $\Sigma^{< s_{i-c+1}}$ and their running time, bounded by $p_k(p_j(s_{i-c}))$, does not suffice to ask words of length $\geq s_{i-c+1}$ by Claim \ref{claim:small-polynomial}. But then V\ref{V6}($w,r_s$) is violated for $(k,j)$, because $r_s(k,j)=r(k,j)=m>0$, $y \in \Sigma^{>m}$ by (\ref{claim:supposition-i}), $\hhat{y}_{F_k^{w}}$ exists, $s_{i(k,y)} < s_{i-c} \leq s_{i-1} = \ls{w}$, and $|\hhat{y}_{F_k^{w}}| \leq |\tilde{y}| \leq p_j(s_{i-c})$ by assumption, but $C_{k,y}^{\smash{i(k,y)}} \cap w = \emptyset$ by (\ref{claim:supposition-iv}) and $w \subseteq v$. Hence, V\ref{V6}($w,r_s$) is violated, a contradiction that $(w,\sigma)$ is $r_s$-valid as assumed in Lemma \ref{lemma:valid-extension}.
\end{claimproof}
\begin{claim}\label{claim:short-cw-long-proof}
$\max\{p_j(s_{i(k,y)}),p_j(s_{i-c)}\} < |\tilde{y}|$.
\end{claim}
\begin{claimproof}
If $i(k,y) < i-c$, then $p_j(s_{i(k,y)}) < p_j(s_{i-c}) < |\tilde{y}|$ by Claim \ref{claim:bounded-haty}. Otherwise $p_j(s_{i-c}) \leq p_j(s_{i(k,y)}) < |\tilde{y}|$ by Claim \ref{claim:smallest-code-word-length-small}.
\end{claimproof}
\subparagraph{2. An independent short code word for (k,y) exists.}
Let $d$ be the biggest number for which $p_j(s_d) < |\tilde{y}|$ holds. By Claim \ref{claim:short-cw-long-proof} and (\ref{claim:supposition-iii}) it holds that $\max\{i-c,i(k,y)\} \leq d \leq i-1$. Consequently, $C_{k,y}^{d} \not = \emptyset$.

The following claims show that there is a code word $c(k,y,b)$ of length $s_d$ that, when added to $v$, will not have a conflict with the validity of $v$.
\begin{claim}\label{claim:noduplicatecw}
$C_{k,y}^{d} \cap v = \emptyset$.
\end{claim}
\begin{claimproof}
Suppose there is a code word $c(k,y,b) \in C_{k,y}^{d} \cap v$. We argue that $(v,\rho)$ is $r^0$-valid. By Observation \ref{obs:vrho-rskj-valid}, $(v,\rho)$ is $r_{s_{k,j}-1}$-valid. By Claims \ref{claim:small-polynomial} and \ref{claim:l<i}, $p_k(p_j(s_i)) < s_{i+1}$. By (\ref{claim:supposition-iii}), $F_k^v(\tilde{y})=y$ is definite, because $p_k(|\tilde{y}|) \leq p_k(p_j(s_i)) < s_{i+1}$. By the choice of $d$, $p_j(|c(k,y,b)|) = p_j(s_d) < |\tilde{y}|$. Then the assumption that $c(k,y,b) \in v$ gives that V\ref{V5}($v,r^0$) is satisfied for $(k,j)$ and hence $(v,\rho)$ is $r^0$-valid. 

By Claim \ref{claim:wskj-defined-for-stage}, $\ls{w_{s_{k,j}}} \leq s_{i-c}$. Hence, $\ls{w_{s_{k,j}-1}} \leq s_{i-c-1}$. Since $w_{s_{k,j}-1} \sqsubseteq w_s \sqsubseteq w$ (cf.~Lemma \ref{lemma:valid-extension}) and $v$ and $w$ agree on $\Sigma^{\leq s_{i-c}}$ (cf.~Claim \ref{claim:V9-recursive}), also $w_{s_{k,j}-1} \sqsubsetneq w_{s_{k,j}} \sqsubseteq v$. Thus, at step $s_{k,j}$, the oracle construction would have chosen $r^0$ with the $r^0$-valid extension $(v,\rho) \sqsupsetneq (w_{s_{k,j}-1},\sigma_{s_{k,j}-1})$ instead of $r$, a contradiction.
\end{claimproof}
\begin{claim}\label{claim:cw-choice}
$C_{k,y}^{d} \setminus (N_{\mathcal{G}^v_{\rho}}(v) \cup Q(F_k^v(\tilde{y}))) \neq \emptyset$.
\end{claim}
\begin{claimproof}
By Observation \ref{obs:vrho-rskj-valid}, we have that $(v,\rho)$ is $r_{s_{k,j}-1}$-valid. Hence, $(v,\rho)$ is valid. Furthermore, V\ref{V8}($w_s,r_s$) is satisfied (cf.~Lemma \ref{lemma:valid-extension}). Since $v \supseteq w_s$, V\ref{V8}($v,r_s$) is also satisfied. Furthermore, $\ls{v} = s_i$ and we have already seen that $C_{k,y}^{d} \neq \emptyset$. This lets us apply Lemma \ref{lemma:graph-structure} with $i$ for $\ell$, $r_s$ for $r$, $v$ for $w$, $\rho$ for $\sigma$, and 
\[d \geq i - c  \geq i - \card{\N^2 \cap \supp(r_s)} > i - 2\max\{k' \mid k' \in \dom(r_s)\}\] 
for $j$ which shows that there are $s_{i+2}$ choices for $c(k,y,b) \in C_{k,y}^{d} \setminus N_{\mathcal{G}^v_{\rho}}(v)$. Since by (\ref{claim:supposition-iii}), $|\tilde{y}| \leq p_j(s_i)$, we have $\card{Q(F_k^v(\tilde{y}))} \leq p_k(p_j(s_i)) < s_{i+1}$. So $C_{k,y}^{d} \setminus (N_{\mathcal{G}^v_{\rho}}(v) \cup Q(F_k^v(\tilde{y}))) \neq \emptyset$.
\end{claimproof}
From now on let us fix some code word $c(k,y,b)$ that is chosen according to Claim \ref{claim:cw-choice}, i.e.,
\[c(k,y,b) \in C_{k,y}^{d} \setminus (N_{\mathcal{G}^v_{\rho}}(v) \cup Q(F_k^v(\tilde{y}))).\]
Let
\[v_{d-1} \coloneqq v \cup \{c(k,y,b)\} \text{ and } \rho_{d-1} \coloneqq \rho \cup \{(k,y,b) \mapsto \hhat{y}_{\smash{F_k^{\smash{v_{d-1}}}}}\}.\]
Note that $\hhat{y}_{\smash{F_k^{\smash{v_{d-1}}}}}$ exists, because $c(k,y,b) \notin Q(F_k^v(\tilde{y}))$ and Observation \ref{obs:oracle-questions} gives 
\[F_k^{v \cup \{c(k,y,b)\}}(\tilde{y}) = F_k^v(\tilde{y}).\]
From this, Observation \ref{obs:well-defined}, $\rho \in \mathcal{S}^v$ (follows from V\ref{V1}($v,\rho$)), and $c(k,y,b) \notin v$, it follows that $\rho_{d-1}$ is well-defined. Furthermore,  $v_{d-1} \supsetneq v$, and $\rho_{d-1} \sqsupsetneq \rho$.
\begin{claim}\label{claim:v-rho-valid}
The pair $(v_{d-1},\rho_{d-1})$ is valid.
\end{claim}
\begin{claimproof}
By Observation \ref{obs:vrho-rskj-valid}, $(v,\rho)$ is valid. In order to apply Lemma \ref{lemma:remain-valid} to $(v,\rho)$ and $\{c(k,y,b)\}$, let us verify that the assumptions \ref{lemma:remain-valid}.\ref{lemma:remain-valid-i}--\ref{lemma:remain-valid-iv} hold. For $c(k,y,b)$ holds \ref{lemma:remain-valid}.\ref{lemma:remain-valid-i} and \ref{lemma:remain-valid}.\ref{lemma:remain-valid-ii} by Claim \ref{claim:noduplicatecw} and $\card{\{c(k,y,b)\}} = 1$, and \ref{lemma:remain-valid}.\ref{lemma:remain-valid-iii} by Claim \ref{claim:cw-choice} and the choice of $c(k,y,b)$.

To \ref{lemma:remain-valid}.\ref{lemma:remain-valid-iv}: By Claim \ref{claim:cw-choice}, $c(k,y,b) \notin Q(F_k^v(\tilde{y}))$. Hence, $F_k^{v_{d-1}}(\tilde{y}) = y$. Thus,
\[|\hhat{y}_{\smash{F_k^{\smash{v_{d-1}}}}}| \leq |\tilde{y}| \leq p_j(s_{d+1}) \leq p_k(p_j(s_{d+1})) < s_{d+2}\]
where the second inequality holds by the choice of $d$ and the last inequality by Claim \ref{claim:small-polynomial} and $d \geq i(k,y)$. Therefore, we can apply Lemma \ref{lemma:remain-valid} and obtain that $(v_{d-1},\rho_{d-1})$ is valid.
\end{claimproof}
\subparagraph{3. Construct the final oracle.}
Recall the current situation: We have $p_j(s_{i(k,y)}) < |\tilde{y}| \leq p_j(s_{i})$, i.e., there is a long shortest $F_k^v$-proof whose underlying computation is definite. We added $c(k,y,b)$ to $v$ in order to satisfy $r^0(k,j)=0$. However, we must take two things into account: First, by adding a word of length $s_d$ to $v$, the stages from $s_d$ on may recognize this and are not ``maximal'' anymore, so we may have to add more code words on these stages to satisfy V\ref{V7}. Second, $r^0(k,j)=0$ is only satisfied via $c(k,y,b)$ when no significantly shorter proof than $\tilde{y}$ appears while we add further words to $v$. As we will prove in the following, it suffices to create the final pair $(v',\rho')$ by iteratively satisfying V\ref{V7} with respect to $r^0$ for the stages $s_d$ to $s_i$. 

Let 
\[(v',\rho') \coloneqq \text{prod}(v_{d-1},\rho_{d-1},r^0,d)\]
and let $(v_i,\rho_i), \dots ,(v_{d},\rho_{d})$ as well as $Y_{i-1}, \dots , Y_{d-1}$ be defined according to this procedure call. Observe that $(v_{d-1},\rho_{d-1})$ from the procedure call is the same pair as in the current proof and that $(v',\rho') = (v_i,\rho_i)$. We state several simple observations.
\begin{observation}\label{obs:prod1-iterative}
\begin{romanenumerate}
\item\label{obs:prod1-iterative-i} $Y_{i-1}, \dots ,Y_{d-1} \subseteq C$.
\item\label{obs:prod1-iterative-ii} $\rho_{i}, \dots , \rho_{d-1}$ are well-defined.
\item\label{obs:prod1-iterative-iii} $\rho_i \sqsupseteq \rho_{i-1} \sqsupseteq \dots \sqsupseteq \rho_{d-1} \sqsupseteq \rho \sqsupseteq \sigma \sqsupseteq \sigma_s$.
\item\label{obs:prod1-iterative-iv} $v_i \supseteq v_{i-1} \supseteq \dots \supseteq v_{d-1} \supseteq v \supseteq w \supseteq w_s$.
\item\label{obs:prod1-iterative-v} $v_{\ell+1}$ and $v_{\ell}$ agree on $N_{\smash{\mathcal{G}^{\smash{v_\ell}}_{\rho_\ell}}}(v_\ell)$ for $d-1 \leq \ell < i$.
\item\label{obs:prod1-iterative-vi} $N_{\smash{\mathcal{G}^{\smash{v_{\ell+1}}}_{\rho_{\ell+1}}}}(v_{\ell+1}) \supseteq N_{\smash{\mathcal{G}^{\smash{v_{\ell}}}_{\rho_{\ell}}}}(v_{\ell})$ for $d-1 \leq \ell < i$.
\item\label{obs:prod1-iterative-vii} $N_{\smash{\mathcal{G}^{\smash{v_i}}_{\rho_i}}}(v_i) \supseteq N_{\smash{\mathcal{G}^{\smash{v_{i-1}}}_{\rho_{i-1}}}}(v_{i-1}) \supseteq \dots \supseteq N_{\smash{\mathcal{G}^{\smash{v_{d-1}}}_{\rho_{d-1}}}}(v_{d-1}) \supseteq N_{\mathcal{G}^v_\rho}(v) \supseteq N_{\mathcal{G}^w_\sigma}(w)$.
\item\label{obs:prod1-iterative-viii} For $c(k',y',b') \in Y_{\ell}$ with $d-1 \leq \ell < i$ it holds $p_{k'}(p_{j'}(s_{\ell+1})) < s_{\ell + 2}$ for $j'$ such that $r^0(k',j')>0$.
\item\label{obs:prod1-iterative-ix} $v_i, \dots , v_{d-1}, v, w$ agree on $N_{\mathcal{G}^w_\sigma}(w)$.
\item\label{obs:prod1-iterative-x} $\max\{\hhat{y}_{\smash{F_k^{\smash{v_i}}}}, \hhat{y}_{\smash{F_k^{\smash{v_{i-1}}}}}, \dots , \hhat{y}_{\smash{F_k^{\smash{v_{d-1}}}}}\} \leq \hhat{y}_{F_{k}^v} = \tilde{y}$ 
\item\label{obs:prod1-iterative-xi} $v_i$ and $v_\ell$ agree on $\Sigma^{\smash{\leq s_\ell}}$ for $d-1 \leq \ell \leq i$.
\item\label{obs:prod1-iterative-xii} $v_i$ and $v$ agree on $\Sigma^{\smash{\leq s_{d-1}}}$.
\item\label{obs:prod1-iterative-xiii} $v_i$ and $w$ agree on $\Sigma^{\smash{\leq s_{i-c-1}}}$.
\item\label{obs:prod1-iterative-xiv} $\ls{v_i} = s_i$.
\end{romanenumerate}
\end{observation}
\begin{proof}
To (\ref{obs:prod1-iterative-i}): Follows from line \ref{alg2:5}.
\medskip
\\
To (\ref{obs:prod1-iterative-ii}): We have already proven that $\rho_{d-1}$ is well-defined. By Observation \ref{obs:well-defined}, $\rho_{d-1} \in \mathcal{S}^{v_{d-1}}$ (cf.~Claim \ref{claim:v-rho-valid}), line \ref{alg2:7} and line \ref{alg2:9}, we have that  
\[\{(k,y,b) \mid c(k,y,b) \in Y_\ell \text{ for }  i-1 \leq \ell < d\} \cap \dom(\rho_{d-1}) = \emptyset,\]
$Y_{i-1}, \dots , Y_{d-1}$ are pairwise disjoint and the respective newly defined values exist. Hence, $\rho_i, \dots, \rho_{d}$ are also well-defined.
\medskip
\\
To (\ref{obs:prod1-iterative-iii}) and (\ref{obs:prod1-iterative-iv}): This holds because we only add words to these oracles, or respectively only increase the domains of the proof selectors and $(v,\rho) \supseteq (w,\sigma) \supseteq (w_s,\sigma_s)$ holds by the assumptions of Claim \ref{claim:V9-recursive} and Lemma \ref{lemma:valid-extension}.
\medskip
\\
To (\ref{obs:prod1-iterative-v}): Follows from line \ref{alg2:6}.
\medskip
\\
To (\ref{obs:prod1-iterative-vi}): Follows from (\ref{obs:prod1-iterative-iii}), (\ref{obs:prod1-iterative-iv}), (\ref{obs:prod1-iterative-v}) and Claim \ref{claim:graph-monotonicity}. 
\medskip
\\
To (\ref{obs:prod1-iterative-vii}): By the definition of $v_{d-1}$, $v_{d-1}$ and $v$ agree on $N_{\mathcal{G}^v_\rho}(v)$. By the assumptions of Claim \ref{claim:V9-recursive}, $v$ and $w$ agree on $N_{\mathcal{G}^{w}_{\sigma}}(w)$. Then (\ref{obs:prod1-iterative-iii}), (\ref{obs:prod1-iterative-iv}), and Claim \ref{claim:graph-monotonicity} gives $N_{\smash{\mathcal{G}^{\smash{v_{d-1}}}_{\rho_{d-1}}}}(v_{d-1}) \supseteq N_{\mathcal{G}^v_\rho}(v) \supseteq N_{\mathcal{G}^{w}_{\sigma}}(w)$. Together with (\ref{obs:prod1-iterative-vi}), the whole chain follows.
\medskip
\\
To (\ref{obs:prod1-iterative-viii}): We have $r^0(k',j') = m < |y'| < |c(k',y',b')| = s_{\ell+1}$ (line \ref{alg2:5}). Together with R\ref{R2}($r^0$) and Claim \ref{claim:stages2}.\ref{claim:stages2-vii}, $p_{k'}(p_{j'}(s_{\ell + 1})) < s_{\ell+2}$.
\medskip
\\
To (\ref{obs:prod1-iterative-ix}): By the definition of $v_{d-1}$, $v_{d-1}$ and $v$ agree on $N_{\mathcal{G}^v_\rho}(v)$. By the assumptions of Claim \ref{claim:V9-recursive}, $v$ and $w$ agree on $N_{\mathcal{G}^{w}_{\sigma}}(w)$. Together with (\ref{obs:prod1-iterative-v}) and (\ref{obs:prod1-iterative-vii}), all these sets agree at least on $N_{\mathcal{G}^{w}_{\sigma}}(w)$.
\medskip
\\
To (\ref{obs:prod1-iterative-x}): Since $c(k,y,b) \notin Q(F_k^v(\tilde{y}))$, $F_k^{\smash{v_{d-1}}}(\tilde{y})=y$ and thus, $\hhat{y}_{\smash{F_k^{\smash{v_{d-1}}}}} \leq \tilde{y} = \hhat{y}_{F_k^v}$. Furthermore, $\rho_{d-1}(k,y,b) = \hhat{y}_{\smash{F_k^{\smash{v_{d-1}}}}}$. Hence, by the definition of dependency graphs, $Q(F_k^{\smash{v_{d-1}}} (\rho_{d-1} (k,y,b))) = Q(F_k^{\smash{v_{d-1}}} (\hhat{y}_{\smash{F_k^{\smash{v_{d-1}}}}})) \subseteq N_{\smash{\mathcal{G}^{\smash{v_{d-1}}}_{\rho_{d-1}}}}(v_{d-1})$. By (\ref{obs:prod1-iterative-vii}) and line \ref{alg2:6}, $F_k^{\smash{v_{\ell}}}(\hhat{y}_{\smash{F_k^{\smash{v_{d-1}}}}}) = y$ for all $d \leq \ell \leq i$. This gives $\hhat{y}_{\smash{F_k^{\smash{v_{\ell}}}}} \leq \hhat{y}_{\smash{F_k^{\smash{v_{d-1}}}}} \leq \tilde{y}$ for all $d \leq \ell \leq i$ and the observation follows.
\medskip
\\
To (\ref{obs:prod1-iterative-xi}): $v_i$ and $v_{d-1}$ agree on $\Sigma^{\leq s_{d-1}}$, because Algorithm \ref{alg:2} adds only words of length $\geq s_d$ to $v_{d-1}$. After $v_\ell$ is defined in iteration $\ell - 1$, Algorithm \ref{alg:2} only adds words of length $\geq s_{\ell+1}$ until $v_i$ is defined. Hence, $v_\ell$ and $v_i$ agree on $\Sigma^{\leq s_\ell}$ for $d \leq \ell \leq i$.
\medskip
\\
To (\ref{obs:prod1-iterative-xii}): By (\ref{obs:prod1-iterative-xi}), $v_i$ and $v_{d-1}$ agree on $\Sigma^{\leq s_{d-1}}$. Since $v$ and $v_{d-1}$ differ only by $c(k,y,b)$ which is of length $s_d$, they also agree on $\Sigma^{\leq s_{d-1}}$. Hence, $v$ and $v_i$ agree on $\Sigma^{\leq s_{d-1}}$.
\medskip
\\
To (\ref{obs:prod1-iterative-xiii}): By assumption of Claim \ref{claim:V9-recursive}, $v$ and $w$ agree on $\Sigma^{\leq s_{i-c}}$. By (\ref{obs:prod1-iterative-xii}) and $d \geq i-c$, $v_i$ and $v$ agree on $\Sigma^{\leq s_{d-1}} \supseteq \Sigma^{\leq s_{i-c-1}}$. Hence, $v_i$ and $w$ agree on $\Sigma^{\leq s_{i-c-1}}$.
\medskip
\\
To (\ref{obs:prod1-iterative-xiv}): Follows from $\ls{v} =s_i$ (cf.~Claim \ref{claim:V9-recursive}), $d < i$ and that Algorithm \ref{alg:2} adds only words of length $\leq s_i$.
\end{proof}
To complete the proof of Claim \ref{claim:V9-recursive}, it remains to show that $(v_i,\rho')$ is $r^0$-valid with the exception of V\ref{V6}($v_i,r^0$), which might be violated. 
\begin{claim}\label{claim:wc-graphvalid}
The pairs $(v_i,\rho_i), \dots ,(v_{d-1},\rho_{d-1})$ are valid.
\end{claim}
\begin{claimproof}
We proof this inductively. By Claim \ref{claim:v-rho-valid}, $(v_{d-1}, \rho_{d-1})$ is valid. Let $(v_\ell,\rho_\ell)$ for $d-1 \leq \ell < i$ be valid. In order to apply Lemma \ref{lemma:remain-valid} to $(v_\ell,\rho_\ell)$ and $Y_\ell$, let us verify that the assumptions \ref{lemma:remain-valid}.\ref{lemma:remain-valid-i}--\ref{lemma:remain-valid-iv} hold. Observe that \ref{lemma:remain-valid}.\ref{lemma:remain-valid-i} holds by line \ref{alg2:7}, \ref{lemma:remain-valid}.\ref{lemma:remain-valid-ii} holds by line \ref{alg2:8}, and \ref{lemma:remain-valid}.\ref{lemma:remain-valid-iii} holds by line \ref{alg2:6}. 

To \ref{lemma:remain-valid}.\ref{lemma:remain-valid-iv}: Let $c(k',y',b') \in Y_{\ell}$. By Observation \ref{obs:prod1-iterative}.\ref{obs:prod1-iterative-viii}, $p_{k'}(p_{j'}(s_{\ell + 1})) < s_{\ell+2}$ for $j'$ with $r^0(k',j')>0$. Hence, by line \ref{alg2:9} and $p_{k'}(p_{j'}(s_{\ell+1})) < s_{\ell+2}$, \ref{lemma:remain-valid}.\ref{lemma:remain-valid-iv} holds. We obtain that $(v_{\ell+1},\rho_{\ell+1})$ is valid.
\end{claimproof}
\begin{claim}\label{claim:wc-V9-valid}
The requirement V\ref{V8}($v_i,r^0$) is satisfied.
\end{claim}
\begin{claimproof}
From the $r_s$-validity of $(w,\sigma)$ (cf.~Lemma \ref{lemma:valid-extension}) it follows that
\[\ls{v_i} \geq \ls{w} \geq \max\{r_s(k) \mid k \in \dom(r_s)\} \geq \max\{r^0(k) \mid k' \in \dom(r^0)\}.\tag*{\claimqedhere}\]
\end{claimproof}
\begin{claim}\label{claim:wc-V6-valid}
The requirement V\ref{V5}($v_i,r^0$) is satisfied.
\end{claim}
\begin{claimproof}
Consider V\ref{V5}($v_i,r^0$) for some tuple $(k',j') \not = (k,j)$ with $r^0(k',j')=0$. Then $(k',j')$ was treated before $(k,j)$ in the oracle construction. Note that by Claim \ref{claim:wskj-defined-for-stage}, $\ls{w_{s_{k,j}}} \leq s_{i-c}$, so $\ls{w_{\smash{s_{k',j'}}}} \leq s_{i-c-1}$. Observe that $w_s \sqsupseteq w_{\smash{s_{k',j'}}}$. Moreover, $v_i$ and $w$ agree on $\Sigma^{< s_{i-c}}$, because $v$ and $w$ agree on $\Sigma^{\leq s_{i-c}}$ (cf.~Claim \ref{claim:V9-recursive}) and we added only words of length $ \geq s_d \geq s_{i-c}$ to $v$. Hence also $v_i \sqsupseteq w_{\smash{s_{k',j'}}}$. 

Since V\ref{V5}($w_{s_{k',j'}},r_{s_{k',j'}}$) is satisfied, there is some $c(k',x,b') \in w_{s_{k',j'}}$ with $F_{k'}^{\smash{w_{s_{k',j'}}}}(\hat{x}) = x$ defined\footnote{For readability we abbreviate $\hhat{x}_{F_{k'}^{w_{s_{k',j'}}}}$ with $\hat{x}$.} and definite and $p_{j'}(|c(k',x,b')|) < |\hat{x}|$. The definite computation stays unchanged relative to $v_i \sqsupseteq w_{s_{k',j'}}$, $c(k',x,b') \in v_i$, and $p_{j'}(|c(k',x,b')|) < |\hhat{x}_{\smash{F_{k'}^{v_i}}}|$, because $\hhat{x}_{\smash{F_{k'}^{v_i}}}$ can not get shorter than $\hat{x}$, since the computations $F_{k'}^{\smash{w_{s_{k',j'}}}}(x')$ are also definite for $x' < \hat{x}$. Hence, V\ref{V5}($v_i,r^0$) is satisfied for $(k',j')$.

It remains to show that V\ref{V5}($v_i,r^0$) is satisfied for $(k,j)$. More precisely, we show that $c(k,y,b) \in v_i$ satisfies all requirements of V\ref{V5}($v_i,r^0$). We do this by analyzing V\ref{V7}($v,\rho,r,d$). Recall that V\ref{V7}($v,\rho,r,d$) is satisfied, because $s_d < s_i$ and $(v,\rho)$ is $r$-valid with the exception of V\ref{V6}($v,r$) for $(k,j)$ (cf.~Claim \ref{claim:V9-recursive}). Let
\[X \coloneqq \{c(k,y,b)\} \cup Y_{d-1}.\]
Then
\begin{align}\label{align:wc-V6-valid}
v \cup X = v_{d-1} \cup Y_{d-1},
\end{align}
because $v_{d-1} = v \cup \{c(k,y,b)\}$. Below we show that $X$ satisfies all requirements V\ref{V7a}($v,\rho,r,d$) to V\ref{V7d}($v,\rho,r,d$). Hence, V\ref{V7e}($v,\rho,r,d$) has to be violated, from which we will conclude that V\ref{V5}($v_i,r^0$) is satisfied.
\medskip
\\
To V\ref{V7a}($v,\rho,r,d$): By $r(k,j)>0$, (\ref{claim:supposition-i}), and $c(k,y,b) \in C^d_{k,y}$, $c(k,y,b)$ satisfies the properties of V\ref{V7a}($v,\rho,r,d$). 

Observe that $r^0$ differs from $r$ only for $(k,j)$ where $r^0(k,j)=0$. Hence, whenever $r^0(k',j') > 0$, then $r^0(k',j') = r(k',j')$. Together with line \ref{alg2:5}, all words in $Y_{d-1}$ also satisfy the properties of V\ref{V7a}($v,\rho,r,d$).
\medskip
\\
To V\ref{V7b}($v,\rho,r,d$): $c(k,y,b)$ was chosen such that $c(k,y,b) \notin N_{\mathcal{G}_\rho^v}(v)$. By Observation \ref{obs:prod1-iterative}.\ref{obs:prod1-iterative-vii} and line \ref{alg2:6}, it holds that $Y_{d-1} \cap N_{\mathcal{G}^{v}_{\rho}}(v) \subseteq Y_{d-1} \cap N_{\smash{\mathcal{G}^{\smash{v_{d-1}}}_{\rho_{d-1}}}}(v_{d-1})  = \emptyset$. Hence, $X \cap N_{\mathcal{G}^{v}_{\rho}}(v) = \emptyset$. 
\medskip
\\
To V\ref{V7c}($v,\rho,r,d$): By Claim \ref{claim:noduplicatecw}, $c(k,y,b) \notin v$. By line \ref{alg2:7} and Observation \ref{obs:prod1-iterative}.\ref{obs:prod1-iterative-iv}. $Y_{d-1} \cap v \subseteq Y_{d-1} \cap v_{d-1} = \emptyset$. Hence, $X \cap v = \emptyset$.
\medskip
\\
To V\ref{V7d}($v,\rho,r,d$): We have for all $k' \in \N^+$ and $y' \in \Sigma^*$
\begin{align*}
\card{(v \cup X) \cap C^d_{k',y'}} &\overset{\text{(\ref{align:wc-V6-valid})}}{=} \card{(v_{d-1} \cup Y_{d-1})  \cap C^d_{k',y'}} \overset{\text{line \ref{alg2:8}}}{\leq} 1.
\end{align*}
\noindent To V\ref{V7e}($v,\rho,r,d$): Since V\ref{V7}($v,\rho,r,d$) is satisfied (cf.~Claim \ref{claim:V9-recursive} and $d<i$) and $X$ satisfies V\ref{V7a}($v,\rho,r,d$) to V\ref{V7d}($v,\rho,r,d$), it follows that V\ref{V7e}($v,\rho,r,d$) must be violated for $X$. We show that V\ref{V7e}($v,\rho,r,d$) is satisfied for all words in $Y_{d-1}$, which implies that V\ref{V7e}($v,\rho,r,d$) is violated for $c(k,y,b)$.

Let $c(k',x,b') \in Y_{d-1}$, then line \ref{alg2:9} gives that $\hat{x}$ exists\footnote{For readability, we abbreviate $\hhat{x}_{\smash{F_{k'}^{v_{d-1} \cup Y_{d-1}}}}$ with $\hat{x}$} with $|\hat{x}| \leq p_{j'}(s_{d})$ for $j'$ with $r^0(k',j')>0$. Then both of the following holds:
\begin{align*}
&r(k',j') > 0\\
&F_{k'}^{v \cup X}(\hat{x}) = F_{k'}^{v_{d-1} \cup Y_{d-1}}(\hat{x}) = x \text{ by equation (\ref{align:wc-V6-valid}).}
\end{align*}
So V\ref{V7e}($v,\rho,r,d$) holds for all words in $Y_{d-1}$. Therefore V\ref{V7e}($v,\rho,r,d$) must be violated for $c(k,y,b) \in X$, i.e., it does not hold that $F_{k}^{v \cup X}(\hhat{y}_{F_k^{v \cup X}}) = y$ with $|\hhat{y}_{F_k^{v \cup X}}| \leq p_j(|c(k,y,b)|)$. Hence, 
\begin{align}\label{align:large-or-not-defined}
|\hhat{y}_{F_k^{v \cup X}}| > p_j(|c(k,y,b)|) \mbox{ or }|\hhat{y}_{F_k^{v \cup X}}| \mbox{ does not exist.}
\end{align}

This finishes the analysis of V\ref{V7}($v,\rho,r,d$). Now, we use equation (\ref{align:large-or-not-defined}) to show that V\ref{V5}($v_i,r^0$) is satisfied. By Claim \ref{claim:wc-graphvalid}, it holds that $\rho_i \in \mathcal{S}^{v_i}$. Since $c(k,y,b) \in v_i$, $F_k^{\smash{v_i}}(\rho_i(k,y,b)) = y$. Hence, $\hhat{y}_{\smash{F_k^{\smash{v_i}}}}$ exists. 

Suppose that $|\hhat{y}_{\smash{F_k^{\smash{v_i}}}}| \leq p_j(|c(k,y,b)|)$, then $F_k^{\smash{v_i}}(\hhat{y}_{F_k^{v_i}})$ can only ask queries of length $<s_{d+1}$, since $p_k(p_j(s_d)) < s_{d+1}$ by Claim \ref{claim:small-polynomial}. But then 
\begin{align*}
F_{k}^{v \cup X}(\hhat{y}_{F_k^{v_i}}) &= F_{k}^{(v \cup X) \cap \Sigma^{<s_{d+1}}}(\hhat{y}_{F_k^{v_i}}) \overset{\text{(\ref{align:wc-V6-valid})}}{=} F_k^{(v_{d-1} \cup Y_{d-1}) \cap \Sigma^{<s_{d+1}}}(\hhat{y}_{F_k^{v_i}})\\
 &\overset{\text{line \ref{alg2:11}}}{=} F_k^{v_d \cap \Sigma^{<s_{d+1}}}(\hhat{y}_{F_k^{v_i}})  \overset{\text{\ref{obs:prod1-iterative}.\ref{obs:prod1-iterative-xi}}}{=} F_k^{v_i \cap \Sigma^{<s_{d+1}}}(\hhat{y}_{F_k^{v_i}}) = F_k^{v_i}(\hhat{y}_{F_k^{v_i}}) = y 
\end{align*}
and hence $F_{k}^{v \cup X}(\hhat{y}_{\smash{F_k^{\smash{v \cup X}}}}) = y$ where $|\hhat{y}_{\smash{F_k^{\smash{v \cup X}}}}| \leq |\hhat{y}_{\smash{F_k^{v_i}}}| \leq p_j(|c(k,y,b)|)$ by supposition. This contradicts equation (\ref{align:large-or-not-defined}), so $|\hhat{y}_{\smash{F_k^{\smash{v_i}}}}| > p_j(c|(k,y,b)|)$ holds. 

Finally, $F_k^{v_i}(\hhat{y}_{\smash{F_k^{\smash{v_i}}}})$ is definite, because
\[p_k(|\hhat{y}_{F_k^{v_i}}|) \overset{\text{\ref{obs:prod1-iterative}.\ref{obs:prod1-iterative-x}}}{\leq} p_k(|\hhat{y}_{F_k^v}|) \overset{\text{(\ref{claim:supposition-iii})}}{\leq} p_k(p_j(s_i)) \overset{\text{\ref{claim:small-polynomial}}}{<} s_{i+1}.\] 
Then, V\ref{V5}($v_i,r^0$) is satisfied for $(k,j)$, because $F_k^{v_i}(\hhat{y}_{F_k^{v_i}}) = y$ is defined and definite, $c(k,y,b) \in v_i$, and $p_j(|c(k,y,b)|) < |\hhat{y}_{F_k^{v_i}}|$ as proved above.
\end{claimproof}
\begin{claim}\label{claim:wc-V8-valid}
The requirement V\ref{V7}($v_i,\rho_i,r^0$) is satisfied.
\end{claim}
\begin{claimproof}
We want to invoke Claim \ref{claim:requirement-monotonicity}.\ref{claim:requirement-monotonicity-vi} for the stages $< s_d$. 
It holds that V\ref{V7}($v,\rho,r,d'$) is satisfied for $d' < d$ (cf.~Claim \ref{claim:V9-recursive}), $(v_i,\rho_i) \supseteq (v, \rho)$ by Observation \ref{obs:prod1-iterative}.\ref{obs:prod1-iterative-iii} and \ref{obs:prod1-iterative}.\ref{obs:prod1-iterative-iv}, $v_i $ and $v$ agree on all $\Sigma^{<s_{d'+1}}$ for $d' < d$ by \ref{obs:prod1-iterative}.\ref{obs:prod1-iterative-xii}, and $N_{\smash{\mathcal{G}^{\smash{v_i}}_{\rho_i}}}(v_i) \supseteq N_{\mathcal{G}^{v}_{\rho}}(v)$ by Observation \ref{obs:prod1-iterative}.\ref{obs:prod1-iterative-vii}. Thus by Claim \ref{claim:requirement-monotonicity}.\ref{claim:requirement-monotonicity-vi}, it holds that V\ref{V7}($v_i,\rho_i,r,d'$) is satisfied for all $d' < d$. By the arguments used in the proof of Claim \ref{claim:hereditary-r}, it holds that also V\ref{V7}($v_i,\rho_i,r^0,d'$) is satisfied for $d' < d$.

Let $\ell \in \N^+$ such that $d \leq \ell \leq i$. We show that V\ref{V7}($v_i,\rho_i,r^0,\ell$) is satisfied. We will invoke Lemma \ref{lemma:V8-extension} to show that V\ref{V7}($v_\ell,\rho_\ell,r^0,\ell$) is satisfied and use Claim \ref{claim:requirement-monotonicity}.\ref{claim:requirement-monotonicity-vi} to transfer this to $(v_i,\rho_i)$. 

Considering prod($v_{d-1},\rho_{d-1},r^0,d$), the lines \ref{alg2:5} to \ref{alg2:9} in the iteration of the loop with value $\ell - 1$ are identical to V\ref{V7a}($v_{\ell-1},\rho_{\ell-1},r^0,\ell$) to V\ref{V7e}($v_{\ell-1},\rho_{\ell-1},r^0,\ell$). Hence, for $(v_{\ell-1},\rho_{\ell-1})$, $Y_{\ell-1}$ is of greatest cardinality satisfying all properties of V\ref{V7}($v_{\ell-1},\rho_{\ell-1},r^0,\ell$). Furthermore, $(v_{\ell-1},\rho_{\ell-1})$ is valid by Claim \ref{claim:wc-graphvalid} and extends to $(v_\ell,\rho_\ell)$ via $Y_{\ell-1}$ exactly as in Lemma \ref{lemma:V8-extension}. Then Lemma \ref{lemma:V8-extension} gives that V\ref{V7}($v_\ell,\rho_\ell,r^0,\ell$) is satisfied. 

Since additionally $(v_i,\rho_i) \supseteq (v_\ell, \rho_\ell)$ by Observation \ref{obs:prod1-iterative}.\ref{obs:prod1-iterative-iii} and \ref{obs:prod1-iterative}.\ref{obs:prod1-iterative-iv}, $v_i$ and $v_\ell$ agree on $\Sigma^{< s_{\ell+1}}$ by Observation \ref{obs:prod1-iterative}.\ref{obs:prod1-iterative-xi}, and $N_{\smash{\mathcal{G}^{\smash{v_i}}_{\rho_i}}}(v_i) \supseteq N_{\smash{\mathcal{G}^{\smash{v_\ell}}_{\rho_\ell}}}(v_\ell)$ by Observation \ref{obs:prod1-iterative}.\ref{obs:prod1-iterative-vii}, Claim \ref{claim:requirement-monotonicity}.\ref{claim:requirement-monotonicity-vi} gives that V\ref{V7}($v_i,\rho_i,r^0,\ell$) is satisfied.

In total, V\ref{V7}($v_i,\rho_i,r^0,i'$) is satisfied for all stages $s_{i'} \leq s_i$. Hence, V\ref{V7}($v_i,\rho_i,r^0$) is satisfied.
\end{claimproof}

Now we got all pieces to finalize the proof of Claim \ref{claim:V9-recursive}. We have $v' = v_i \supseteq v$ (Observation \ref{obs:prod1-iterative}.\ref{obs:prod1-iterative-iv}) with $\ls{v'} = s_i$ (Observation \ref{obs:prod1-iterative-xiv}), $v'$ and $w$ agree on $\Sigma^{\leq s_{i-c-1}} \cup N_{\mathcal{G}^{w}_{\sigma}}(w)$ (Observation \ref{obs:prod1-iterative}.\ref{obs:prod1-iterative-xii} and \ref{obs:prod1-iterative}.\ref{obs:prod1-iterative-ix}), $\rho' = \rho_i \sqsupseteq \rho$ (Observation \ref{obs:prod1-iterative}.\ref{obs:prod1-iterative-iii}), and $(v',\rho')$ is $r^0$-valid (Claims \ref{claim:wc-graphvalid}, \ref{claim:wc-V9-valid}, \ref{claim:wc-V6-valid}, \ref{claim:wc-V8-valid}) except V\ref{V6}($v',r^0$), which may be violated.
\end{claimproof}

We continue the proof of Lemma \ref{lemma:valid-extension}. Recall that we want to prove that $(v,\rho)$ is $r_s$-valid and we have already shown that V\ref{V1}, V\ref{V2}, V\ref{V3}, V\ref{V4}, V\ref{V5}, V\ref{V7} and V\ref{V8} are satisfied with respect to $(v,\rho,r_s)$ (Claims \ref{claim:V2-V3-V5-V6}, \ref{claim:V5-V8}, \ref{claim:V7}). Thus, the following remains to show.
\begin{claim}\label{claim:V6}
The requirement V\ref{V6}($v,r_s$) is satisfied.
\end{claim}
\emph{Proof sketch:} Using induction, we show that we can repeatedly invoke Claim \ref{claim:V9-recursive} and derive a contradiction to the construction of $r_s$. Intuitively, when invoking Claim \ref{claim:V9-recursive} with an erroneous pair, the claim either ``repairs'' the error or propagates it to a previously treated tuple. Since the domain of $r_s$ is finite, this will eventually lead to a contradiction, at latest when arguing for the first treated tuple in $\supp(r_s)$. So as long as V\ref{V6} is not satisfied, we show: 
\begin{enumerate}
\item $(v^0,\rho^0) \coloneqq (v,\rho)$ satisfies the prerequisites of Claim \ref{claim:V9-recursive} and lets us invoke this claim, giving a pair $(v^1,\rho^1)$ with the respective properties described in Claim \ref{claim:V9-recursive}.
\item When given $(v^n, \rho^n)$ by the $n$-th invocation of Claim \ref{claim:V9-recursive}, we can invoke this claim with $(v^n,\rho^n$), obtaining $(v^{n+1}, \rho^{n+1})$ with the respective properties described in Claim \ref{claim:V9-recursive}.
\end{enumerate}
\begin{claimproof}
Assume V\ref{V6}($v,r_s$) is violated. Then let $(k,j)$ be the first tuple according to $\mathcal{T}$ (i.e., the tuple that is treated first in the oracle construction) for which V\ref{V6}($v,r_s$) is violated. Let $(v^0,\rho^0) \coloneqq (v,\rho)$ and $(k^0,j^0) \coloneqq (k,j)$.
\medskip
\\
\textbf{Base case:} We have $(k^0,j^0) \in \dom(\varphi)$ and define $c^0 \coloneqq \varphi(k^0,j^0) \geq 1$. Observe that $v^0 \supseteq w$ and $v^0$ and $w$ agree on $\Sigma^{\leq s_{i-1}} \supseteq \Sigma^{\leq s_{i-c^0}}$ (Observation \ref{obs:w-v-agree2}) and on $N_{\mathcal{G}^{w}_{\sigma}}(w)$ (Observation \ref{obs:w-ws-agree}). It holds that $\ls{v^0} = s_{i}$ (by $(v^0,\rho^0) \coloneqq \text{prod}(w,\sigma,r_s,i)$) and $\rho^0  \sqsupseteq \sigma_s$. Also, $(v^0,\rho^0)$ is $r_s$-valid (Claims \ref{claim:V2-V3-V5-V6}, \ref{claim:V5-V8}, \ref{claim:V7}) with the exception of V\ref{V6}($v^0,r_s$), where $(k^0,j^0)$ is the first tuple according to $\mathcal{T}$ for which V\ref{V6}($v^0,r_s$) is violated (as assumed in this claim). Hence, $(v^0,\rho^0)$ is $r_{s_{\smash{k^0,j^0}}}$-valid (Claim \ref{claim:hereditary-r}) with the exception of V\ref{V6}($v^0,r_{s_{\smash{k^0,j^0}}}$) for $(k^0,j^0)$. Then we can invoke Claim \ref{claim:V9-recursive} with $(v^0,\rho^0)$ and $(k^0,j^0)$ and get a pair $(v^1,\rho^1)$ with the respective properties described in Claim \ref{claim:V9-recursive}.
\medskip
\\
\textbf{Induction step:} Let $(v^{n-1},\rho^{n-1})$, $(k^{n-1},j^{n-1})$ and $c^{n-1} \coloneqq \varphi(k^{n-1},j^{n-1})$ be the input and $(v^n,\rho^n)$ be the output of the $n$-th iterated invocation of Claim \ref{claim:V9-recursive}. The induction hypothesis gives the following properties for $(v^{n-1},\rho^{n-1})$ and $(v^n,\rho^n$):
\begin{romanenumerate}
\item\label{claim:V6-i} $v^n \supseteq v^{n-1} \supseteq w$.
\item\label{claim:V6-ii} $\ls{v^n} = \ls{v^{n-1}} = s_i$.
\item\label{claim:V6-iii} $v^n$ and $w$ agree on $\Sigma^{\leq s_{i-c^{n-1}-1}} \cup N_{\mathcal{G}^{w}_{\sigma}}(w)$.
\item\label{claim:V6-iv} $\rho^n \sqsupseteq \rho^{n-1} \sqsupseteq \sigma$.
\item\label{claim:V6-v} $(v^n, \rho^n)$ is $r_{s_{\smash{k^{n-1},j^{n-1}}}-1} \cup \{(k^{n-1},j^{n-1}) \mapsto 0\}$-valid except for V\ref{V6}($v^n,r_{s_{\smash{k^{n-1},j^{n-1}}}-1} \cup \{(k^{n-1},j^{n-1}) \mapsto 0\}$), which may be violated.
\end{romanenumerate}

If $(v^n, \rho^n)$ is indeed $r_{s_{\smash{k^{n-1},j^{n-1}}}-1} \cup \{(k^{n-1},j^{n-1}) \mapsto 0\}$-valid, we get a contradiction to our oracle construction, i.e., to the construction of $r_s$. By Claim \ref{claim:wskj-defined-for-stage}, $\ls{w_{s_{\smash{k^{n-1},j^{n-1}}}}} \leq s_{i-c^{n-1}}$, so $\ls{w_{s_{\smash{k^{n-1},j^{n-1}}}-1}} \leq s_{i-c^{n-1}-1}$.
Since $w \sqsupseteq w_s \sqsupseteq w_{s_{\smash{k^{n-1},j^{n-1}}}-1}$, and $v^n$ and $w$ agree on $\Sigma^{\leq s_{i-c^{n-1}-1}}$ by (\ref{claim:V6-iii}), it follows that $v^n \sqsupseteq w_{s_{\smash{k^{n-1},j^{n-1}}}-1}$. Thus, when treating the task $\tau _{k^{n-1},j^{n-1}}$, the oracle construction would have chosen
\begin{align*}
r_{s_{\smash{k^{n-1},j^{n-1}}}} &\coloneqq r_{s_{\smash{k^{n-1},j^{n-1}}}-1} \cup \{(k^{n-1},j^{n-1}) \mapsto 0\},\\
w_{s_{\smash{k^{n-1},j^{n-1}}}} &\coloneqq v^n \sqsupsetneq w_{s_{\smash{k^{n-1},j^{n-1}}}-1},\\
\sigma_{s_{\smash{k^{n-1},j^{n-1}}}} &\coloneqq \rho^n \sqsupseteq \sigma_{s_{\smash{k^{n-1},j^{n-1}}}-1},
\end{align*}
a contradiction to the definition of $r_s$ with $r_s(k^{n-1},j^{n-1}) > 0$.

Otherwise, $(v^n,\rho^n)$ is not $r_{s_{\smash{k^{n-1},j^{n-1}}}-1} \cup \{(k^{n-1},j^{n-1}) \mapsto 0\}$-valid, because the requirement V\ref{V6}($v^n,r_{s_{\smash{k^{n-1},j^{n-1}}}-1} \cup \{(k^{n-1},j^{n-1}) \mapsto 0\}$) is violated. Let $(k^n,j^n)$ be the first violating tuple according to $\mathcal{T}$. We can invoke Claim~\ref{claim:V9-recursive} for $(v^n,\rho^n)$. Observe that
\begin{itemize}
\item $c^n \coloneqq \varphi(k^n,j^n) > c^{n-1}$,
\item $v^n \supseteq v^{n-1} \supseteq w$ by (\ref{claim:V6-i}),
\item $\ls{v^n} = s_i$ by (\ref{claim:V6-ii}),
\item $v^n$ and $w$ agree on $\Sigma^{\leq s_{i - c^{n-1} - 1}} \supseteq \Sigma^{\leq s_{i-c^n}}$ and on $N_{\mathcal{G}^{w}_{\sigma}}(w)$ by (\ref{claim:V6-iii}),
\item $\rho^n \sqsupseteq \rho^{n-1} \sqsupseteq \sigma$ by (\ref{claim:V6-iv}),
\item $(v^n, \rho^n)$ is $r_{s_{k^n,j^n}}$-valid with the exception of V\ref{V6}($v^n,r_{s_{k^n,j^n}}$) for $(k^n,j^n)$ by (\ref{claim:V6-v}), Claim \ref{claim:hereditary-r} and $\tau_{k^n,j^n}$ being treated before $\tau_{k^{n-1},j^{n-1}}$.
\end{itemize}
Then Claim \ref{claim:V9-recursive} gives $(v^{n+1},\rho^{n+1})$ with the properties described there.
\medskip
\\
At some point, the induction has to stop via a contradiction. Recall that for all $n$ it holds that $\tau_{k^n,j^n}$ was treated before $\tau_{k^{n-1},j^{n-1}}$, $r_s(k^n,j^n) > 0$. Let $z \coloneqq \card{\{(k,j) \mid r_s(k,j)>0\}} \in \N$.
So, at latest after the $z$-th step of the induction, V\ref{V6}($v^{z},r_{s_{\smash{k^{z-1},j^{z-1}}}-1} \cup \{(k^{z-1},j^{z-1}) \mapsto 0\}$) can not be violated, because there is no tuple $(k^{z},j^{z})$ appearing before $(k^{z-1},j^{z-1})$ in $\mathcal{T}$ with $r_s(k^{z},j^{z}) > 0$. Consequently, the assumption is false and V\ref{V6}($v,r_s$) is not violated. This proves Claim \ref{claim:V6}.
\end{claimproof}
The Claims \ref{claim:V2-V3-V5-V6}, \ref{claim:V5-V8}, \ref{claim:V7}, and \ref{claim:V6} show that $(v,\rho)$ is $r_s$-valid. Since $v \sqsupsetneq w$ is completely defined for the stage $s_i$ and $\rho \sqsupseteq \sigma$, we have proven Lemma \ref{lemma:valid-extension}.
\end{proof}
\begin{proposition}\label{prop:valid-triples}
$(w_s,\sigma_s,r_s)_{s \in \N}$ is a sequence of valid triples and for all $s \in \N^+$ it holds that $(w_s, \sigma_s) \sqsupsetneq (w_{s-1},\sigma_{s-1})$ and $r_s \sqsupseteq r_{s-1}$.
\end{proposition}
\begin{proof}
We prove this inductively. The triple $(w_0,\sigma_0,r_0)$ is valid. By Claim \ref{claim:r-always-valid}, $r_0$ is valid. Furthermore V\ref{V5}($w_0,r_0$), V\ref{V6}($w_0,r_0$), V\ref{V7}($w_0,\sigma_0,r_0$), and V\ref{V8}($w_0,r_0$) hold trivially for a nowhere defined $r_0$. V\ref{V1}($w_0,\sigma_0$) is satisfied because $\sigma_0 \in \mathcal{S}^{\emptyset}$. V\ref{V3}($w_0$) and V\ref{V4}($w_0,\sigma_0$) are satisfied, because $w_0$ does not contain any code words. V\ref{V2}($w_0$) is satisfied, because $w_0$ is empty.

Let $s \in \N$ be arbitrary and $(w_{s-1},\sigma_{s-1},r_{s-1})$ be a valid triple. By Claim \ref{claim:r-always-valid}, $r_s$ is valid and $r_s \sqsupseteq r_{s-1}$ for $s \in \N^+$. Either $(w_s,\sigma_s)$ is chosen such that $(w_s,\sigma_s) \sqsupsetneq (w_{s-1},\sigma_{s-1})$ is $r_s$-valid and we are done. Otherwise, independent of the task, $(w_s,\sigma_s)$ is the result of a possibly repeated invocation of Lemma \ref{lemma:valid-extension}. Observe that the invocation of Lemma \ref{lemma:valid-extension} with $(w_{s-1},\sigma_{s-1}, r_{s-1})$ gives an $r_{s-1}$-valid pair $(v_1,\rho_1) \sqsupsetneq (w_{s-1},\sigma_{s-1})$. If necessary, we can invoke Lemma \ref{lemma:valid-extension} again with $(v_1,\rho_1,r_{s-1})$ and get an $r_{s-1}$-valid pair $(v_2,\rho_2) \sqsupsetneq (v_1,\rho_1)$. We can repeat this as often as necessary. Hence, we can extend $(w_{s-1},\sigma_{s-1})$ to $(w_s,\sigma_s)$ as described in the oracle construction, even when repeated invocations of Lemma \ref{lemma:valid-extension} are necessary. From this it follows that $(w_s,\sigma_s)$ is $r_{s-1}$-valid. By Claim \ref{claim:valid-stepup}, $(w_s,\sigma_s)$ is even $r_s$-valid. Hence, $(w_s,\sigma_s,r_s)$ is a valid triple. Consequently, $(w_s,\sigma_s,r_s)_{s\in \N}$ is a sequence of valid triples.
\end{proof}

Define $O \coloneqq \bigcup _{s \in \N} w_s$, $\sigma \coloneqq \bigcup _{s \in \N} \sigma_s$, and $r \coloneqq \bigcup _{r \in \N} r_s$. All three of them are well-defined and $O$ is fully defined, because by Proposition \ref{prop:valid-triples}, $(w_s, \sigma_s) \sqsupsetneq (w_{s-1}, \sigma_{s-1})$ and $r_s \sqsupseteq r_{s-1}$ for all $s \in \N^+$. 

\subsection{Properties of the Oracle} In this section, we prove the properties that hold relative to $O$. We start by proving a helpful claim about $O$ for later proofs.
\begin{claim}\label{claim:cw-trust}
If $c(k,y,b) \in C \cap O$, then $y \in \ran(F_k^O)$.
\end{claim}
\begin{claimproof}
Let $s_i \coloneqq |c(k,y,b)|$ and $s \in \N$, such that $\ls{w_s} \geq s_{i+2}$. Then it holds that $c(k,y,b) \in w_s$. Since $(w_s,\sigma_s,r_s)$ is valid, V\ref{V1}($w_s,\sigma_s$) and V\ref{V4}($w_s,\sigma_s$) are satisfied. By V\ref{V1}($w_s,\sigma_s$), $F_k^{w_s}(\sigma_s(k,y,b)) = y$ and by V\ref{V4}($w_s,\sigma_s$), $p_k(|\sigma_s(k,y,b)|) < s_{i+2}$. Hence, $F_k^{w_s}(\sigma_s(k,y,b))=y$ is definite and thus stays unchanged relative to $O \sqsupseteq w_s$. So $y \in \ran(F_k^O)$.
\end{claimproof}
In Theorem \ref{thm:thick-oracle-2}, we want to show that several $F_k^O$ are not optimal proof systems. For this, we define witness proof systems $G_k^O$ as
\begin{align*}
G_k^O(x) = \begin{cases}
y, &\text{if } x = 0c(k,y,b) \text{ and } c(k,y,b) \in O\\
F_k^O(x'), &\text{if } x=1x' \\
F_k^O(x), & \text{otherwise}
\end{cases}
\end{align*}
We show that $G_k^O$ is a proof system for $\ran(F_k^O)$. 
\begin{claim}\label{claim:GpsF}
$G_k^O$ is a proof system for $\ran(F_k^O)$.
\end{claim}
\begin{claimproof}
The inclusion $\ran(G_k^O) \supseteq \ran(F_k^O)$ can be seen by line 2 of the definition of $G_k^O$. The inclusion $\ran(G_k^O) \subseteq \ran(F_k^O)$ can not be violated by lines 2 and 3 of the definition of $G_k^O$. This also holds for line 1, because by V\ref{V2} and Claim \ref{claim:cw-trust}, if $c(k,y,b) \in O$, then $c(k,y,b) \in C$ and $y \in \ran(F_k^O)$.

The runtime of $G_k^O$ is polynomial. Checking for the case in line 1 can be done in polynomial time by asking $x$ without the leading bit as oracle query. Checking for the other cases is trivial. Computing the output takes either linear time for line 1 or one has to simulate at most $p_k(|x|)$ steps of $F_k^O$.
\end{claimproof}
Now we prove Theorem \ref{thm:thick-oracle-1}.
\begingroup
\def\thetheorem{\ref{thm:thick-oracle-1}}
\begin{theorem}\label{thm:thick-oracle-2}
There exists an oracle $O$ such that for all $k \in \N^+$ one of the following holds:
\begin{romanenumerate}
\item\label{thm:thick-oracle-2-i} $F_k^O$ is not optimal for $\ran(F_k^O)$.
\item\label{thm:thick-oracle-2-ii} $\ran(F_k^O) \in \UTIME^O (h^c)$ for a constant $c \in \N$.
\end{romanenumerate}
\end{theorem}
\addtocounter{theorem}{-1}
\endgroup
\begin{proof}
Let $k \in \N^+$ be arbitrary and $O$, $\sigma$, and $r$ as constructed at the end of section \ref{sec:oracle}. Let $F_k^O$ be optimal for $\ran(F_k^O)$. Then $F_k^O$ simulates $G_k^O$ via a simulation polynomial $p_j$, i.e., for all $y$, the shortest $F_k^O$-proof of $y$ is at most $p_j$ longer than its shortest $G_k^O$-proof. Intuitively, in this case it was not possible to diagonalize against $(k,\ell)$ for every $\ell \in \N^+$, because otherwise $G_k^O$ would infinitely often have short proofs using code words whereas $F_k^O$ has long proofs. The following claim makes this intuition precise.
\begin{claim}\label{claim:r>0}
There is some $\ell \leq j+1$ with $r(k,\ell) > 0$.
\end{claim}
\begin{claimproof}
Suppose otherwise, then the task $\tau _{k,j+1}$ gets treated at some step $s$ in the oracle construction and $r(k,j+1) = 0$. Since the pair $(w_s,\sigma_s)$ is $r_s$-valid, V\ref{V5}($w_s,r_s$) is satisfied. This means that there is some $y$ such that $F_k^{w_s}(\hhat{y}_{\smash{F_k^{\smash{w_s}}}}) = y$ is defined and definite, $c(k,y,b) \in w_s$ for some $b$, and $p_{j+1}(|c(k,y,b)|) < |\hhat{y}_{\smash{F_k^{\smash{w_s}}}}|$. Since $O \sqsupseteq w_s$, all of this also holds relative to $O$, especially that $\hhat{y}_{\smash{F_k^{\smash{w_s}}}} = \hhat{y}_{\smash{F_k^{\smash{O}}}}$, because due to the definiteness, no shorter proof for $y$ can emerge relative to $O$. Hence, $F_k^O$ does not simulate $G_k^O$ via the simulation polynomial $p_j$: it holds that $G_k^O(0c(k,y,b)) = y$, but
\[p_j(|c(k,y,b)|+1) \leq p_{j+1}(|c(k,y,b)|) < |\hhat{y}_{F_k^O}|.\] 
This makes a simulation of $G_k^O$ via the polynomial $p_j$ impossible.
\end{claimproof}
Let $\ell$ be the unique value such that $m \coloneqq r(k,\ell) > 0$ (cf.~R\ref{R3}). Then, intuitively, every fact $y \in \ran(F_k^O)$ of length $>m$ is encoded into $O$ via shortest code words. Hence for $|y| > m$, a $G_k^O$-proof for $y$ is approximately as long as the shortest code word for $k,y$. The next claim shows this.
\begin{claim}\label{claim:G-has-short-proofs}
For every $y \in \ran (F_k^O)^{>m}$ there is some $b$ such that $c(k,y,b) \in O \cap C^{i(k,y)}_{k,y}$.
\end{claim}
\begin{claimproof}
Let $y \in \ran (F_k^O)^{>m}$ be witnessed by $x$ and let $s_i > p_k(|x|) + s_{i(k,y)}$, i.e., $F_k^O(x)=y$ can not ask queries of length $s_i$ and $O^{< s_i}$ defines the membership of the shortest code words for $y,k$. Let $s \in \N$ such that $s_i \leq \ls{w_s} $. Since $(w_s,\sigma_s)$ is $r_s$-valid, V\ref{V6}($w_s,r_s$) holds. Observe that $r(k,\ell) = m > 0$, $y \in \Sigma^{>m}$, $s_{i(k,y)}  \leq s_i \leq \ls{w_s}$, $F_k^{w_s}(x) = F_k^{O}(x) = y$, and thus $|\hhat{y}_{\smash{F_k^{\smash{w_s}}}}| \leq |x| \leq p_k(|x|) \leq s_i \leq p_j(s_i)$. So, by V\ref{V6}($w_s,r_s$), there is some $b$ such that $c(k,y,b) \in w_s \cap C^{\smash{i(k,y)}}_{k,y} \subseteq O \cap C^{\smash{i(k,y)}}_{k,y}$.
\end{claimproof}
The following algorithm shows that $\ran (F_k^O) \in \UTIME^O (h^c)$ for a constant $c$.
\begin{algorithm}
\caption{decide-Fk}\label{alg3}
\begin{algorithmic}[1]
\State \textbf{Input:} $y \in \Sigma^*$ \label{alg3:1}
\If{$|y| \leq m$ or $i(k,y) <21$}\label{alg3:2}
\State \Return $1$ if $y \in \ran(F_k^O)$, \Return $0$ else\label{alg3:3}
\EndIf\label{alg3:4}
\State compute $i$ with $t_{i-1} < k + |y| + 2 \leq t_i$\label{alg3:5}
\State compute $s_i$\label{alg3:6}
\State query all $x =c(k,y,b) \in C^i_{k,y}$ non-deterministically\label{alg3:7}
\State \Return $1$ if $x \in O$, \Return $0$ else\label{alg3:8}
\end{algorithmic}
\end{algorithm}
\begin{claim}\label{claim:algo1-accepts}
Algorithm \ref{alg3} decides $\ran(F_k^O)$.
\end{claim}
\begin{claimproof}
For inputs $|y| \leq m$ with $i(k,y)<21$, the algorithm is correct. Otherwise $|y| > m$. Line \ref{alg3:6} computes $s_i = s_{i(k,y)}$, because line \ref{alg3:5} chooses $t_i$ as small as possible for $k,y$. Then, the algorithm accepts in line \ref{alg3:8} if and only if there exists some $b$ with $c(k,y,b) \in O \cap C^{i(k,y)}_{k,y}$. By Claim \ref{claim:cw-trust}, if $y$ gets accepted, then $y \in \ran(F_k^O)$. By Claim \ref{claim:G-has-short-proofs}, if $y \in \ran(F_k^O)$, then the algorithm accepts $y$ on at least one computation path. Hence, the algorithm decides $\ran(F_k^O)$.
\end{claimproof}
\begin{claim}\label{claim:algo1-time}
There is $c \in \N$ such that each computation path of algorithm \ref{alg3} takes $O(h^c)$ time.
\end{claim}
\begin{claimproof}
The lines \ref{alg3:2} and \ref{alg3:3} can be processed in polynomial time, because there are only finitely many $y$ with $|y| \leq m$ and $i(k,y) < 21$. So we can use a precoded table of constant size for step \ref{alg3:3}. Line \ref{alg3:5} can be computed in polynomial time in $|y|$ by $T \in \P$ (cf.~Claim \ref{claim:stages2}.\ref{claim:stages2-i}). Similarly, from $S \in \P$, line \ref{alg3:6} can be computed in polynomial time in $s_i$, so this step takes $\leq (s_i/2)^{c'}$ steps for a constant $c'$. Since $|y| > r(k,\ell) > r(k) \geq k+2$ (cf.~R\ref{R2}) and $i(k,y) \geq 21$, we can invoke Claim \ref{claim:smallest-cw-length-y} and get $h(|y|) \geq s_{i(k,y)}/2 = s_i/2$. So line \ref{alg3:6} takes $\leq h(|y|)^{c'}$ steps. The lines \ref{alg3:7} and \ref{alg3:8} can be processed in $(s_i/2)^{c''} \leq h(|y|)^{c''}$ steps for a constant $c''$, once $t_i$ and $s_i$ are known. In total, a computation path takes $O(h^{c})$ time for a constant $c$. 
\end{claimproof}
\begin{claim}\label{claim:inUtime}
There is $c \in \N$ such that $\ran(F_k^O) \in \UTIME^O(h^c)$.
\end{claim}
\begin{claimproof}
The Claims \ref{claim:algo1-accepts} and \ref{claim:algo1-time} show that there is a non-deterministic algorithm that decides $\ran(F_k^O)$ in time $O(h^c)$ for some $c \in \N$, giving $\ran(F_k^O) \in \NTIME^O(h^c)$. Consider any input $y$. If $y$ is accepted or rejected in line \ref{alg3:3}, we only have one computation path. In line \ref{alg3:8}, the algorithm accepts for every $b \in \Sigma^{t_{i+\floor{\log i}}}$ with $c(k,y,b) \in O$. Let $s$ be the first step with $\ls{w_s} \geq s_i$. Then V\ref{V3}($w_s$) gives that there is at most one $b \in \Sigma^{t_{i+\floor{\log i}}}$ with $c(k,y,b) \in w_s$. Since the stage $s_i$ is completely defined for $w_s$ and $O \sqsupseteq w_s$, this also holds for $O$ and hence, there is at most one computation path that accepts at line \ref{alg3:8}. Thus, $\ran(F_k^O) \in \UTIME^O(h^c)$.
\end{claimproof}
Summarized, either $F_k^O$ is not optimal for $\ran(F_k^O)$ and (\ref{thm:thick-oracle-2-i}) is satisfied or $F_k^O$ is optimal for $\ran(F_k^O)$ and (\ref{thm:thick-oracle-2-ii}) is satisfied by Claim \ref{claim:inUtime}.
\end{proof}
\begin{theorem}\label{thm:NQP}
There exists an oracle $O$ such that all $L \in \RE^O \setminus \NQP^O$ have no optimal proof system relative to $O$.
\end{theorem}
\begin{proof}
Choose $O$ according to Theorem \ref{thm:thick-oracle-2} and $L \in \RE^O \setminus \NQP^O$ arbitrary. Let $f \in \FP$ be an arbitary proof system for $L$. Since $\{F_k^O\}_{k \in \N^+}$ is a standard enumeration of $\FP^O$-machines, there is some $k \in \N^+$ such that $F_k^O(x) = f(x)$ for all $x \in \Sigma ^*$. Observe that for any constant $c \in \N$ it holds that
\[h^c = \left(2^{(\log n)^{40}}\right)^c = 2^{c \cdot (\log n)^{40}},\]
so $\NQP^O \supseteq \NTIME^O(h^c) \supseteq \UTIME^O(h^c)$. Since $\ran(F_k^O) \notin \NQP^O \supseteq \UTIME^O(h^c)$ for any constant $c$, by Theorem \ref{thm:thick-oracle-2}, $F_k^O$ is not optimal for $\ran(F_k^O)=L$.
\end{proof}
\begin{corollary}\label{cor:thick}
There exists an oracle $O$ relative to which all of the following holds:
\begin{enumerate}
\item\label{cor:thick-1} All $L \in \RE \setminus \NQP$ have no optimal proof system.
\item\label{cor:thick-2} All $L \in \RE \setminus \NQP$ have no p-optimal proof system.
\item\label{cor:thick-3} $\NP \neq \coNP$.
\item\label{cor:thick-4} $\NP$ has measure $0$ in $\E$ (i.e., the Measure Hypothesis \cite{lut97} is false).
\end{enumerate}
\end{corollary}
\begin{proof}
Choose $O$ like in Theorem \ref{thm:NQP}.
\medskip
\\
To \ref{cor:thick-1} and \ref{cor:thick-2}: Follows from Theorem \ref{thm:NQP} and that p-optimality implies optimality.
\medskip
\\
To \ref{cor:thick-3}: Messner \cite[Cor.~3.40]{mes00} shows relativizably that if $\NP = \coNP$, then for every time constructible function $t$, there is a set $L \in \RE \setminus \NTIME(t)$ that has an optimal proof system. Theorem~\ref{thm:NQP} shows that relative to $O$, if $\NTIME(t) \supseteq \NQP$, then all sets in $\RE \setminus \NTIME(t)$ have no optimal proof system. Hence $\NP \neq \coNP$. 
\medskip
\\
To \ref{cor:thick-4}: Chen and Flum \cite[Thm.~7.4]{cfm14} show relativizably that if the Measure Hypothesis holds, then there exist problems in $\NE \setminus \NP$ with optimal proof systems. They proof this statement by showing that there is a $Q \in \NE$ which is $\NTIME(2^n)$-immune and has an optimal proof system. Since $Q$ is infinite and $\NTIME(2^n)$-immune, $Q \notin \NTIME(2^n)$. Hence, Chen and Flum even show that if the Measure Hypothesis holds, then there exist problems in $\NE \setminus \NTIME(2^n) \subseteq \NE \setminus \NQP$ with optimal proof systems. Since Theorem \ref{thm:NQP} shows that relative to $O$, $\RE \setminus \NQP \supseteq \NE \setminus \NQP$ has no sets with optimal proof system, the Measure Hypothesis is false.
\end{proof}

\section{Conclusion and Open Questions}

We construct two oracles that give evidence for the inherent difficulty
of finding complex sets with optimal proof systems.
This explains why the search for these sets has so far been unsuccessful.

There is potential to improve Theorem~\ref{thm:thick-oracle-2}.\ref{thm:thick-oracle-2-ii}. We choose $t = n^{\log^* n}$ as a slowly growing superpolynomial function, which determines how fast the distance between consecutive stages grows and subsequently determines the runtime above which optimal proof systems do not exist. Different choices of $t$ or tighter estimates could improve this runtime bound. To what extent can this improve our results? More precisely, is there an oracle relative to which there are no optimal proof systems for sets outside $\NTIME(h')$ where
\begin{itemize}
    \item $h'$ grows slower than $h$ from section \ref{sec:thick-oracle}?
    \item $h'$ is an arbitrary superpolynomial function?
\end{itemize} 
Although we believe that our techniques can be used to tackle the questions above, it seems that they are not sufficient to construct an oracle relative to which no set outside $\NP$ has optimal proof systems, because our construction heavily relies on the fact that the distance between consecutive stages grows faster than any polynomial. 

In this paper we focused on optimal proof systems.
But what about similar questions for p-optimal proof systems?
In particular, we are interested in the following questions: Is there an oracle relative to which
\begin{itemize}
\item $\NP \setminus \P$ has no sets with p-optimal proof systems and $\P \neq \NP$?
\item $\RE \setminus \QP$ has no sets with p-optimal proof system?
\end{itemize}

Moreover, it would be interesting to find evidence for 
the difficulty
of questions about other variants of proof systems.
For example, the difficulty to find optimal heuristic proof systems
for complex sets \cite{hims12}.

\bibliography{Literatursammlung}

\end{document}